\documentclass[11pt]{article}
\usepackage[letterpaper,margin=1in]{geometry}
\usepackage[T1]{fontenc}
\usepackage{lmodern,microtype}
\usepackage{amsmath,amssymb,amsthm,mathtools,bm}
\usepackage{booktabs,longtable,array,enumitem}
\usepackage{xcolor,tikz,graphicx}
\usetikzlibrary{arrows.meta,positioning,fit,calc,decorations.pathreplacing}
\usepackage[round,authoryear]{natbib}
\usepackage{placeins,needspace,float}
\usepackage[colorlinks=true,linkcolor=blue!55!black,citecolor=blue!55!black,urlcolor=blue!55!black]{hyperref}
\hypersetup{pdftitle={Arbitrarily Slow Polynomial Convergence of Fictitious Play},pdfsubject={v25 review copy; phase-height figure and overview refinements}}
\definecolor{Parent}{HTML}{24566F}
\definecolor{Gate}{HTML}{9B622C}
\definecolor{Control}{HTML}{536945}
\definecolor{Ink}{HTML}{263443}
\definecolor{Accent}{HTML}{684895}
\setlist{topsep=4pt,itemsep=3pt,parsep=0pt,leftmargin=*}
\newtheorem{theorem}{Theorem}[section]
\newtheorem{lemma}[theorem]{Lemma}
\newtheorem{proposition}[theorem]{Proposition}
\newtheorem{corollary}[theorem]{Corollary}
\theoremstyle{definition}\newtheorem{definition}[theorem]{Definition}
\theoremstyle{remark}\newtheorem{remark}[theorem]{Remark}
\newcommand{\chivec}{\bm{\chi}}
\newcommand{\zetavec}{\bm{\zeta}}
\newcommand{\hvec}{\mathbf h}
\newcommand{\N}{\mathbb N}\newcommand{\erow}[1]{e_{0,#1}^{\top}}
\newcommand{\Q}{\mathbb Q}\newcommand{\Z}{\mathbb Z}\newcommand{\R}{\mathbb R}
\newcommand{\one}{\mathbf1}\newcommand{\col}{\operatorname{col}}
\newcommand{\Gap}{\operatorname{Gap}}\newcommand{\rank}{\operatorname{rank}}
\newcommand{\supp}{\operatorname{supp}}
\newcommand{\ip}[2]{\langle #1,#2\rangle}\newcommand{\W}{\mathsf W}
\newcommand{\Aw}{\mathsf A}\newcommand{\Bw}{\mathsf B}

\newcommand{\LL}{\mathcal L}\newcommand{\RR}{\mathcal R}
\newcommand{\norm}[1]{\lVert #1\rVert}
\newcommand{\DeltaSet}{\Delta}
\newcommand{\reset}[1]{\mathsf a_{#1}}

\graphicspath{{./}}
\tikzset{
  box/.style={draw,rounded corners=2pt,align=center,minimum height=11mm,inner sep=4pt},
  parbox/.style={box,draw=Parent,fill=Parent!6},
  gatbox/.style={box,draw=Gate,fill=Gate!8},
  conbox/.style={box,draw=Control,fill=Control!7},
  arr/.style={-{Latex[length=2.2mm]},semithick}
}

\title{\textbf{Arbitrarily Slow Polynomial Convergence\\of Fictitious Play}}
\author{\begin{tabular}{@{}c@{\hspace{5mm}}c@{\hspace{5mm}}c@{}}
\normalsize Jacob Abernethy & \normalsize John Lazarsfeld & \normalsize Andre Wibisono\\[3pt]
\small Google Research & \small School of Computer Science & \small Department of Computer Science\\
\small \& School of Computer Science & \small Georgia Institute of Technology & \small Yale University\\
\small Georgia Institute of Technology & & \\
\small\texttt{abernethyj@google.com} & \small\texttt{jlazarsfeld@gatech.edu} & \small\texttt{andre.wibisono@yale.edu}
\end{tabular}}
\date{\today}
\makeatletter
\renewcommand{\maketitle}{%
  \begingroup\centering
  {\LARGE\@title\par}\vspace{1em}%
  \ifx\@author\@empty\else{\large\@author\par}\vspace{.6em}\fi
  \endgroup
}
\makeatother

\begin{document}

\pagenumbering{roman}
\thispagestyle{empty}
\thispdfpagelabel{Title}

\maketitle
\begin{center}\small\today\end{center}

\begin{abstract}
  We show that fictitious play can converge 
  at arbitrarily slow polynomial rates in two-player zero-sum games. 
  For every integer $k \ge 2$, we construct a 
  payoff matrix with $(k+1)^2 - 5$ actions per player
  for which the duality gap of the empirical strategies
  decays as $\Theta(t^{-1/k})$ after $t$ steps.
  The family starts from the standard
  rock-paper-scissors matrix, with each 
  higher-order game constructed recursively from 
  the preceding one. 
  After a prescribed common initial action, 
  every subsequent best response under fictitious play 
  is unique.
  For $k\ge3$, these games give counterexamples to Karlin's conjectured
  $O(t^{-1/2})$ convergence rate, and they extend the recent
  $\Theta(t^{-1/3})$ construction of \citet*{wang2025tie}
  to arbitrarily slow polynomial rates.
\end{abstract}

\vspace*{2em}
{\setlength{\parskip}{-0.5em}\tableofcontents}
\newpage
\pagenumbering{arabic}

\section{Introduction}
\label{sec:intro}

Fictitious play is one of the earliest 
and most natural dynamics for
computing a Nash equilibrium in two-player zero-sum games.
Introduced by~\citet*{brown1951iterative},
fictitious play is an iterative method in which, at each time step,
each player best responds to the empirical distribution
of their opponent's past actions. 
Shortly after its introduction,~\citet*{robinson1951iterative}
proved that the empirical strategies 
generated by the algorithm 
converge to a game's Nash equilibrium set,
albeit with a polynomially slow convergence rate
upper bound of $O(t^{-1/(m+n-2)})$ on $m$-by-$n$
payoff matrices after $t$ steps~\citep*{shapiro1958note}.
On the other hand, Karlin 
conjectured in 1959 that fictitious play 
converges in zero-sum games at the substantially faster
rate of  $O(t^{-1/2})$
\citep*{karlin1959mathematical}.

Despite the simplicity of the algorithm,
settling the true convergence rate of 
fictitious play has remained a surprisingly difficult
open question over the past several decades.
While empirical results have long suggested 
that Karlin's conjectured rate should hold,
a seminal result of~\citet*{daskalakis2014counter}
disproved the \textit{strong} variant 
of the conjecture.
Allowing \textit{adversarial tie-breaking},
Daskalakis and Pan proved that 
fictitious play can converge at the slow rate of
$\Theta(t^{-1/n})$ on 
the $n$-by-$n$ identity payoff matrix.
While the dimension-dependence of 
this rate effectively matches 
that in the upper bound of Robinson, the analysis of
Daskalakis and Pan relies heavily on 
an adaptive, adversarial tie-breaking rule.
In contrast, 
more recent works 
have suggested that Karlin's conjecture
may still hold in its weaker form
when the tie-breaking rule is fixed.
Indeed,~\citet*{abernethy2021fast} established the faster $O(t^{-1/2})$ rate
for positive \textit{diagonal} payoff matrices when each player uses
a fixed tie-breaking rule.
For generalized \textit{rock-paper-scissors} (RPS) matrices,
\citet*{lazarsfeld25a} established the same rate under symmetric initialization
and any common tie-breaking rule.

However, a recent breakthrough result of~\citet*{wang2025tie}
disproved Karlin's conjecture in its weaker form.
Wang constructed a $10$-by-$10$ payoff matrix
for which the duality gap under fictitious play scales as
$\Theta(t^{-1/3})$.
Importantly, aside from a prescribed common initial action,
the construction involves no best-response ties.
The key to Wang's construction is to couple
three blocks of rock-paper-scissors within a larger
skew-symmetric game so as to generate the cubic slowdown.
The natural question left open by Wang's result
is whether the $\Theta(t^{-1/3})$ rate is a fundamental limit
in the absence of best-response ties, 
or whether fictitious play can converge
at still slower polynomial rates.

The main contribution of this paper is to show
that the latter is possible.
For every integer $k \ge 2$, we recursively construct
a payoff matrix with $(k+1)^2-5$ actions per player
for which the duality gap under fictitious play decays as
$\Theta(t^{-1/k})$.
After a prescribed common initial action,
every subsequent best response is unique.
Thus, even in the absence of best-response ties,
fictitious play can converge
at an arbitrarily slow polynomial rate. 

We start by recalling the zero-sum game 
setting and the fictitious play update rule
before formally stating our main result. 

\subsection{Fictitious play and duality gap}
\label{sec:fp}

\paragraph{Notation.}
For an integer $m \ge 1$, 
let $\one_m\in\R^m$ denote the all-ones vector.
For an action $i$, let $e_i \in\R^m$ denote the $i$th standard
basis vector. 
Let $\DeltaSet_m:=\{p\in\R_{\ge0}^m:\one_m^{\top}p=1\}$ 
denote the probability simplex on $m$ actions.
We write $\Q$, $\R$, and $\Z$ for the rational numbers, real numbers, and integers,
respectively, and $\N=\{1,2,\ldots\}$ for the positive integers.

\paragraph{Zero-sum games, Nash equilibrium, and duality gap.}
Let $A\in\R^{m\times n}$ be a payoff matrix. 
We study the two-player zero-sum game 
with the bilinear payoff
\[
 \max_{p\in\DeltaSet_m}\min_{q\in\DeltaSet_n}p^{\top}Aq,
\]
in which the row player seeks to maximize the payoff and the column player seeks to minimize it. 
In this setting, von Neumann's minimax theorem \citep*{v1928theorie} guarantees the existence 
of a \textit{Nash equilibrium} 
$(p^\star,  q^\star) \in \DeltaSet_m \times \DeltaSet_n$ 
satisfying
\begin{equation}
    p^\top A q^\star 
    \le 
    (p^\star)^\top A q^\star 
    \le 
    (p^\star)^\top A q
    \qquad
    \text{$\forall \, p \in \DeltaSet_m$, $q \in \DeltaSet_n$.}
    \label{eq:NE}
\end{equation}
For any strategy pair $(p,q) \in \DeltaSet_m \times \DeltaSet_n$, 
a standard measure of proximity to Nash equilibrium
is the \emph{duality gap}
\begin{equation}
  \label{eq:gapdef}
  \Gap_A(p,q)=\max_i(Aq)_i-\min_j(A^{\top}p)_j.
\end{equation}
It is well-known that $\Gap_A(p, q) \ge 0$
for all $(p, q) \in \DeltaSet_m \times \DeltaSet_n$.
Moreover, $\Gap_A(p, q) = 0$ 
if and only if $(p, q)$ is a Nash equilibrium 
satisfying~\eqref{eq:NE}.

\paragraph{Fictitious play update.}
We study the \emph{simultaneous} discrete-time fictitious-play update rule. Let $x_t\in\Z_{\ge0}^m$ and $y_t\in\Z_{\ge0}^n$ be the two cumulative \emph{count vectors} after $t$ completed rounds, with $x_0=y_0=0$. At times $t \ge 0$, the players choose
actions $i_t, j_t$ and update their count vectors as
\begin{equation}\label{eq:fp}
\begin{aligned}
  i_t&\in\arg\max_i(Ay_t)_i,
  \quad
  &x_{t+1}=x_t+e_{i_t},\\
  j_t&\in\arg\min_j(A^{\top}x_t)_j,
  \quad
  &y_{t+1}=y_t+e_{j_t}.
\end{aligned}
\end{equation}
For $t \ge 1$, 
the row score $(Ay_t)_i$ is the cumulative payoff 
that row action $i$ would have earned against all past column 
plays through round $t-1$.
Similarly, $(A^{\top}x_t)_j$ is the cumulative payoff 
to the row player against column action $j$ through round $t-1$.
After $t$ steps, the empirical strategies of the 
players are given by $x_t/t$ and $y_t/t$.

\paragraph{Skew-symmetry and common actions.}

When $m=n$, a payoff matrix $A$ is \textit{skew-symmetric}
if $A = -A^\top$. 
In this case, if $x_t = y_t = z_t$ for some $t \ge 0$,
then the two best-response sets in~\eqref{eq:fp}
coincide. If the players choose the same action from
this set, then their count vectors remain 
equal at time $t+1$, meaning $x_{t+1} = y_{t+1} = z_{t+1}$.
For \(t\ge1\), synchronization gives
\begin{equation}\label{eq:symmetricgap}
  \Gap_A(z_t/t,z_t/t)=\frac{2}{t}\max_i(Az_t)_i.
\end{equation}
Thus, unique common best responses preserve 
synchronization 
without any additional tie-breaking assumptions.

\Needspace{16\baselineskip}
\subsection{Main result}\label{sec:main}

Our main result gives a family
of zero-sum games in which
fictitious play converges
at progressively slower 
polynomial rates. 
In each game, 
apart from a prescribed common first action, every subsequent 
best response under fictitious play 
is unique.
Formally, we prove the following.

\begin{theorem}[Arbitrarily slow polynomial convergence]\label{thm:main}
For every integer $k\ge2$, there exists a skew-symmetric matrix $\widehat A_k\in\Q^{n_k\times n_k}$, where
\[
  n_k=(k+1)^2-5,
\]
with one of these $n_k$ actions designated as the setup action. Start simultaneous fictitious play from zero counts, and let both players choose the designated setup action in round $0$. At every decision time $t\ge1$, both players have the same unique best response. Their common count vector $z_t$ satisfies
\begin{equation}\label{eq:mainrate}
  \Gap_{\widehat A_k}(z_t/t,z_t/t)
     =\Theta(t^{-1/k})
\end{equation}
for each fixed $k$, with constants depending on the fixed game.
\end{theorem}

\begin{remark}[Payoff normalization]\label{rem:normalization}
Multiplying the constructed matrix by $1/M_*$, where $M_*:=\max_{i,j}|(\widehat A_k)_{ij}|>0$, places its entries in $[-1,1]$. This multiplies every score comparison and the gap by the same positive factor, preserving the action sequence and exponent.
\end{remark}

\vspace*{0.25em}
\begin{remark}[Initialization]\label{rem:initialization}
At $t=0$, all scores are zero and every action is a best response. The theorem prescribes the same setup action $\omega$ for both players; its payoff column produces the initial score vector required by the construction. Wang uses an analogous prescribed setup action, played only once at the beginning, to realize initial scores from zero counts~\citep*[Section 4.2 and Lemma 1]{wang2025tie}. Equivalently, our count vectors may start at the unit vector $e_\omega$; after $r$ additional rounds the empirical denominator is $r+1$.
\end{remark}

\paragraph{No universal polynomial exponent.}
As a consequence of Theorem~\ref{thm:main}, there is no constant $\gamma>0$
such that the duality gap is bounded by $O(t^{-\gamma})$ for every finite
zero-sum game and every simultaneous fictitious-play trajectory with a
prescribed common first action and unique best responses thereafter, even
with a game- and initialization-dependent constant.
This follows by taking $k>1/\gamma$ in Theorem~\ref{thm:main}.

\vspace*{-0.5em}
\paragraph{Rate and action count.}
At order $k=2$, Theorem~\ref{thm:main} is realized by the three-action RPS game augmented with a single setup action, for four actions in total. For $k\ge3$, the gap exponent $1/k$ is strictly smaller than $1/2$, yielding convergence slower than the $O(t^{-1/2})$ rate in Karlin's conjecture. The cubic obstruction was already established by \citet*{wang2025tie}; the rates with $k\ge4$ extend it to arbitrarily small positive exponents. Table~\ref{tab:dimensions} lists the first seven members of the family. Since $n_k+5=(k+1)^2$, the exponent can be written exactly as
\[
 \alpha_{n_k}:=\frac1k=\frac{1}{\sqrt{n_k+5}-1}.
\]
Thus, the construction gives an infinite family of dimensions $n=n_k$ with
\[
 \Gap_{\widehat A_k}(z_t/t,z_t/t)
 =\Theta(t^{-\alpha_{n_k}}),\qquad
 \alpha_{n_k}=\Theta(n_k^{-1/2}).
\]
In this dimension-dependent sense the rate is $t^{-\Theta(1/\sqrt n)}$. The constants and onset may depend on $k$; the dimension statement concerns the sequence $n=n_k$.

\begin{table}[htbp]\centering
\begin{tabular}{@{}lrrrrrrr@{}}\toprule
Order $k$ & $2$ & $3$ & $4$ & $5$ & $6$ & $7$ & $8$\\
Actions per player $n_k$ & $4$ & $11$ & $20$ & $31$ & $44$ & $59$ & $76$\\
Gap rate & $t^{-1/2}$ & $t^{-1/3}$ & $t^{-1/4}$ & $t^{-1/5}$ & $t^{-1/6}$ & $t^{-1/7}$ & $t^{-1/8}$\\\bottomrule
\end{tabular}
\caption{Games with unique best responses after setup. Each rate holds over all integer times up to positive constant factors. The action count includes the setup action.}\label{tab:dimensions}
\end{table}

The quadratic action count follows from retaining a single parent block and adding \(2k+3\) actions in the lift from order \(k\) to order \(k+1\). Including the final setup action gives \(n_k=(k+1)^2-5\).

\medskip
\begin{remark}[A smaller construction]\label{rem:seven-base}
We do not claim that the action count is optimal. An alternative explicit \(7\times7\) game achieves gap \(\Theta(t^{-1/3})\), improving on the \(10\times10\) construction of \citet*{wang2025tie}. Its six recurrent actions can serve as the \(k=3\) seed for the same lift. Adding one final setup action gives \(\tilde n_k=(k+1)^2-9\) actions and gap \(\Theta(t^{-1/k})\) for \(k\ge3\), saving four actions while retaining quadratic growth. 
We use the RPS seed in the main construction because its argument is simpler.
\end{remark}

\paragraph{Formal verification.}
The RPS-based recursive construction in Theorem~\ref{thm:main}, including its seed, lift, and realization as simultaneous fictitious play starting from zero counts with a prescribed common setup action, has been formalized in Lean~4. The checked proofs establish unique best responses after setup and a positive exact leading constant for the empirical duality gap as time tends to infinity through all integers, with relative error $O(t^{-1/k})$. The formal proofs use only Lean's standard axioms. Its computational claims require verified realizers, supplied for this family, and do not assert practical extraction performance. 
The sources, pinned build instructions, and theorem correspondence are available at \url{https://github.com/awibisono/fictitious-play-experiments/tree/main/lean}.

\vspace*{-0.5em}
\paragraph{Open questions.}
Our lower bound construction leaves open two natural questions.
First, what is the optimal dependence on dimension of the
worst-case convergence rate under fixed tie-breaking?
In particular, can a family of
$n\times n$ games exhibit
the slower rate of
$t^{-\Theta(1/n)}$?
Such a construction would match, up to constant factors
in the exponent, both the lower bound under adversarial
tie-breaking of~\citet*{daskalakis2014counter}
and the general upper bound of~\citet*{robinson1951iterative}.
Second, which structural properties ensure Karlin's
conjectured $O(t^{-1/2})$ rate?
This rate holds for positive diagonal and generalized
rock-paper-scissors matrices under the tie-breaking and
initialization assumptions discussed
above~\citep*{abernethy2021fast,lazarsfeld25a},
but a general characterization remains open.

\paragraph{AI methodology.}
We developed the family of hard constructions and their convergence proofs through an iterative process of mathematical analysis and interaction with ChatGPT and Gemini. 
The authors verified and simplified the construction and proofs.
ChatGPT assisted with writing
and developing the Lean formalization of our results.
The authors take full responsibility for the results and their exposition. 
We provide a more detailed description of our AI usage in 
Appendix~\ref{sec:ai-methodology}.

%%%%%%%%%%%%%%%
\subsection{Related work}

\paragraph{Lower bounds for fictitious play in zero-sum games.}
\citet*{daskalakis2014counter} construct trajectories on the $n\times n$ identity payoff matrix $I_n$ that obtain a duality gap of $\Theta(t^{-1/n})$. Their proof is inductive, and in the inductive step from $I_{n-1}$ to $I_n$, the row player repeats one action while the column player raises its counts on the other actions to a common level. Both players then follow an $I_{n-1}$ trajectory on those other actions until the row-player counts are equal again. A padding segment then adjusts the length of the embedded trajectory to reach this endpoint. The
construction repeatedly selects among tied best responses,
and thus this result disproves the strong version of
Karlin's conjecture that permits arbitrary 
time-dependent tie-breaking.

\citet*{wang2025tie} constructs a $10\times10$ skew-symmetric 
payoff matrix whose duality gap scales as $\Theta(t^{-1/3})$,
and with no best-response ties after a prescribed common first action.
The construction has nine recurrent actions 
arranged in three coupled RPS blocks.
A tenth action, played only in the first round by both players, 
produces a prescribed starting score, including a common
 shift and small offsets that remove all
subsequent best-response ties. 
This construction cycles between the three RPS blocks. Our construction uses the same idea of long shifted segments followed by relative-score restoration at a common height. The new step is a lift retaining a single parent block while preserving all conditions needed to repeat the construction (Theorem~\ref{thm:lift}). A direct recursion retaining a fixed number \(c>1\) of parent copies and adding \(O(k)\) new actions per lift would have \(\Theta(c^k)\) actions. 
Instead, retaining one parent copy and adding linearly many actions in the induction order, as in our construction,
yields a total action count that is quadratic in that order.

\vspace*{-0.5em}
\paragraph{Upper bounds for fictitious play.}
For positive diagonal matrices, \citet*{abernethy2021fast} 
proved an $O(t^{-1/2})$ convergence rate when each player resolves ties
by a fixed total priority ordering (equivalently, lexicographical tie-breaking
under the player's fixed ordering).
Note that this result also holds for the identity 
matrix used in the lower bound of~\cite{daskalakis2014counter}.
For weighted generalized RPS games, 
\citet*{lazarsfeld25a} proved a similar $O(t^{-1/2})$ convergence 
rate upper bound under symmetric initialization and using 
any common tie-breaking rule.

\vspace*{-0.5em}
\paragraph{Fictitious play beyond zero-sum games.}
Fictitious play has also been extensively studied in 
settings beyond two-player zero-sum games,
but in general, its empirical strategies do not
always converge to a game's Nash equilibrium set;
see~\citet*[Section 3]{mertikopoulos2026regret}
and the references therein for a recent overview.
Several recent works have also proven 
lower bounds on the convergence rate of
fictitious play in \textit{potential games}~\citep*{brandt2010rate, panageas2023exponential,anagnostides2026doubly}.

More broadly, prior works have studied continuous-time variants of Fictitious Play (see, e.g.,
\citet*{harris1998rate, hofbauer2002global, viossat2013no, 
ostrovski2014payoff,swenson2017exponential}).
In general online learning settings,
several canonical families of algorithms like
\textit{Follow-the-Regularized-Leader} (FTRL)
and \textit{Follow-the-Perturbed-Leader} (FTPL)
can be viewed as smoothed or regularized variants
of Fictitious Play, which itself corresponds 
to the unregularized \textit{Follow-the-Leader} (FTL) algorithm
(see, e.g,~\citet*{kalai2005efficient, hazan2016introduction, shalev2025online}).
Characterizing the convergence (or lack thereof) of such regularized
algorithms to equilibria in games has been the subject of 
extensive work in the past decade 
(see, e.g.,~\citet*{syrgkanis2015fast, 
    mertikopoulos2018cycles, bailey2018multiplicative, 
    daskalakis2021near, soleymani2025faster}).

\paragraph{Organization of the paper.}
Section~\ref{sec:trialoverview} gives an overview of the construction.
Section~\ref{sec:technical} develops the technical framework, and Section~\ref{sec:clock} shows how the construction yields the main result.
Section~\ref{sec:seed} establishes the rock-paper-scissors base case, 
and Sections~\ref{sec:gateway} and~\ref{sec:onestepproof} present and prove the inductive step.
The appendices provide explicit bounds, example matrices, and numerical illustrations.

\newcommand{\trialref}[1]{\ref{#1}}
\section{Overview of the construction}\label{sec:trialoverview}

Our recursive step turns a game with gap of order $t^{-1/k}$ into a larger game with gap of order $t^{-1/(k+1)}$. We keep the parent payoff matrix $A$ unchanged and add one endpoint action $g$ and two groups of $k+1$ reset actions, as in Figure~\ref{fig:trial-overview-matrix}. The new payoffs let fictitious play follow long segments of the parent trajectory and then restore the relative scores needed to run another such segment. This makes the maximum cumulative payoff grow as $\Theta(t^{k/(k+1)})$, so dividing by elapsed time gives the slower decay of the gap. All entries of the enlarged matrix are fixed before the trajectory begins. Starting from rock-paper-scissors (RPS), we repeat this construction, preserving the properties that allow each enlarged game to serve as the next parent.

The recursive step has four components. First, we bundle consecutive parent phases into longer phases. Second, an endpoint action becomes the next best response only when the bundle ends. Third, reset actions restore the relative parent scores needed for the next bundle at a higher common height. We choose their payoffs so that fictitious play leaves them inactive during the bundle and selects them in the required order afterward. Fourth, we move phase boundaries to give each phase a fixed starting action and a final endpoint play, and verify that the construction can be repeated. Finally, we convert this auxiliary trajectory into a fictitious play trajectory from zero counts, with unique best responses after the prescribed common first play.

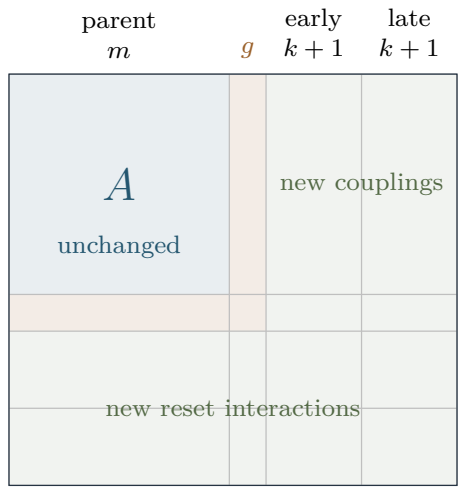
\begin{figure}[H]\centering
\resizebox{.36\linewidth}{!}{%
\begin{tikzpicture}[x=.82cm,y=.82cm]
\fill[Parent!10] (0,2.6) rectangle (3,5.6);
\fill[Gate!12] (3,0) rectangle (3.5,5.6);
\fill[Gate!12] (0,2.1) rectangle (6.1,2.6);
\fill[Control!8] (3.5,0) rectangle (6.1,5.6);
\fill[Control!8] (0,0) rectangle (6.1,2.1);
\foreach \x in {0,3,3.5,4.8,6.1}{\draw[gray!50] (\x,0)--(\x,5.6);}
\foreach \y in {0,1.05,2.1,2.6,5.6}{\draw[gray!50] (0,\y)--(6.1,\y);}
\draw[Ink] (0,0) rectangle (6.1,5.6);
\node[text=Parent,font=\Large] at (1.5,4.1) {$A$};
\node[text=Parent,font=\scriptsize] at (1.5,3.25) {unchanged};
\node[above,font=\scriptsize,align=center] at (1.5,5.65) {parent\\$m$};
\node[above,font=\scriptsize,text=Gate] at (3.25,5.65) {$g$};
\node[above,font=\scriptsize,align=center] at (4.15,5.65) {early\\$k+1$};
\node[above,font=\scriptsize,align=center] at (5.45,5.65) {late\\$k+1$};
\node[font=\scriptsize,text=Control] at (4.8,4.1) {new couplings};
\node[font=\scriptsize,text=Control] at (3.05,1.05) {new reset interactions};
\end{tikzpicture}%
}
\caption{One unchanged parent block, one endpoint action $g$, and two groups of $k+1$ reset actions. The new cross-block entries are chosen together, subject to skew symmetry. The figure shows a schematic block partition, not a numerical matrix.
See Figure~\trialref{fig:matrix} for a more detailed schematic.}\label{fig:trial-overview-matrix}
\end{figure}

\needspace{7\baselineskip}
\paragraph{Scores and phases.}
To develop the construction, we first allow a rational initial score vector $f$ and a fixed priority order for breaking ties. This priority order is only an auxiliary device, and the final construction requires no tie-breaking after the prescribed common first play. If $z_t$ counts the $t$ completed plays, then
\[
 u_t=f+Az_t,\qquad H_t=\max_i u_{t,i}.
\]
Thus, an action's score is its initial bias plus its cumulative payoff against the completed plays, and the \emph{score height} $H_t$ is the largest of these scores. The next action maximizes the current score, with ties resolved by the priority order. Playing action $i$ adds column $i$ of $A$ to $u_t$. Since $A$ is skew-symmetric ($A^\top=-A$), the chosen action's own score stays fixed, so $H_t$ cannot decrease. For synchronized fictitious play from zero counts and without an initial bias, both players choose the same action each round, and the empirical duality gap is $2\max_i(Az_t)_i/t$, as in~\eqref{eq:symmetricgap}. The setup action introduced at the end will generate the initial scores through actual play, allowing us to use this identity for the final game.

We divide the trajectory into finite consecutive blocks called \emph{phases} so that we can track the score height and number of plays at their boundaries. A \emph{phase word} lists the actions in one such block, and playing the word means executing those actions in their listed order. A \emph{run} is a consecutive repetition of one action.

An order-$k$ \emph{clock} records such a trajectory together with polynomial formulas for its entry scores and phase action counts. Its phase index $q$ is distinct from elapsed time $t$, which counts individual plays. The entry height and phase duration are polynomials of degree $k-1$ with positive leading coefficients. Summing the phase durations gives elapsed time of order $q^k$, and hence
\[
 \text{height}=\Theta(q^{k-1}),\qquad t=\Theta(q^k),\qquad \frac{\text{height}}{t}=\Theta(q^{-1})=\Theta(t^{-1/k}).
\]
The ratio shows why increasing the clock order slows the decay of the gap. We start with an order-$2$ clock in RPS, whose phase $q$ plays actions $0,1,2$ in runs of lengths $3q+5$, $3q+6$, and $3q+7$, respectively, and raises the entry height by one. Its linear phase durations accumulate to quadratic elapsed time, giving the square-root rate. Sections~\ref{sec:auxiliaryoverview} and~\ref{sec:rpsintuition} give the score process and RPS data, while Definition~\ref{def:clock} states the full conditions that the recursion preserves.

We now describe the four components that turn an order-$k$ clock into an order-$(k+1)$ clock.

\paragraph{1. Bundling parent trajectories.}
The first task is to form longer outer phases by revisiting overlapping ranges of parent phases at successively higher score heights. The added reset actions make this schedule possible by restoring the relative parent scores needed at the start of each next bundle.

The first lift from order $k=2$ to order $k=3$ makes the idea concrete. For this illustration, we leave aside the fixed integer scale introduced below. In outer phase $q$, the construction runs the $q$ RPS phases indexed by $q,q+1,\ldots,2q-1$. This \emph{parent bundle} is their concatenated sequence of actions. It lasts $\Theta(q^2)$ plays and raises the parent height by $q$. Its exit has the relative-score pattern for RPS phase $2q$, whereas the next bundle must start with the pattern for phase $q+1$. The added reset actions restore that pattern at a higher common height. Only the differences among the parent scores are restored, while elapsed time and cumulative counts continue to increase. A common shift of all parent scores leaves every comparison among parent actions unchanged. In the auxiliary process, the reset runs at a height a fixed amount $c\in(0,1)$ above the parent maximum at the bundle endpoint. The maximum score remains constant throughout the reset, so outer phase $q$ gains $q+c$ in total. We will choose a reset of length $\Theta(q^2)$, so the linear height gains accumulate to height $\Theta(q^2)$ and the quadratic outer-phase lengths accumulate to elapsed time $\Theta(q^3)$. Their ratio gives the rate $t^{-1/3}$.

\needspace{6\baselineskip}
For the general lift, let $\W_k(\sigma)$ be the parent phase word at index $\sigma$. We choose an integer scale $\tau\ge1$ once for the lift so that the reset run lengths can be made integral. In outer phase $q$, we play the bundle
\[
 \W_k(\tau q)\W_k(\tau q+1)\cdots\W_k(2\tau q-1).
\]
This bundle contains $\tau q$ parent phases and lasts $\Theta(q^k)$ plays. If $\chi(\sigma)$ is the parent's entry height, the gain across the bundle is $\chi(2\tau q)-\chi(\tau q)=\Theta(q^{k-1})$. The next bundle requires the relative-score pattern at $\tau(q+1)$ rather than the pattern at $2\tau q$ where the current bundle ends.

Let $\zeta(q)$ be the height at which the reset actions play. We choose the new payoffs and reset schedule to enforce the following height relation, which adds the parent bundle's height gain and a fixed increment $c$ for the transition to the reset:
\[
 \zeta(q)-\zeta(q-1)=\chi(2\tau q)-\chi(\tau q)+c,
 \qquad 0<c<1.
\]
Thus the parent bundle starts at height $\zeta(q-1)$ and reaches height $\zeta(q)-c$ when it ends. The endpoint action then brings play to the reset actions at height $\zeta(q)$. We choose their plays to restore the parent pattern for phase $\tau(q+1)$ with every parent score shifted by the same amount so that its maximum is $\zeta(q)$. This common shift preserves the comparisons among parent actions, allowing the next bundle to follow the same parent trajectory from the required index. Successive bundles therefore revisit overlapping ranges of parent phases at increasingly large heights, even as time and counts continue to increase. Summing the height increments gives $\zeta(q)=\Theta(q^k)$. The reset construction provides another $\Theta(q^k)$ plays per outer phase, so outer-phase lengths have degree $k$ and accumulate to elapsed time of degree $k+1$. Figure~\ref{fig:overview-phase-height} summarizes the height and duration calculation across successive outer phases, before we move the phase boundaries in Component~4.
\begin{figure}[H]
\centering
\resizebox{.98\linewidth}{!}{\begin{tikzpicture}[x=.85cm,y=1cm,>=Latex,font=\small]
\definecolor{parentcol}{RGB}{57,91,111}
\definecolor{resetcol}{RGB}{148,111,67}
\tikzset{height/.style={line width=1.05pt},guide/.style={gray!45,densely dashed,line width=.35pt}}
% Illustrative coordinates only: neither an exact path nor a numerical sample.
% Aligned chronological bands.
\foreach \a/\b in {0/1.65,3.05/5.30,7.25/10.10}{
 \fill[parentcol!7] (\a,1.35) rectangle (\b,5.15);
 \fill[parentcol!18] (\a,.25) rectangle (\b,.73);
}
\foreach \a/\b in {1.8/3.05,5.45/7.25,10.25/12.6}{
 \fill[resetcol!8] (\a,1.35) rectangle (\b,5.15);
 \fill[resetcol!23] (\a,.25) rectangle (\b,.73);
}
\draw[->,black!65] (0,1.35)--(13,1.35) node[right,font=\small] {$t$};
\draw[->,black!65] (0,1.35)--(0,5.35);
\foreach \y/\lab in {3.55/{\zeta(q+1)},4.95/{\zeta(q+2)}}{
 \draw[gray!60] (-.05,\y)--(.05,\y);
 \node[left,font=\small] at (-.10,\y) {$\lab$};
}
\draw[<->,black!65,line width=.45pt] (-.16,1.65)--(-.16,2.45);
\node[left,align=center,font=\small] at (-.24,2.05) {height gain\\$\Theta(q^{k-1})$};
\draw[guide] (0,2.45)--(3.05,2.45);
\draw[guide] (0,3.55)--(7.25,3.55);
\draw[guide] (0,4.95)--(12.6,4.95);
% Parent-bundle schematic: increasing entry heights, not actual within-phase scores.
\draw[height,parentcol] (0,1.65)--(.45,1.65)--(.45,1.85)--(.95,1.85)--(.95,2.05)--(1.35,2.05)--(1.35,2.25)--(1.65,2.25);
\draw[height,black] (1.65,2.25)--(1.80,2.45);
\draw[height,resetcol] (1.80,2.45)--(3.05,2.45);
\draw[height,parentcol] (3.05,2.45)--(3.55,2.45)--(3.55,2.675)--(4.10,2.675)--(4.10,2.90)--(4.65,2.90)--(4.65,3.125)--(5.05,3.125)--(5.05,3.35)--(5.30,3.35);
\draw[height,black] (5.30,3.35)--(5.45,3.55);
\draw[height,resetcol] (5.45,3.55)--(7.25,3.55);
\draw[height,parentcol] (7.25,3.55)--(7.8,3.55)--(7.8,3.79)--(8.35,3.79)--(8.35,4.03)--(8.90,4.03)--(8.90,4.27)--(9.45,4.27)--(9.45,4.51)--(9.85,4.51)--(9.85,4.75)--(10.10,4.75);
\draw[height,black] (10.10,4.75)--(10.25,4.95);
\draw[height,resetcol] (10.25,4.95)--(12.6,4.95);

\draw[black!60,line width=.4pt] (12.6,4.95)--(13.05,4.95);
\node[anchor=west,font=\small,align=left] at (13.10,4.95) {total height after\\phase $q+2$\\$\zeta(q+2)=\Theta((q+2)^k)$};
\node[font=\small,anchor=west,fill=white,inner sep=2pt] (hand) at (.85,4.1) {endpoint handoff};
\draw[->,black!60,line width=.4pt] (hand.south)--(1.74,2.49);
\foreach \a/\b/\lab in {0/3.05/{\Theta(q^k)},3.05/7.25/{\Theta((q+1)^k)},7.25/12.6/{\Theta((q+2)^k)}}{
 \draw[<->,black!60,line width=.4pt] ({\a+.12},1.16)--({\b-.12},1.16);
 \node[above,font=\small,fill=white,inner sep=1.5pt] at ({(\a+\b)/2},1.16) {$\lab$};
}
% A thin endpoint slot, not an extended plateau.
\foreach \a/\b in {1.65/1.8,5.30/5.45,10.10/10.25}{\fill[black!65] (\a,.25) rectangle (\b,.73);}
\foreach \x in {.825,4.175,8.675}{\node[font=\small] at (\x,.49) {parent};}
\foreach \x in {2.425,6.35,11.425}{\node[font=\small] at (\x,.49) {reset};}
\foreach \a/\b/\lab in {0/3.05/q,3.05/7.25/{q+1},7.25/12.6/{q+2}}{
 \draw[black!55,line width=.4pt] (\a,.11)--(\a,-.02)--(\b,-.02)--(\b,.11);
 \node[below,font=\small] at ({(\a+\b)/2},-.04) {outer phase $\lab$};
}
\end{tikzpicture}
}
\caption[Growing outer phases in the auxiliary construction.]{Schematic height and duration of outer phases, each comprising a parent bundle, one endpoint play, and a reset block, before the boundary rotation in Component~4. 
The parent height rises from $\zeta(q-1)$ to $\zeta(q)-c$. The black connector summarizes the handoff to the reset ceiling $\zeta(q)$. 
Phase $q$ gains $\Theta(q^{k-1})$ in height, with $\Theta(q^k)$ plays in each of the parent bundle and reset block. 
Thus height $\Theta(q^k)$ over elapsed time $T_q=\Theta(q^{k+1})$ gives a duality gap of order $\Theta(T_q^{-1/(k+1)})$. 
The steps and widths are illustrative, 
and the flat reset ceilings belong to the auxiliary process. Removing ties preserves the action sequence, not exact flatness.
}
\label{fig:overview-phase-height}
\end{figure}
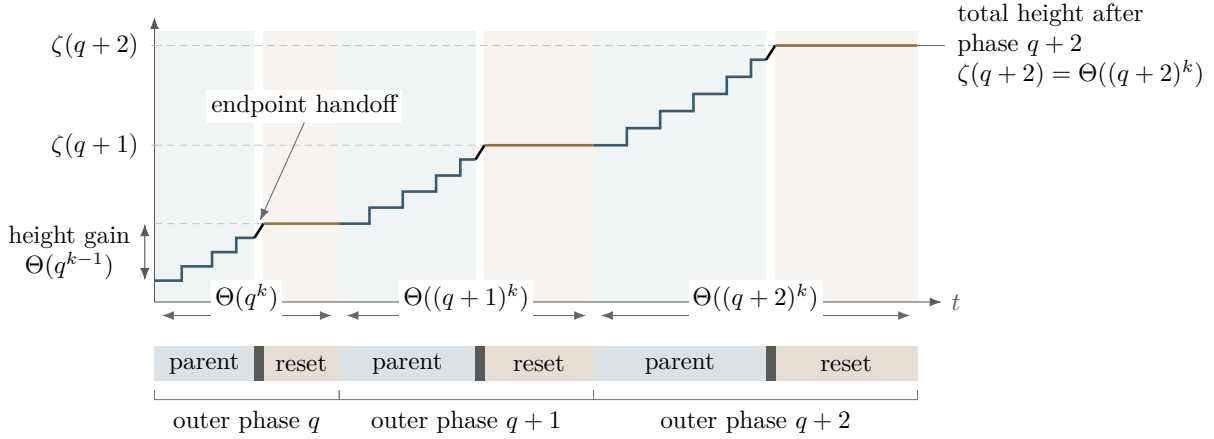

Section~\ref{sec:overviewblocks} specifies the bundle and shifted scores, and Lemma~\ref{lem:clockgrowth} and Corollary~\ref{cor:liftgrowth} establish these growth formulas. What remains is to choose fixed payoffs that keep the added actions from exceeding the parent maximum during the bundle and then make best responses perform the restoration.

\paragraph{2. Endpoint detection.}
We want the endpoint action $g$ to remain below the parent maximum throughout the bundle and become the unique next maximizing action when the bundle ends. Its play will then make a reset action the next best response, providing the \emph{handoff} from the parent actions to the reset block.

Write $u_i$ for action $i$'s current score, $H_{\rm parent}$ for the maximum among the parent scores, and $b$ for the parent's \emph{entry leader}, a designated action attaining the maximum at every phase entry, though not necessarily during the phase. Let $z^{\rm rem}$ count the parent plays still remaining in the bundle, one coordinate for each parent action. The parent clock supplies a fixed rational row $v$ assigning total weight one to each complete parent phase. Thus $vz^{\rm rem}$ counts remaining phases at phase boundaries, although within a phase it can fluctuate and is not generally a monotone countdown. We combine this count with the difference between the parent maximum and the score of $b$. The clock guarantees a fixed margin $\delta_0>0$ after every \emph{proper prefix}, meaning an initial segment strictly before the end of the bundle, including the empty segment before any play. At each such decision time,
\[
 H_{\rm parent}-u_b+vz^{\rm rem}\ge\delta_0.
\]
At the endpoint both terms vanish because no plays remain and $b$ again attains the parent maximum.

We make $g$ follow the score of $b$ while accounting for the remaining plays by setting its payoff row against the parent to $A_{b,:}+v$. Copying row $b$ reproduces its score changes, and adding $v$ makes the difference $u_g-u_b$ increase by $v_i$ when parent action $i$ is played. We initialize $g$ so that
\[
 u_g-u_b=\delta-vz^{\rm rem},\qquad 0<\delta<\min\{c,\delta_0\}.
\]
Since the same play changes $vz^{\rm rem}$ by $-v_i$, this identity persists throughout the bundle. At every proper prefix it gives
\[
 H_{\rm parent}-u_g
 =H_{\rm parent}-u_b+vz^{\rm rem}-\delta
 \ge\delta_0-\delta>0.
\]
At the endpoint, $u_g=H_{\rm parent}+\delta$. Figure~\ref{fig:trial-overview-detector} illustrates how the detector stays below the parent maximum before completion and overtakes it by $\delta$ at completion. In designing the reset, we also keep all reset actions below $g$ at this endpoint, making $g$ the unique next action. Its single play brings the first reset action to the ceiling $\zeta(q)$.

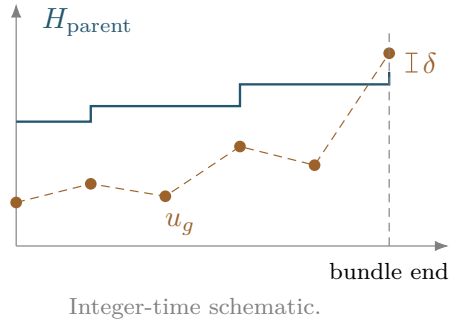
\begin{figure}[H]\centering
\resizebox{.36\linewidth}{!}{%
\begin{tikzpicture}[x=.9cm,y=.75cm,>=Latex]
\draw[->,gray] (0,0)--(5.8,0);
\draw[->,gray] (0,0)--(0,3.9);
\draw[Parent,thick] (0,2)--(1,2)--(1,2.25)--(3,2.25)--(3,2.6)--(5,2.6)--(5,2.8);
\foreach \x/\y in {0/.7,1/1.0,2/.8,3/1.6,4/1.3,5/3.1}{\fill[Gate] (\x,\y) circle(2pt);}
\draw[Gate,densely dashed] (0,.7)--(1,1)--(2,.8)--(3,1.6)--(4,1.3)--(5,3.1);
\draw[gray,densely dashed] (5,0)--(5,3.5);
\node[text=Parent,font=\small,anchor=west] at (.2,3.6) {$H_{\rm parent}$};
\node[text=Gate,font=\small] at (2.2,.35) {$u_g$};
\draw[Gate] (5.3,2.8)--(5.3,3.1);\draw[Gate] (5.22,2.8)--(5.38,2.8);\draw[Gate] (5.22,3.1)--(5.38,3.1);
\node[right,font=\small,text=Gate] at (5.3,2.95) {$\delta$};
\node[below,font=\scriptsize,align=center] at (5,-.1) {bundle end};
\node[below,font=\scriptsize,text=gray,align=center] at (2.4,-.65) {Integer-time schematic.};
\end{tikzpicture}%
}
\caption{The endpoint detector waits below the parent maximum at every proper prefix and overtakes it at the bundle endpoint. The figure is schematic. See Figure~\trialref{fig:score-chronology}(a) for a more detailed schematic of endpoint detection.}
\label{fig:trial-overview-detector}
% The plotted heights are an integer-time schematic, not numerical data; $u_g$ need not be monotone.}\label{fig:trial-overview-detector}
\end{figure}

For RPS, we use $v=(0,-1,1)$ because each phase has one more play of action $2$ than of action $1$. In later lifts, each parent phase ends with exactly one play of its endpoint action, and $v$ selects that coordinate. Section~\ref{sec:overviewendpoint} and Lemma~\ref{lem:gateway} give the detector identity and comparisons, while Proposition~\ref{prop:chronology} verifies the handoff against all the added actions.

\paragraph{3. Reset actions and score restoration.}
The reset must now restore the scores required by the next bundle, while ensuring that the fictitious-play best-response update rule selects its actions in the prescribed order. Its payoffs must also keep these actions inactive during the parent bundle just described.

Call the parent actions together with $g$ the \emph{old actions}, and write their matrix as $G_O$. We write the enlarged skew-symmetric payoff matrix as
\[
 G=\begin{pmatrix}G_O&P\\-P^\top&J-J^\top\end{pmatrix},
 \qquad P=[P_E\mid P_L],\qquad P_E<0,\quad P_L>0,
\]
Here $P$ records old--reset payoffs, with $P_E$ and $P_L$ corresponding to the early and late groups, respectively. The inequalities are entrywise. Playing an early reset action adds a negative column to every old score, whereas playing a late reset action adds a positive column. Lowering the old scores first keeps them below the reset ceiling. Raising them afterward toward their target values from below prevents them from taking over before restoration is complete.

The strictly lower-triangular matrix $J$ has positive entries below its diagonal. Hence a reset action leaves its own score fixed, lowers earlier reset scores, and raises later reset scores. These signs determine the directions of the score changes but not when an action becomes maximal. We choose the couplings and run lengths so that the next reset action reaches the ceiling at the end of the current run, while later reset actions remain below it until their own turn. The reset actions then take turns at $\zeta(q)$ in their listed order. The priority order favors the later reset action at each handoff and favors parent actions when the reset ends. Figure~\ref{fig:trial-overview-reset} shows the resulting decrease and recovery of the old scores. The ceiling is flat in this auxiliary process. We will preserve the action sequence when removing ties, without requiring the perturbed maximum to remain exactly flat.

\begin{figure}[H]\centering
\resizebox{.45\linewidth}{!}{%
\begin{tikzpicture}[x=.82cm,y=.8cm,>=Latex]
\draw[gray,->] (0,0)--(9.6,0);
\foreach \a/\b/\lab in {0/1.5/{a_0},1.5/3/{\cdots},3/4.5/{a_{r-1}},4.5/6/{a_r},6/7.5/{\cdots},7.5/9/{a_{2r-1}}}{
 \draw[Gate,very thick] (\a,4.2)--(\b,4.2);
 \draw[Gate] (\a,4.1)--(\a,4.3);
 \node[above,font=\scriptsize] at ({(\a+\b)/2},4.3) {$\lab$};}
\node[above,text=Gate,font=\small] at (4.5,4.95) {reset actions take turns at $\zeta(q)$};
\draw[Parent,very thick] (0,3.2)--(4.5,1.0);
\draw[Control,very thick] (4.5,1.0)--(9,4.2);
\draw[gray!50,densely dashed] (4.5,0)--(4.5,4.2);
\node[below,font=\small,align=center] at (2.25,0) {early reset\\old scores decrease};
\node[below,font=\small,align=center] at (6.75,0) {late reset\\old scores increase};
\node[align=center,font=\small] at (4.5,2.6) {old-score envelope\\(schematic)};
\fill[Control] (9,4.2) circle(2pt);
\node[below,font=\scriptsize,align=right,text=Control] at (8,2.05) {$b$ reaches\\the next entry};
\end{tikzpicture}%
}
\caption{With $r=k+1$, early reset plays lower every old score, while late reset plays raise the old scores toward the next entry state. Reset actions take turns at the common ceiling $\zeta(q)$. This schematic describes the auxiliary process before the fixed bias removes ties. See Figure~\trialref{fig:score-chronology}(b) for a more detailed schematic of the reset and restoration of the next entry state.}\label{fig:trial-overview-reset}
\end{figure}
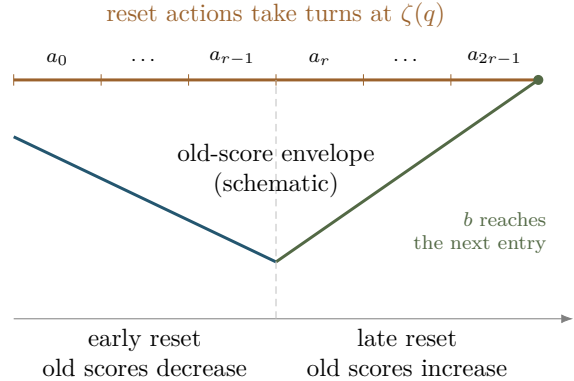

Skew-symmetry of the enlarged matrix $G$ requires the reverse coupling to be $-P^\top$, so choosing how reset plays change old scores also determines how old plays change reset scores. The same couplings that make early reset plays lower the old scores make parent plays raise the early reset scores, while late reset scores move in the opposite direction. The late reset scores start at or below the parent maximum. The parent leader wins the possible entry tie by the priority order, and subsequent parent plays lower the late reset scores while the parent maximum does not decrease.

Keeping the early reset actions below the parent requires more than these signs. At any decision time during the bundle, an early reset action's current score is its endpoint score minus the payoff it will receive from the remaining parent plays. We choose its couplings so that this remaining payoff is at least the remaining increase in the parent maximum, apart from a fixed additive error. If the reset action's endpoint score is below the endpoint parent maximum by more than this error, its current score is also below the current parent maximum. This explains how an endpoint deficit keeps the action inactive throughout the bundle.

We therefore cannot choose the reset payoffs solely to restore the old scores. The same entries must prepare the reset scores during the preceding bundle, and the two tasks constrain each other. To impose both requirements for every sufficiently large $q$, we express entry scores and action counts as polynomials in $q$. A \emph{score table} or \emph{count table} is a fixed matrix whose row for each action lists the coefficients of its score or count polynomial in a common polynomial basis. Matching finitely many coefficients then enforces the score identities at every such phase index without changing the payoffs from phase to phase.

\needspace{6\baselineskip}
Let $X$ and $M$ be the coefficient tables for old-action counts during the bundle and endpoint play, and for reset counts, respectively. Let $D$ be the coefficient table of the required old-score changes during reset, and let $-Q$ be that of the required reset-score changes during those old plays. The fixed coupling matrix $P$ must satisfy both construction requirements
\[
 PM=D,\qquad P^\top X=Q.
\]
The first equation requires reset plays to restore the old scores. The second requires old plays to prepare the reset scores, whose actual change is $-P^\top X=-Q$ because the reverse coupling is $-P^\top$. These equations constrain the same entries of $P$ and cannot be solved independently. For example, multiplying the first by $X^\top$ and using the second gives the necessary compatibility identity $X^\top D=Q^\top M$. We choose the reset count polynomials and new payoffs together so that the compatibility, positivity, and sign requirements hold simultaneously. The space of polynomials of degree at most $k$ has dimension $k+1$, which motivates using $k+1$ reset actions in each group with count polynomials that form a basis for this space. Their values give positive integer run lengths once the phase index is sufficiently large.

Section~\ref{sec:polynomialdata} introduces the tables, Lemma~\ref{lem:interpolate} characterizes compatibility, and Sections~\ref{sec:moment}--\ref{sec:synthesis} construct the solution. The score identities establish restoration (Lemma~\ref{lem:resetidentity}), while the separate best-response comparisons in Proposition~\ref{prop:chronology} establish that the prescribed runs really occur.

\paragraph{4. Phase boundaries and recursion.}
To use the enlarged game as a parent, we first need phases that begin with a fixed maximizing action and have an endpoint that the next lift can detect. We obtain these features by changing where phases begin, without changing any play.

Let $\Aw(q)$ denote the parent bundle followed by the single play of $g$, and let $\Bw(q)$ denote the reset block. The trajectory we have constructed has the form
\[
 \cdots\Aw(q)\Bw(q)\Aw(q+1)\Bw(q+1)\Aw(q+2)\cdots.
\]
Instead of placing the boundary before each parent bundle, we place it before each reset block. After a finite initial prefix, the same sequence is grouped as
\[
 \cdots\underbrace{\Bw(q)\Aw(q+1)}_{\W_{k+1}(q)}
 \underbrace{\Bw(q+1)\Aw(q+2)}_{\W_{k+1}(q+1)}\cdots,
 \qquad \W_{k+1}(q)=\Bw(q)\Aw(q+1).
\]
Each new phase begins with the first reset action at the maximum. Its final play of $g$ brings that same reset action back to the maximum, ready for the next phase. Moreover, $g$ is played exactly once, at the end of the new phase. The row selecting its coordinate therefore counts one remaining play at every proper prefix and zero at completion. Together with the return of the entry leader to the maximum, this gives the endpoint comparison needed for the next lift. Figure~\ref{fig:trial-overview-rotation} shows both groupings.

To complete the induction, the enlarged game must retain the polynomial score and count structure, the comparisons needed for endpoint detection, and the remaining-height bound used to keep reset actions inactive. We must also preserve the algebraic conditions that make the reset equations solvable with the required signs at the next lift. The clock conditions are stated formally in Definition~\ref{def:clock}. Section~\ref{sec:phaseboundary} explains the regrouping, and Section~\ref{sec:next} and Proposition~\ref{prop:inherit} verify these conditions and the algebraic data needed for another lift. Theorem~\ref{thm:lift} summarizes the recursive step.

\begin{figure}[H]\centering
\resizebox{.72\linewidth}{!}{%
\begin{tikzpicture}[x=1cm,y=1cm]
\draw[gray] (0,1.3)--(0,1.55)--(6.7,1.55)--(6.7,1.3);
\draw[gray] (7,1.3)--(7,1.55)--(13.7,1.55)--(13.7,1.3);
\node[above,font=\scriptsize,text=gray] at (3.35,1.55) {outer phase $q$};
\node[above,font=\scriptsize,text=gray] at (10.35,1.55) {outer phase $q+1$};
\foreach \x/\lab/\col in {0/{\mathsf A(q)}/Parent,3.5/{\mathsf B(q)}/Control,7/{\mathsf A(q+1)}/Parent,10.5/{\mathsf B(q+1)}/Control}{
\draw[draw=\col,fill=\col!6] (\x,0) rectangle +(3.2,1.0);
\node at (\x+1.6,.5) {$\lab$};}
\draw[Gate,very thick] (3.5,-.22)--(3.5,-.6)--(10.2,-.6)--(10.2,-.22);
\node[text=Gate,below] at (6.85,-.65) {$\W_{k+1}(q)=\mathsf B(q)\mathsf A(q+1)$};
\end{tikzpicture}%
}
\caption{The same infinite word can be grouped as parent bundle plus endpoint followed by reset, or as reset followed by the next bundle plus endpoint. The latter grouping supplies the next clock with a fixed entry leader and a single final endpoint play. See Figure~\trialref{fig:rotation} for a more detailed schematic of the phase-boundary rotation.}\label{fig:trial-overview-rotation}
\end{figure}
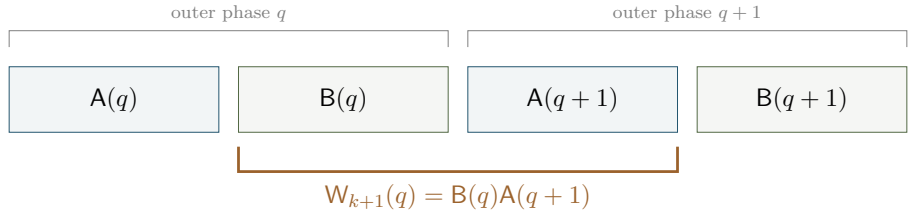

\paragraph{From the auxiliary clock to fictitious play.}
Finally, we choose fixed score bonuses to implement the auxiliary priority order without ties. Because the payoff matrix and initial scores are rational, all auxiliary score differences lie on a fixed rational lattice, so every nonzero difference is separated from zero by a common positive amount. We choose distinct bonuses in an interval shorter than this spacing, giving larger bonuses to higher-priority actions. These resolve each tie in the prescribed order without reversing any strict comparison, at any time. We then add one setup action whose first play produces the required initial scores, including the bonuses, up to a common shift. Starting from zero counts, we prescribe this common first play for both players. All subsequent best responses are unique, and the setup action is never used again (Lemmas~\ref{lem:bias} and~\ref{lem:setup}).

To pass from phase boundaries to every integer time, consider a time between two consecutive boundaries. Monotonicity places its score height between the corresponding entry heights, while its elapsed time lies between the two boundary times. The ratios of consecutive entry heights and consecutive boundary times tend to one, so the upper and lower bounds give the same leading asymptotic for height divided by time. Fixed bonuses, a common shift, and the single setup play do not change that leading term. We therefore obtain the gap asymptotic at every integer time with the same positive leading constant (Lemma~\ref{lem:alltime}).

Let $n_k$ denote the number of actions available to each player in the final order-$k$ game, including the setup action. We start with three recurrent RPS actions, and each lift from order $j$ to order $j+1$ adds one endpoint action and two groups of $j+1$ reset actions, for a total of $2j+3$ new actions. Adding the final setup action gives

\[
 n_k=3+\sum_{j=2}^{k-1}(2j+3)+1=(k+1)^2-5.
\]
Section~\ref{sec:mainproof} completes the proof of the main theorem. The precise score states and polynomial data behind the construction are given in Section~\ref{sec:technical}.

\FloatBarrier

\section{Technical preliminaries for the construction}\label{sec:technical}
We now give the exact score states, polynomial tables, and comparisons behind the overview.

\subsection{Scores and phases in the auxiliary process}\label{sec:auxiliaryoverview}
\paragraph{Counts, scores, and state.}
To describe the recursive construction, fix a skew-symmetric matrix $A\in\Q^{m\times m}$, a rational initial score vector $f\in\Q^m$, and a total priority order $\prec$ on the actions. Earlier in $\prec$ means higher priority. Let $z_t\in\Z_{\ge0}^m$ count the $t$ plays made after initialization at score vector $f$, with $z_0=0$. The current \emph{score state} is $\mathbf u_t=f+Az_t$. We use the scalar aliases $u_{t,i}:=(\mathbf u_t)_i$ and, when time is suppressed, $u_i:=(\mathbf u)_i$ for the score of action $i$. The auxiliary process is
\begin{equation}\label{eq:process}
\begin{aligned}
 \mathbf u_0&=f,\qquad z_0=0,\\
 i_t&=\text{the highest-priority maximizer of }\mathbf u_t,\\
 z_{t+1}&=z_t+e_{i_t},\\
 \mathbf u_{t+1}&=\mathbf u_t+Ae_{i_t}.
\end{aligned}
\end{equation}
Thus, one play adds one payoff column to the score vector and one unit to a count. Recall from Section~\ref{sec:fp} 
that skew-symmetry, equal counts, and a common priority order synchronize the two players. Here $f$ need not be generated by previous plays; Lemmas~\ref{lem:bias} and~\ref{lem:setup} later realize this sequence from zero counts with unique best responses. The priority order is an auxiliary construction device; after setup, the trajectory in Theorem~\ref{thm:main} needs no tie-breaking.

The \emph{height} is the largest current score:
\[
 H_t:=\max_{0\le i<m}(\mathbf u_t)_i=\max_{0\le i<m}u_{t,i}.
\] Since $A_{ii}=0$, playing a maximizing action leaves its own score unchanged, and hence
\begin{equation}\label{eq:monotone}
 H_{t+1}\ge H_t\qquad(t\ge0).
\end{equation}
Adding a common constant to all initial scores does not change any best response. In local comparisons within a phase (for example, in Section~\ref{sec:chronology}), we write $u_i$ for the current scalar score of action $i$; we retain $\mathbf u_t$ whenever elapsed time matters.

\paragraph{Words, phases, and time.}
A \emph{word} $w=i_1\cdots i_L$ is a sequence of $L$ action labels. A prefix consists of its first $h$ plays, where $0\le h\le L$; it is \emph{proper} if $h<L$, including the empty prefix $h=0$. The suffix is the remaining word. A \emph{run} $i^L$ repeats action $i$ for $L$ plays. A word is \emph{valid} from a score state when, after every proper prefix, its next action maximizes the score over every action in the game and wins any tie by priority. Actions absent from the word are still included in this comparison.

A \emph{phase} is a designated finite word in the infinite best-response sequence generated by~\eqref{eq:process}, which we denote by $\W_k(\sigma)$ for an order-$k$ game on $m$ actions (where $k\ge2$ is the degree parameter described below and $\sigma\in\N$ is the parent phase index). Its entry and exit are the score states immediately before and after that word. All general phase statements below apply to integers $\sigma\ge Q_k$, where $Q_k$ is a fixed onset threshold. Each such phase index $\sigma$ specifies a prescribed \emph{entry score state} $\mathbf r(\sigma)\in\Q^m$, whose coordinates are polynomials of degree at most $k-1$ in $\sigma$, and a valid phase word $\W_k(\sigma)$, whose action counts form a vector $\mathbf c(\sigma)\in\Z_{\ge0}^m$ of total length $|\W_k(\sigma)|=\sum_{i=0}^{m-1}\mathbf c_i(\sigma)$. Starting from the entry score vector $\mathbf u_t=\mathbf r(\sigma)$, playing the word $\W_k(\sigma)$ adds $\mathbf c(\sigma)$ to the action counts and updates the score vector to the phase-exit state
\[
 \mathbf r(\sigma)+A\mathbf c(\sigma)=\mathbf r(\sigma+1),
\]
which is the prescribed entry score state for phase $\sigma+1$. Elapsed time $t$ counts individual plays; in the final game with setup it also includes the initial setup play. Thus, a phase index $\sigma$ is a parameter that specifies a prescribed score state $\mathbf r(\sigma)$ and word $\W_k(\sigma)$, rather than elapsed time $t$ or the cumulative number of phases already played; this distinction is essential in the inductive lift below, where newly added actions will restore the relative scores from one phase's exit state to a different target entry state $\mathbf r(\sigma')$.

\paragraph{Constants in the growth estimates.}
In estimates as a phase index or time tends to infinity, the game and all construction parameters are fixed. Constants and thresholds in the asymptotic notations $O$, $o$, and $\Theta$ may depend on those fixed data. Section~\ref{sec:synthesis} separately gives estimates uniform in parameters that are still to be chosen.

Definition~\ref{def:clock} makes the clock structure in Section~\ref{sec:trialoverview} precise. Lemma~\ref{lem:clockgrowth} proves the polynomial growth statements, and Lemma~\ref{lem:alltime} converts them to an all-integer-time gap after bias and setup.

\subsection{Polynomial score and count data}\label{sec:polynomialdata}
A central challenge of the construction is that the payoff matrix $A\in\Q^{m\times m}$ must be fixed before the score process begins, yet it must guide the trajectory through an infinite sequence of phases $\sigma$ whose entry score vectors $\mathbf r(\sigma)\in\Q^m$ and phase action-count vectors $\mathbf c(\sigma)\in\Z_{\ge0}^m$ grow as polynomials in $\sigma$. To design a fixed matrix $A$ that satisfies trajectory conditions at every phase index simultaneously, we factor each vector of polynomials into a fixed matrix of coefficients—which we call a \emph{score table} or \emph{count table}—multiplying a single column vector of basis polynomials. This reduces infinite families of polynomial score equations to finite matrix identities.

\paragraph{The binomial basis.}
We use $\sigma$ for a parent phase index and $q$ for an outer phase index. The same polynomial-table notation applies to either index. A run $i^L$ consists of $L$ consecutive plays of action $i$, so its length $L$ must be a nonnegative integer. When we specify $L$ by a polynomial in $q$, we must therefore choose a polynomial that takes integer values at integer phase indices. Integer linear combinations of binomial polynomials ensure integrality; we will separately choose a threshold above which the values are positive. For a coefficient vector of length $s\ge1$, write the column vector of binomial basis polynomials of degree at most $s-1$ as
\[
 p_{s-1}(q)=\left(1,\,q,\,\binom q2,\,\ldots,\,\binom q{s-1}\right)^{\top},
 \qquad\text{where }\binom qj=\frac{q(q-1)\cdots(q-j+1)}{j!}.
\]
All coefficient indices start at $0$. Let $e_{0,s}\in\Q^s$ be the first standard basis column (index $0$), and write its transpose as $\erow{s}=(1,0,\ldots,0)$. Then $\erow{s}p_{s-1}(q)=1$. A zero-indexed row vector $a=(a_0,a_1,\ldots,a_{s-1})\in\Q^{1\times s}$ represents the polynomial $\phi(q) = a p_{s-1}(q)=\sum_{j=0}^{s-1}a_j\binom qj$ of degree at most $s-1$, whose leading $q^{s-1}$ coefficient is $a_{s-1}/(s-1)!$. Because $\binom qj\in\Z$ for every integer $q\in\Z$, any row $a \in \Z^{1 \times s}$ with integer coefficients in this basis automatically takes integer values at all integers $q \in \Z$.

\paragraph{Score and count tables for an order-$k$ clock.}
Stacking one coefficient row for each action produces a \emph{coefficient table}: a fixed matrix $Z\in\Q^{d\times s}$ whose $(i,j)$ entry is the coefficient of $\binom qj$ for action $i$. Multiplying the table $Z$ on the right by the basis vector $p_{s-1}(q)$ evaluates all $d$ polynomials at index $q$, giving the column vector
\[
 \mathsf{ev}_q(Z):=Zp_{s-1}(q)\in\Q^d,
 \qquad\text{whose }i\text{th coordinate is }\sum_{j=0}^{s-1}Z_{ij}\binom qj.
\]
We keep a fixed coefficient table $Z$ distinct from its evaluated vector $\mathsf{ev}_q(Z)$: for the score and count data below, uppercase letters denote the fixed tables and bold lowercase letters their evaluated columns. Payoff matrices and shift operators are also uppercase.
The bold row $\chivec$ denotes height coefficients, whereas $\chi(\sigma)=\chivec p_{k-1}(\sigma)$ denotes the scalar height. 
Indexed components retain the typeface of their parent vector or row, except for explicitly defined scalar aliases.

In particular, for an order-$k$ clock on $m$ actions (whose phase polynomials have degree at most $k-1$, corresponding to table width $s=k$), the phase data at index $\sigma$ are encoded by a fixed \emph{score table} $R\in\Q^{m\times k}$ and a fixed \emph{count table} $C\in\Z^{m\times k}$:
\begin{equation}\label{eq:parenttables}
\begin{aligned}
 \mathbf r(\sigma)&:=Rp_{k-1}(\sigma)\in\Q^m,
 &\mathbf c(\sigma)&:=Cp_{k-1}(\sigma)\in\Z^m,\\
 \chivec&:=R_{b,:}\in\Q^{1\times k},
 &\chi(\sigma)&:=\chivec p_{k-1}(\sigma)\in\Q.
\end{aligned}
\end{equation}
Here $R$ and $C$ are constant matrices independent of $\sigma$: $\mathbf r(\sigma)$ gives the prescribed scores at the start of phase $\sigma$, and $\mathbf c(\sigma)$ counts plays made during that single phase $\W_k(\sigma)$ (in contrast to $z_t$, which counts cumulative completed plays from time $0$). In our clocks, a single designated action $b$ attains the maximum score at the start of every phase $\sigma$ above the onset threshold, so the $b$th row of $R$, written $\chivec=(\chivec_0,\ldots,\chivec_{k-1})$, gives the binomial coefficients of the entry height $\chi(\sigma)$. (For a general matrix $Z$, the notation $Z_{:,j}$ denotes its $j$th column, $Z_{i,:}$ denotes its $i$th row, and $\col(X,Y)=\bigl(\begin{smallmatrix}X\\Y\end{smallmatrix}\bigr)$ denotes vertical concatenation of matrices with the same number of columns.)

\paragraph{Shift matrices and why phase transitions become matrix equations.}
A second advantage of the binomial basis is that shifting the phase index $q\mapsto q+1$ or taking the discrete difference $\phi(q+1)-\phi(q)$ corresponds to multiplying the coefficient table by a fixed integer matrix. Let $I_s$ be the $s\times s$ identity matrix, and let $N_s$ be the $s\times s$ first-subdiagonal shift matrix with $(N_s)_{i,i-1}=1$ for $1\le i<s$ and zeros elsewhere. Define $\mathsf{T}_s=I_s+N_s$ and $B_s=\mathsf{T}_s^{-1}$ as the forward and backward shift matrices. 
For example, at widths $s=2$ and $s=3$,
\[
 \mathsf{T}_2=\begin{pmatrix}1&0\\1&1\end{pmatrix},\quad
 B_2=\begin{pmatrix}1&0\\-1&1\end{pmatrix},\qquad
 \mathsf{T}_3=\begin{pmatrix}1&0&0\\1&1&0\\0&1&1\end{pmatrix},\quad
 B_3=\begin{pmatrix}1&0&0\\-1&1&0\\1&-1&1\end{pmatrix}.
\]
By Pascal's identity, $\binom{q+1}{j}=\binom qj+\binom q{j-1}$ (with $\binom q{-1}=0$), the forward shift, forward difference, and backward shift of the basis vector obey the exact polynomial identities
\begin{equation}\label{eq:translations}
\begin{aligned}
 p_{s-1}(q+1)&=\mathsf{T}_sp_{s-1}(q),\\
 p_{s-1}(q+1)-p_{s-1}(q)&=N_sp_{s-1}(q),\\
 p_{s-1}(q-1)&=B_sp_{s-1}(q).
\end{aligned}
\end{equation}
Multiplying a coefficient row or table on the right by $\mathsf{T}_s$, $B_s$, or $N_s$ therefore shifts the represented polynomials forward by one ($\mathsf{ev}_q(Z\mathsf{T}_s)=\mathsf{ev}_{q+1}(Z)$), backward by one ($\mathsf{ev}_q(ZB_s)=\mathsf{ev}_{q-1}(Z)$), or takes their forward difference ($\mathsf{ev}_q(ZN_s)=\mathsf{ev}_{q+1}(Z)-\mathsf{ev}_q(Z)$). For example, take the one-row table $Z=(2,3,4)$, so that $\mathsf{ev}_q(Z)=\phi(q)=2+3q+4\binom q2$. Right-multiplication gives
\[
 Z\mathsf{T}_3=(5,7,4),\qquad ZB_3=(3,-1,4),\qquad ZN_3=(3,4,0).
\]
Thus $\mathsf{ev}_q(Z\mathsf{T}_3)=\phi(q+1)$, $\mathsf{ev}_q(ZB_3)=\phi(q-1)$, and $\mathsf{ev}_q(ZN_3)=\phi(q+1)-\phi(q)=3+4q$.

Crucially,~\eqref{eq:translations} turns the phase-transition condition $\mathbf r(\sigma)+A\mathbf c(\sigma)=\mathbf r(\sigma+1)$ into a finite matrix equation between the tables $R$ and $C$. During phase $\sigma$, playing $\W_k(\sigma)$ changes the score vector by $A\mathbf c(\sigma)=ACp_{k-1}(\sigma)$, while the required difference between consecutive entry score states is $\mathbf r(\sigma+1)-\mathbf r(\sigma)=R\bigl(p_{k-1}(\sigma+1)-p_{k-1}(\sigma)\bigr)=RN_kp_{k-1}(\sigma)$. Thus,
\[
 \mathbf r(\sigma)+A\mathbf c(\sigma)=\mathbf r(\sigma+1)\quad\text{for all }\sigma
 \qquad\Longleftrightarrow\qquad
 AC=RN_k.
\]
In the same way, every score-restoration requirement in the inductive lift from order $k$ to order $k+1$ will be expressed as a finite matrix equation between coefficient tables of width $s=k+1$.

\subsection{Exact data for the RPS base clock}\label{sec:rpsintuition}
We now express the RPS clock from Section~\ref{sec:trialoverview} in the polynomial notation above. Proposition~\ref{prop:seed} verifies its trajectory and clock conditions, and Proposition~\ref{prop:seedadmissible} gives the additional algebraic data for the first lift.

\paragraph{One phase of the RPS base clock.}
The standard RPS payoff matrix on actions $\{0,1,2\}$ is
\begin{equation}\label{eq:rpsmatrix}
 A_{\mathrm{RPS}}=\begin{pmatrix}0&-1&1\\1&0&-1\\-1&1&0\end{pmatrix}.
\end{equation}
For each integer phase index $q\ge1$, we prescribe the entry score vector $\mathbf r_2(q)=(q+1,\,-2q-11/3,\,q+2/3)^\top$. The best-response process~\eqref{eq:process} starting from $\mathbf r_2(q)$ follows the exact three-run word
\[
 \W_2(q)=0^{3q+5}1^{3q+6}2^{3q+7}
\]
with a unique maximizer at every step. Each play raises the next action's score by one and lowers the remaining action's score by one. The fractional offsets make each switch strict. Over the three runs, the height moves from $q+1$ through $q+4/3$ and $q+5/3$ to $q+2$, and the final score vector is exactly $\mathbf r_2(q+1)$. In total, the phase takes $9q+18$ plays and raises the entry height by one. The run-by-run verification appears in Proposition~\ref{prop:seed}.

In the polynomial table notation of Section~\ref{sec:polynomialdata} with basis vector $p_1(q)=(1,q)^\top$, the RPS clock has width $k=2$ and is completely summarized by
\begin{equation}\label{eq:rpstables}
 R_2=\begin{pmatrix}1&1\\-11/3&-2\\2/3&1\end{pmatrix},\qquad
 C_2=\begin{pmatrix}5&3\\6&3\\7&3\end{pmatrix},\qquad
 \W_2(q)=0^{3q+5}1^{3q+6}2^{3q+7},
\end{equation}
Thus,
\[
 \begin{aligned}
 \mathbf r_2(q)&=R_2p_1(q)=\begin{pmatrix}q+1\\-2q-11/3\\q+2/3\end{pmatrix},
 \qquad\qquad
 \mathbf c_2(q)&=C_2p_1(q)=\begin{pmatrix}3q+5\\3q+6\\3q+7\end{pmatrix}.
 \end{aligned}
\] The phase-transition identity $AC=RN_k$ takes the concrete $3\times2$ form
\[
 A_{\mathrm{RPS}}C_2=R_2N_2=\begin{pmatrix}1&0\\-2&0\\1&0\end{pmatrix}.
\]
Multiplying both sides on the right by $p_1(q)=(1,q)^\top$ confirms algebraically that playing $\W_2(q)$ changes the score vector by $(1,-2,1)^{\top}=\mathbf r_2(q+1)-\mathbf r_2(q)$ for every $q$.

The first lift retains this $3\times3$ RPS block and adds one endpoint action and six reset actions, giving ten recurrent actions. One final setup action yields the $11\times11$ game displayed in Appendix~\ref{app:examples}. The next section describes the same construction at general order.

\subsection{Retaining the parent block and adding actions}\label{sec:overviewblocks}\label{sec:overviewgrowth}
Now consider the general inductive step from an order-$k$ parent clock ($k\ge2$) with payoff matrix $A\in\Q^{m\times m}$ on $m$ actions to an order-$(k+1)$ clock. Write $\W_k(\sigma)$ for the parent's phase word at parent index $\sigma$, with entry score vector $\mathbf r(\sigma)=Rp_{k-1}(\sigma)$, phase count vector $\mathbf c(\sigma)=Cp_{k-1}(\sigma)$, and entry height $\chi(\sigma)=\mathbf r_b(\sigma)$. All phase words and score comparisons below are considered above the onset thresholds specified by the construction.

An \emph{outer phase} is one phase of the enlarged, order-$(k+1)$ construction, indexed by $q$. For each induction step we fix a positive integer scale factor $\tau\ge1$, and in outer phase $q$ we run the \emph{parent bundle} consisting of $\tau q$ consecutive parent phases:
\[
 \W_k(\tau q)\W_k(\tau q+1)\cdots\W_k(2\tau q-1).
\]
For example, $\tau=2$ and $q=3$ runs the $\tau q=6$ parent phases with indices $6,7,8,9,10,11$, exiting at parent state $12$. The next outer phase $q+1$ will start its bundle at parent index $\tau(q+1)$ (index $8$ in this example), after the newly added actions have adjusted the scores from state $2\tau q$ back to state $\tau(q+1)$. Past action counts $z_t$ are never erased; only the score vector is moved to the next prescribed entry state. The first lift ($k=2\to3$) uses $\tau=1$; for later lifts ($k\ge3$), Lemma~\ref{lem:scale} chooses a fixed integer $\tau\ge1$ from the parent data so that the correction row $\ell_*$ used in the reset counts has integer binomial coefficients. Those coefficients make its values integral at every integer phase index. The bundle transfer matrices already have integer coefficients for every integer $\tau$.

\paragraph{Shifted parent scores, outer height $\zeta(q)$, and growth across the bundle.}
Let $\zeta(q-1)$ denote the maximum score at the start of outer phase $q$, and let $\zeta(q)$ denote the flat maximum score during the reset block of outer phase $q$. At the start of the bundle, the $m$ parent actions are initialized at the uniformly shifted score vector
\[
 \mathbf r(\tau q)+\bigl(\zeta(q-1)-\chi(\tau q)\bigr)\one_m\in\Q^m,
\]
whose maximum score is $\mathbf r_b(\tau q)+\zeta(q-1)-\chi(\tau q)=\zeta(q-1)$, attained by action $b$. Because adding a common scalar to all $m$ parent coordinates leaves all pairwise score differences (and hence internal best responses) unchanged, playing the bundle $\W_k(\tau q)\cdots\W_k(2\tau q-1)$ advances the parent score vector through the shifted states $\mathbf r(\sigma)+(\zeta(q-1)-\chi(\tau q))\one_m$ for $\sigma=\tau q,\ldots,2\tau q$, exiting the bundle at
\[
 \mathbf r(2\tau q)+\bigl(\zeta(q-1)-\chi(\tau q)\bigr)\one_m\in\Q^m,
\]
where the parent maximum has risen by the telescoping sum $\chi(2\tau q)-\chi(\tau q)$ (again attained by action $b$). During the subsequent handoff and reset block, the maximum score rises by an additional rational constant $c\in(0,1)$ (chosen in Section~\ref{sec:synthesis}) to reach the reset ceiling $\zeta(q)$. Thus, $\zeta(q)$ obeys the first-order difference equation
\begin{equation}\label{eq:overviewheight}
  \zeta(q)-\zeta(q-1)=\chi(2\tau q)-\chi(\tau q)+c,
\end{equation}
so the parent bundle starts at maximum score $\zeta(q-1)$ and exits at maximum score $\zeta(q)-c$. Starting from a chosen value $\zeta(0)$,~\eqref{eq:overviewheight} uniquely determines $\zeta(q)$ for all $q\ge1$ (we give its exact coefficient row in the binomial basis in Section~\ref{sec:gateway}).

Why does bundling $\tau q$ phases and resetting the scores from $2\tau q$ back to $\tau(q+1)$ raise the clock order from $k$ to $k+1$? Write the parent's phase length and entry height in ordinary powers of $\sigma$ as
\[
 |\W_k(\sigma)|=\mathsf L_k \sigma^{k-1}+O(\sigma^{k-2}),\qquad
 \chi(\sigma)=h_k \sigma^{k-1}+O(\sigma^{k-2}),\qquad \mathsf L_k,h_k>0.
\]
In a single parent phase $\sigma$, the height rises by only $\chi(\sigma+1)-\chi(\sigma)=\Theta(\sigma^{k-2})$. Across the $\tau q$ phases of the bundle, however, the phase lengths add to degree $k$ while the height increments telescope to $\chi(2\tau q)-\chi(\tau q)=\Theta(q^{k-1})$, so summing~\eqref{eq:overviewheight} over $q$ (Lemma~\ref{lem:clockgrowth}) makes both the bundle length and the outer height $\zeta(q)$ polynomials of degree $k$ with positive leading coefficients:
\begin{align*}
 \sum_{\sigma=\tau q}^{2\tau q-1}|\W_k(\sigma)|
   &=\frac{\mathsf L_k \tau^k(2^k-1)}{k}\,q^k+O(q^{k-1}),\\
 \zeta(q)&=\frac{h_k \tau^{k-1}(2^{k-1}-1)}{k}\,q^k+O(q^{k-1}).
\end{align*}
For example, in the first lift ($k=2\to3$, $\tau=1$), the RPS parent has $h_2=1$ and $\mathsf L_2=9$, giving $\zeta(q)-\zeta(q-1)=q+c$, $\zeta(q)=\tfrac12q^2+O(q)$, and bundle length $\tfrac{27}{2}q^2+O(q)$. However, the bundle exits at parent state $2\tau q$, whereas outer phase $q+1$ must start at parent state $\tau(q+1)$ (if the parent trajectory simply continued from $2\tau q$, then the parent index would double in every outer phase, $q\mapsto 2q\mapsto 4q$, causing exponential rather than polynomial growth). We therefore add actions that detect when the bundle reaches $2\tau q$ and execute $\Theta(q^k)$ reset plays at the flat maximum score $\zeta(q)$ to restore the parent scores to state $\tau(q+1)$ (shifted to maximum $\zeta(q)$). Since the reset block also has length of degree $k$ with positive leading coefficient, the full outer phase has length $\Theta(q^k)$, entry height $\Theta(q^k)$, and cumulative time $\sum_{a\le q}\Theta(a^k)=\Theta(q^{k+1})$, increasing all three degrees by one.

\paragraph{The endpoint action $g$ and the $2(k+1)$ reset actions.}
Recall that the parent actions are indexed by $0,\ldots,m-1$, with action $b$ attaining the maximum score at every parent phase entry (though $b$ need not remain the maximizer throughout the phase). The parent clock also comes equipped with a fixed rational row vector $v\in\Q^{1\times m}$, with $v_b=0$ and $v\mathbf c(\sigma)=1$, which will detect the end of the bundle. We adjoin two types of new actions (Figure~\ref{fig:word}):
\begin{enumerate}[label=(\roman*),leftmargin=*]
 \item \emph{One endpoint action $g$} (at index $m$ in the enlarged matrix), whose payoff row against the $m$ parent actions is set to $A_{b,:}+v$. Giving $g$ the same payoffs $A_{b,:}$ as action $b$ plus the perturbation $v$ ensures that whenever any parent action $i$ is played, the score difference $u_g-u_b$ changes simply by $(A_{b,i}+v_i)-A_{b,i}=v_i$, allowing $g$ to track $u_b$ plus the weighted remaining count $vz^{\rm rem}$. We refer to the $m$ parent actions together with $g$ as the $m+1$ \emph{old actions}.
 \item \emph{$2r$ reset actions} (where $r=k+1$), labeled $\reset{0},\ldots,\reset{2r-1}$ in their order of play (the sans-serif symbols $\reset{j}$ denote action labels, distinct from the constant $c$ and prefactors $c_k$). The first $r$ actions $\reset{0},\ldots,\reset{r-1}$ are the \emph{early reset actions}, and the last $r$ actions $\reset{r},\ldots,\reset{2r-1}$ are the \emph{late reset actions}. Together they move the old actions' scores to the entry state required for the next outer phase $q+1$.
\end{enumerate}

\begin{figure}[H]\centering
\begin{tikzpicture}[font=\small,node distance=3mm]
\node[parbox,text width=39mm] (p) {parent bundle\\$\W_k(\tau q)\cdots\W_k(2\tau q-1)$};
\node[gatbox,right=of p,text width=14mm] (g) {endpoint\\action $g$};
\node[conbox,right=of g,text width=29mm] (e) {$k+1$ early\\reset actions,\\old scores decrease};
\node[conbox,right=of e,text width=29mm] (l) {$k+1$ late\\reset actions,\\old scores increase};
\draw[arr] (p)--(g);\draw[arr] (g)--(e);\draw[arr] (e)--(l);
\node[below=3mm of p,align=center,font=\footnotesize] {$\tau q$ parent phases};
\node[below=3mm of g,align=center,font=\footnotesize] {one play};
\draw[decorate,decoration={brace,mirror,amplitude=4pt},draw=Control] ($(e.south west)+(0,-.15)$)--($(l.south east)+(0,-.15)$)
 node[midway,below=5pt,align=center,text width=62mm,font=\footnotesize] {scores reach the next entry state};
\end{tikzpicture}
\caption{The sequence in one phase of the enlarged game: a bundle of parent phases, one play of $g$, and the early and late reset runs. Section~\ref{sec:phaseboundary} later moves the phase boundary within this same sequence. All payoffs and parameters are fixed before play begins.}\label{fig:word}
\end{figure}
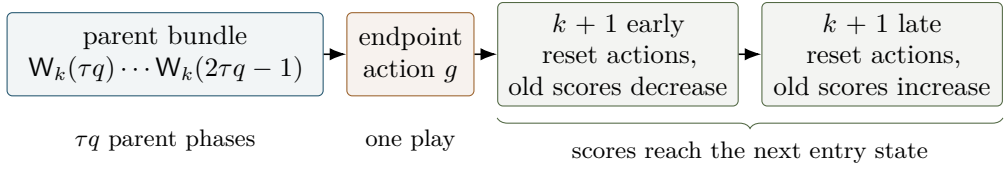

\paragraph{The new matrix.}
Let $G$ denote the enlarged recurrent payoff matrix of size $(m+1+2r)\times(m+1+2r)$ on the $m$ parent actions, $g$, and the $2r=2(k+1)$ reset actions (Figure~\ref{fig:matrix}). Its score and count vectors obey~\eqref{eq:process} with $G$ in place of $A$ (the single setup action is adjoined only after all lifts are complete). The top-left $m\times m$ parent block is unchanged equal to $A$, and 
skew-symmetry fixes the column of $g$ against the parent actions as $-(A_{b,:}+v)^\top$.

Write $P=[P_E\mid P_L]$ for the $(m+1)\times 2r$ coupling block from reset plays to old scores, where $P_E,P_L\in\Q^{(m+1)\times r}$. Because playing an early (or late) reset action adds the corresponding column of $P_E$ (or $P_L$) to the old actions' scores, we require entrywise negativity $P_E<0$ so that every early reset play decreases all old scores, and entrywise positivity $P_L>0$ so that every late reset play increases all old scores (these inequalities are entrywise, not matrix definiteness). Why use two groups of size $r=k+1$ with opposite signs $P_E<0$ and $P_L>0$? First, outer-phase counts are degree-$k$ polynomials, represented in a $(k+1)$-dimensional space. We choose $r=k+1$ actions per group with linearly independent count polynomials. These two full-rank blocks let us solve the coefficient identities while enforcing opposite signs; we do not claim that this number of actions is necessary. In fact, the required old-score change has degree at most $k-1$ (Lemma~\ref{lem:gateway}). Second, during the reset block, pushing old scores down first ($P_E<0$) and pulling them up second ($P_L>0$) makes every old score approach its next-entry target ($\le\zeta(q)$) from below without exceeding the reset ceiling $\zeta(q)$; simultaneously, by skew-symmetry ($-P^\top$), during the parent bundle, parent plays make early reset scores rise ($-P_E^\top>0$) toward their handoff value while late reset scores fall ($-P_L^\top<0$) and stay below the parent height.

Finally, the bottom-right $2r\times2r$ reset-reset block has the form $J-J^\top$, where $J\in\Q^{2r\times2r}$ is strictly lower triangular with positive entries below its diagonal ($J_{j,\ell}>0$ for $j>\ell$). This sign structure makes the reset actions $\reset{0},\ldots,\reset{2r-1}$ take turns in exact order at the flat maximum score $\zeta(q)$: when $\reset{\ell}$ is played, column $\ell$ of $J-J^\top$ has $0$ on the diagonal (keeping $\reset{\ell}$'s own score flat at $\zeta(q)$), negative entries $-J_{\ell,j}<0$ above the diagonal (pushing all earlier reset actions $\reset{0},\ldots,\reset{\ell-1}$ down below $\zeta(q)$), and positive entries $+J_{j,\ell}>0$ below the diagonal (pulling all later reset actions $\reset{\ell+1},\ldots,\reset{2r-1}$ up until $\reset{\ell+1}$ reaches $\zeta(q)$ and takes over). Sections~\ref{sec:moment}--\ref{sec:synthesis} construct $P$ and $J$ so that these sign requirements and the exact polynomial score-return equations hold simultaneously.

\begin{figure}[H]\centering
\begin{tikzpicture}[font=\small,x=1cm,y=1cm]
\draw[draw=Parent,fill=Parent!6] (0,0) rectangle (3,2.4);
\node at (1.5,1.2) {$A$};
\draw[draw=Gate,fill=Gate!8] (0,-.65) rectangle (3,0);
\node at (1.5,-.325) {$A_{b,:}+v$};
\draw[draw=Gate,fill=Gate!8] (3,0) rectangle (3.7,2.4);
\node[rotate=90] at (3.35,1.2) {$-(A_{b,:}+v)^{\top}$};
\draw[draw=Gate,fill=Gate!5] (3,-.65) rectangle (3.7,0);
\node at (3.35,-.325) {$0$};
\draw[draw=Control,fill=Control!6] (3.7,-.65) rectangle (5.7,2.4);
\node[align=center] at (4.7,.875) {$P_E<0$\\early};
\draw[draw=Control,fill=Control!12] (5.7,-.65) rectangle (7.7,2.4);
\node[align=center] at (6.7,.875) {$P_L>0$\\late};
\draw[draw=Control,fill=Control!6] (0,-1.55) rectangle (3.7,-.65);
\node at (1.85,-1.1) {$-P_E^{\top}$};
\draw[draw=Control,fill=Control!12] (0,-2.45) rectangle (3.7,-1.55);
\node at (1.85,-2) {$-P_L^{\top}$};
\draw[draw=Ink,fill=gray!5] (3.7,-2.45) rectangle (7.7,-.65);
\node at (5.7,-1.55) {$J-J^{\top}$};
\node[above,align=center,font=\footnotesize] at (1.5,2.4) {parent actions\\($m$)};
\node[above,align=center,font=\footnotesize] at (3.35,2.4) {$g$\\(1)};
\node[above,align=center,font=\footnotesize] at (4.7,2.4) {early reset\\($r$)};
\node[above,align=center,font=\footnotesize] at (6.7,2.4) {late reset\\($r$)};
\end{tikzpicture}
\caption{Block structure of the fixed skew-symmetric payoff matrix $G$, with $r=k+1$. Row groups follow the same order as column groups; parentheses give the number of actions. The parent block $A$ is unchanged. The signs shown for $P_E$ and $P_L$ are requirements on the new coupling; Lemma~\ref{lem:uniform} proves that they can be obtained together with both reset equations.}\label{fig:matrix}
\end{figure}

\subsection{Detecting the endpoint of a parent bundle}\label{sec:overviewendpoint}
How does the endpoint action $g$ stay inactive throughout the $\tau q$ phases of the parent bundle $\W_k(\tau q)\cdots\W_k(2\tau q-1)$ and then step in as the unique maximizer at the exact end of $\W_k(2\tau q-1)$? Fix an outer phase $q$, and at any decision time within the bundle, let $\mathbf u\in\Q^{m+1+2(k+1)}$ denote the current score vector in the enlarged game (suppressing the time subscript $t$ from $\mathbf u_t$), let $H_{\rm parent}=\max_{0\le i<m}u_i$ be the current maximum among the $m$ parent actions, and let $z^{\rm rem}\in\Z_{\ge0}^m$ be the column vector whose coordinate $z_i^{\rm rem}$ counts the plays of parent action $i$ still remaining in the current bundle (whereas $z_t$ counts all completed plays since $t=0$).

At the exact end of the bundle, no plays remain ($z^{\rm rem}=0$) and the parent is at state $2\tau q$, where action $b$ attains the parent maximum ($H_{\rm parent}-u_b=0$). Thus, both $H_{\rm parent}-u_b$ and $vz^{\rm rem}$ vanish at the bundle endpoint. Conversely, the parent clock conditions guarantee a rational margin $\delta_0>0$ such that at every \emph{proper} prefix of the bundle, their sum is bounded away from zero:
\[
 H_{\rm parent}-u_b+vz^{\rm rem}\ge\delta_0.
\]
Here $v\in\Q^{1\times m}$ is the endpoint-detection row from Section~\ref{sec:overviewblocks}, so $vz^{\rm rem}$ is a scalar. Because $v$ may have negative entries, $vz^{\rm rem}$ need not decrease monotonically at every step; what matters is that the combined expression $H_{\rm parent}-u_b+vz^{\rm rem}$ stays at least $\delta_0>0$ until the final play of the bundle. For instance, in the RPS base clock ($k=2$), where $\W_2(\sigma)=0^{3\sigma+5}1^{3\sigma+6}2^{3\sigma+7}$, the vector $v=(0,-1,1)$ gives $v\mathbf c_2(\sigma)=(3\sigma+7)-(3\sigma+6)=1$ in every phase and satisfies $vz^{\rm rem}\ge1$ across every proper prefix of the bundle (so $\delta_0=1$); in all higher lifts ($k\ge3$), $v$ is simply the unit coordinate row of the parent's own endpoint action, which plays once at the end of each parent phase and hence automatically satisfies $vz^{\rm rem}\ge1$ before the bundle ends.

\paragraph{Initialization and the score identity.}
Given the outer height-gain constant $c\in(0,1)$ (chosen in Section~\ref{sec:synthesis}), we set $\delta=\tfrac12\min\{c,\delta_0\}$, so $0<\delta<\min\{c,\delta_0\}$. At the start of the bundle, the reset block initializes $g$'s score relative to $b$ so that $u_g-u_b=\delta-vz^{\rm rem}$ (since $v\mathbf c(\sigma)=1$ for each of the $\tau q$ parent phases in the bundle, $vz^{\rm rem}=\tau q$ initially, so $u_g-u_b=\delta-\tau q$). Because $g$ has payoff row $A_{b,:}+v$ against the parent actions, each play of a parent action $i$ changes $u_g-u_b$ by $(A_{b,i}+v_i)-A_{b,i}=v_i$ while replacing $z^{\rm rem}$ by $z^{\rm rem}-e_i$ (decreasing $vz^{\rm rem}$ by $v_i$). Consequently, the sum $(u_g-u_b)+vz^{\rm rem}$ is invariant, and the identity
\begin{equation}\label{eq:overviewendpointidentity}
 u_g-u_b=\delta-vz^{\rm rem}
\end{equation}
holds at every decision time throughout the bundle. In particular, $u_g-u_b$ rises by $+1$ across each completed parent phase in the bundle.

\paragraph{Waiting at every proper prefix.}
Subtracting~\eqref{eq:overviewendpointidentity} from the parent maximum $H_{\rm parent}$ gives, at every proper prefix of the bundle,
\[
 H_{\rm parent}-u_g
 =H_{\rm parent}-u_b+vz^{\rm rem}-\delta
 \ge\delta_0-\delta>0.
\]
Thus, $u_g<H_{\rm parent}$ at every proper prefix, so $g$ cannot interrupt the parent bundle even when $b$ is not the current parent maximizer (the reset actions are bounded below $H_{\rm parent}$ separately in Section~\ref{sec:overviewreset}).

\paragraph{The endpoint.}
Immediately after the final play of the bundle $\W_k(\tau q)\cdots\W_k(2\tau q-1)$, no parent plays remain ($z^{\rm rem}=0$) and the parent is at state $2\tau q$, where $b$ attains the parent maximum $H_{\rm parent}=\zeta(q)-c$. By~\eqref{eq:overviewendpointidentity}, $g$'s score is now
\[
 u_g=u_b+\delta=H_{\rm parent}+\delta=\zeta(q)-c+\delta.
\]
Once the reset actions are also shown to lie strictly below this value, $g$ becomes the unique maximizer of the enlarged game and plays once.

\paragraph{Immediately after the play of $g$.}
In Figure~\ref{fig:matrix}, the column of $G$ for action $g$ against the parent actions is $-(A_{b,:}+v)^\top = Ae_b-v^\top$ (using $A^\top=-A$). Thus, the single play of $g$ changes the parent score vector by $Ae_b-v^\top$: it leaves $b$'s score unchanged at $\zeta(q)-c$ (since $A_{bb}=v_b=0$) and leaves $g$'s own score unchanged at $\zeta(q)-c+\delta$ (since $G_{g,g}=0$). Because every parent phase starts with at least two consecutive plays of $b$, adding $Ae_b$ at parent state $2\tau q$ keeps $b$ a maximizing parent action. Outside $\supp(v)$, subtracting $v^\top$ leaves this comparison unchanged; on $\supp(v)$, the clock's entry-safety condition makes $b$'s lead strict after the full update $Ae_b-v^\top$ (Lemma~\ref{lem:gateway}). Simultaneously, this single play of $g$ raises the first reset action $\reset{0}$ to the reset ceiling $\zeta(q)>\zeta(q)-c+\delta$, so the reset block takes over.

\subsection{How the added actions wait and restore the scores}\label{sec:overviewreset}
Figure~\ref{fig:score-chronology} illustrates the score trajectories and comparisons across outer phase $q$: panel (a) covers the parent bundle $\W_k(\tau q)\cdots\W_k(2\tau q-1)$ and the endpoint play of $g$, while panel (b) covers the $2(k+1)$ reset runs at the flat ceiling $\zeta(q)$. We explain each panel and then write down the matrix equations that choose the coupling block $P$.

\begin{figure}[H]\centering
\resizebox{.98\linewidth}{!}{% Insertable vector figure: requires tikz, amsmath, amssymb, and arrows.meta.
% Auxiliary fixed-priority scores only. Coordinates have no numeric meaning.
% All comparisons refer to integer decision times, q >= Q_out, k >= 2.
\begingroup
\definecolor{FivParent}{HTML}{24566F}
\definecolor{FivEndpoint}{HTML}{9B622C}
\definecolor{FivReset}{HTML}{536945}
\definecolor{FivInk}{HTML}{263443}
\definecolor{FivLate}{HTML}{684895}
\begin{tikzpicture}[x=1cm,y=1cm,
  font=\fontsize{8.7}{10.5}\selectfont,text=FivInk,
  axis/.style={draw=FivInk!70,line width=.45pt,-{Latex[length=1.5mm]}},
  guide/.style={draw=FivInk!30,densely dashed,line width=.35pt},
  pcurve/.style={draw=FivParent,line width=1.1pt},
  ecurve/.style={draw=FivReset,line width=.9pt},
  lcurve/.style={draw=FivLate,line width=.85pt,dash pattern=on 3pt off 1.8pt},
  smallnote/.style={font=\fontsize{8}{9.7}\selectfont,align=left},
  boundary/.style={font=\fontsize{8.3}{10.2}\selectfont,align=left,anchor=north west,inner sep=0pt}]
\path[use as bounding box] (0,-9.88) rectangle (16.25,.25);
\node[anchor=west,font=\bfseries\fontsize{9.6}{11.5}\selectfont] at (0,0)
 {Auxiliary scores before the fixed tie-breaking bias};
\node[anchor=east] at (16.2,0) {$r=k+1$};

% PANEL A: the bundle, plus exactly one update by the endpoint action.
\begin{scope}[shift={(0,-.65)}]
\node[anchor=west,font=\bfseries] at (0,0) {(a) Parent bundle and one endpoint play};
\fill[FivParent!4] (1,-3.72) rectangle (11.05,-.45);
\fill[FivEndpoint!7] (11.05,-3.72) rectangle (13.65,-.45);
\draw[axis] (1,-3.72)--(14.02,-3.72);
\draw[axis] (1,-3.72)--(1,-.33);
\node[anchor=east] at (.87,-.43) {score};
\draw[guide] (11.05,-3.72)--(11.05,-.52);
\draw[guide] (13.65,-3.72)--(13.65,-.52);
% Only boundary-height reference segments: no extraneous numeric axis scale.
\draw[guide] (11.05,-.85)--(14.0,-.85);
\node[anchor=west] at (14.0,-.85) {$\zeta(q)$};
\draw[guide] (13.65,-1.55)--(14.0,-1.55);
\node[anchor=west,text=FivEndpoint] at (14.0,-1.55) {$\zeta(q)-c+\delta$};
\draw[guide] (13.65,-2.00)--(14.0,-2.00);
\node[anchor=west,text=FivParent] at (14.0,-2.00) {$\zeta(q)-c$};
% A coarse monotone envelope: its corners do not specify actual jump times.
\draw[pcurve] (1,-2.62)--(2.65,-2.62)--(2.65,-2.48)--(4.2,-2.48)
 --(4.2,-2.32)--(6.0,-2.32)--(6.0,-2.16)--(8.35,-2.16)
 --(8.35,-2.00)--(13.65,-2.00);
\node[anchor=west,text=FivParent,fill=white,inner sep=1.5pt] at (1.8,-1.84)
 {$H_{\rm parent}$};
% \reset{0} strictly rises on parent plays, stays below H_parent, then jumps on g.
\draw[ecurve] (1,-3.26)--(11.05,-2.39)--(13.65,-.85);
\fill[FivReset] (11.05,-2.39) circle (1.7pt);
\fill[FivReset] (13.65,-.85) circle (1.9pt);
\node[anchor=north west,text=FivReset,fill=white,inner sep=1pt] at (2.0,-2.98)
 {$\reset{0}$};
% Last reset action: initial tie, then strictly falls on parent and g plays.
% No late-reset curve is drawn: its crossing time with \reset{0} is not fixed.
\draw[FivLate,fill=white,line width=.9pt] (1,-2.62) circle (2.5pt);
\fill[FivParent] (1,-2.62) circle (1.2pt);
% g: certified points only; no line is asserted inside the parent bundle.
% The initial g value is stated exactly below, without ordering it against \reset{0}.
\draw[draw=FivEndpoint,line width=.9pt] (11.05,-1.55)--(13.65,-1.55);
\fill[FivEndpoint] (11.05,-1.55) circle (1.8pt);
\fill[FivEndpoint] (13.65,-1.55) circle (1.8pt);
\node[anchor=south,text=FivEndpoint] at (12.1,-1.52) {$g$};
\draw[draw=FivEndpoint,{Latex[length=1.2mm]}-{Latex[length=1.2mm]},line width=.5pt]
 (10.58,-1.55)--(10.58,-2.00);
\node[anchor=east,text=FivEndpoint] at (10.48,-1.76) {$\delta$};
\node[anchor=south,text=FivReset] at (12.55,-.82) {$\reset{0}$ unique};
\node[anchor=north] at (1,-3.79) {entry};
\node[anchor=north] at (11.05,-3.79) {bundle end};
\node[anchor=north] at (13.65,-3.79) {after $g$};
\node[anchor=north,align=center,text=FivParent] at (5.7,-3.83)
 {Parent bundle};
\node[anchor=north,align=center,text=FivEndpoint] at (12.35,-4.14)
 {$g$ plays once};
\end{scope}

% PANEL B: a reset maximum, not the score of a single reset action.
\begin{scope}[shift={(0,-5.48)}]
\node[anchor=west,font=\bfseries] at (0,0) {(b) Reset actions and return to the next entry state};
\fill[FivReset!5] (1,-3.38) rectangle (7.6,-.55);
\fill[FivLate!4] (7.6,-3.38) rectangle (14.65,-.55);
\draw[axis] (1,-3.38)--(15.0,-3.38);
\draw[guide] (7.6,-3.38)--(7.6,-.65);
\draw[guide] (14.65,-3.38)--(14.65,-.65);
\node[anchor=east] at (.89,-1.02) {$\zeta(q)$};
\draw[draw=FivReset,line width=1.5pt] (1,-1.02)--(14.65,-1.02);
% The ellipses avoid coincident indexed labels at k=2.
\foreach \xx/\lab in {1.75/\reset{0},3.3/\cdots,5.4/\reset{r-1},8.55/\reset{r},10.6/\cdots,13.3/\reset{2r-1}} {
  \node[anchor=south,text=FivReset] at (\xx,-.96) {$\lab$};
}
\foreach \xx in {2.5,4.3,7.6,9.55,11.7} {
  \draw[draw=FivReset,line width=.7pt] (\xx,-1.10)--(\xx,-.94);
}
% Only parent maximum is traced. g and individual old-score orderings are not invented.
\draw[pcurve] (1,-1.84)--(3.8,-2.58)--(7.6,-2.98)--(11.4,-2.12)--(14.65,-1.02);
\node[anchor=north west,text=FivParent,fill=white,inner sep=1pt] at (1.48,-2.00)
 {$H_{\rm parent}$};
\draw[FivReset,fill=white,line width=.9pt] (14.65,-1.02) circle (2.5pt);
\fill[FivParent] (14.65,-1.02) circle (1.2pt);
\node[anchor=west,text=FivParent,align=left] at (14.78,-1.27) {$b$ wins\\final tie};
\node[anchor=north] at (1,-3.45) {after $g$};
\node[anchor=north] at (7.6,-3.45) {early / late};
\node[anchor=north] at (14.65,-3.45) {next entry};
\node[anchor=north,align=center,text=FivReset] at (4.3,-3.45)
 {Early reset: $P_E<0$};
\node[anchor=north,align=center,text=FivLate] at (11.1,-3.45)
 {Late reset: $P_L>0$};
\end{scope}

\end{tikzpicture}
\endgroup
}
\caption{Score comparisons before the fixed tie-breaking bias. (a) The parent bundle ends when $g$ first exceeds the parent height; one play of $g$ starts the reset block. (b) Successive reset actions attain the constant maximum $\zeta(q)$ while the old scores decrease and then rise to the next entry state. The curves are schematic and refer to integer decision times. The bias preserves the action sequence but need not preserve the flat maximum.}
\label{fig:score-chronology}

\end{figure}

\paragraph{Panel (a): the added actions wait for the bundle to end.}
As shown in Section~\ref{sec:overviewendpoint}, the endpoint action $g$ stays strictly below the parent height $H_{\rm parent}$ until the final play of the bundle, when $H_{\rm parent}$ reaches $\zeta(q)-c$ and $u_g$ reaches $\zeta(q)-c+\delta$. What keeps the $2(k+1)$ reset actions from interrupting the bundle?
\begin{itemize}[leftmargin=*]
 \item \emph{Early reset actions ($\reset{0},\ldots,\reset{r-1}$):} Because $P_E<0$, each parent play changes an early reset action's score by $-P_{ij}>0$, so early reset scores \emph{rise} during the bundle. To prevent $\reset{j}$ from overtaking the parent prematurely, we use the parent clock's remaining-height bound (Condition~(III) of Definition~\ref{def:clock} and Lemma~\ref{lem:bundle}): the parent has a fixed constant $\mathcal A_k\ge0$, and for any rational slope $\epsilon>0$ there is a phase threshold above which every prefix of the bundle satisfies
 \[
  H_{\rm end}-H_{\rm prefix} \le \epsilon |\W^{\rm rem}|+\mathcal A_k,
 \]
 where $H_{\rm end}=\zeta(q)-c$ is the bundle-exit height, $H_{\rm prefix}$ is the current parent height, and $|\W^{\rm rem}|=\one_m^\top z^{\rm rem}$ is the number of parent plays remaining in the bundle. If we choose the payoff magnitudes so that $-P_{ij}\ge\epsilon$ for all parent actions $i$ and prescribe $\reset{j}$'s score at the bundle endpoint to satisfy $u_{\reset{j}}^{\rm end}\le H_{\rm end}-\Delta$ with a deficit $\Delta>\mathcal A_k$, then working backward by $|\W^{\rm rem}|$ remaining parent plays gives
 \[
  u_{\reset{j}} \le u_{\reset{j}}^{\rm end}-\epsilon|\W^{\rm rem}|
  \le H_{\rm end}-\Delta-\epsilon|\W^{\rm rem}|
  \le H_{\rm prefix}-(\Delta-\mathcal A_k) < H_{\rm prefix}.
 \]
 The $\epsilon|\W^{\rm rem}|$ terms cancel, keeping every early reset action strictly below the parent height by at least $\Delta-\mathcal A_k>0$ throughout the bundle.
 \item \emph{Late reset actions ($\reset{r},\ldots,\reset{2r-1}$):} Because $P_L>0$, each parent play changes a late reset action's score by $-P_{ij}<0$, so late reset scores \emph{fall} throughout the bundle. At the very start of the bundle (the handoff from outer phase $q-1$), the last late reset action $\reset{2r-1}$ starts at $\zeta(q-1)$, tied with $b$; the parent action $b$ wins this entry tie by priority, and its immediate play pushes $\reset{2r-1}$ strictly below $\zeta(q-1)$.
\end{itemize}

\paragraph{Panel (b): the reset actions restore the old scores.}
At the bundle endpoint, the single play of $g$ raises the first early reset action $\reset{0}$ to the reset ceiling $\zeta(q)$. Because $\delta<c$, every old action lies at or below $\zeta(q)-c+\delta<\zeta(q)$, so $\reset{0}$ is the unique maximizer. The $2r$ reset actions $\reset{0},\ldots,\reset{2r-1}$ then take turns in order at the flat maximum score $\zeta(q)$ (driven by the lower-triangular entries of $J-J^\top$). During the $r$ early reset runs, $P_E<0$ pushes every old action's score downward; during the $r$ late reset runs, $P_L>0$ pulls every old action's score upward so that it approaches its next-entry target from below, staying strictly below $\zeta(q)$ until the final play of $\reset{2r-1}$. At that final play, $b$ reaches $\zeta(q)$ to tie $\reset{2r-1}$ and wins the handoff by priority, starting the next parent bundle at parent state $\tau(q+1)$ (shifted to maximum $\zeta(q)$). These statements describe the auxiliary fixed-priority process~\eqref{eq:process}; Lemma~\ref{lem:bias} later removes the boundary ties via a small initial score bias without altering the play sequence. This bias need not preserve the exactly flat reset maximum shown in the illustration.

\paragraph{The two sets of score changes.}
The reset must return every score to its next-entry value. Using the basis $p_k(q)$ from Section~\ref{sec:polynomialdata}, define width-$r=k+1$ coefficient tables by
\[
 \begin{aligned}
 \mathbf s(q)&:=Sp_k(q),& \mathbf x(q)&:=Xp_k(q),\\
 \mathbf s_{\mathrm{mid}}(q)&:=S_{\mathrm{mid}}p_k(q),&
 \mathbf d(q)&:=Dp_k(q),\qquad \mathbf m(q):=Mp_k(q).
 \end{aligned}
\]
The first four tables have $m+1$ rows, one per old action; $M\in\Z^{2r\times r}$ gives the reset run lengths. Their roles are:
\begin{itemize}[leftmargin=*]
 \item $\mathbf s(q)$ gives the old scores before the parent bundle and the single play of $g$. The old-action count vector $\mathbf x(q)$ counts the plays in exactly that word.
 \item With $G_O\in\Q^{(m+1)\times(m+1)}$ the old payoff block, the intermediate scores just after $g$ are $\mathbf s_{\mathrm{mid}}(q)=\mathbf s(q)+G_O\mathbf x(q)$. Thus, $S_{\mathrm{mid}}=S+G_OX$.
 \item The next-entry target is $\mathbf s(q+1)=S\mathsf{T}_rp_k(q)$. The reset must supply $\mathbf d(q)=\mathbf s(q+1)-\mathbf s_{\mathrm{mid}}(q)$, or $D=S\mathsf{T}_r-S_{\mathrm{mid}}=SN_r-G_OX$.
\end{itemize}
Let $Q\in\Q^{2r\times r}$ denote the coefficient table for the \emph{negative} of the reset-score change caused by the old word. Thus, that score change is $-Qp_k(q)$.
For the reset plays $\mathbf m(q)=Mp_k(q)$ to produce the required old-score change $\mathbf d(q)=Dp_k(q)$ while the old plays $\mathbf x(q)=Xp_k(q)$ produce the required reset-score change $-Qp_k(q)$ (via $-P^\top$), the fixed coupling matrix $P\in\Q^{(m+1)\times2(k+1)}$ must satisfy both matrix equations simultaneously:
\[
 \begin{aligned}
 PM&=D,\\
 P^{\top}X&=Q.
 \end{aligned}
\]
Together they require the compatibility identity
\[
 D^{\top}X=(PM)^{\top}X=M^{\top}(P^{\top}X)=M^{\top}Q.
\]
Lemma~\ref{lem:interpolate} shows that when $X$ and $M$ have full column rank $k+1$, this compatibility identity $D^\top X=M^\top Q$ is both necessary and sufficient for a solution $P$ to exist. We derive the explicit tables $S,X,S_{\mathrm{mid}},D,Q$ in Section~\ref{sec:gateway}, choose $M,P,J$ with the required signs ($M_{:,k}>0$ in its highest-degree column, $P_E<0$, $P_L>0$, and $J$ positive below its diagonal) in Sections~\ref{sec:moment}--\ref{sec:synthesis}, and verify every best-response comparison in Section~\ref{sec:chronology}.

\subsection{Making the recursion repeatable}\label{sec:phaseboundary}
To apply the inductive step again from order $k+1$ to order $k+2$, the enlarged game must itself satisfy all conditions of an order-$(k+1)$ clock (Definition~\ref{def:clock}). We make the entry-leadership and endpoint-detection conditions immediate by placing the official phase boundary of the order-$(k+1)$ clock just before the first reset action $\reset{0}$ rather than before the parent bundle.

Write $\Aw(q)$ for the parent bundle $\W_k(\tau q)\cdots\W_k(2\tau q-1)$ followed by the single play of $g$, and write $\Bw(q)$ for the reset block of $\reset{0},\ldots,\reset{2r-1}$ runs. In the preceding sections we analyzed outer phase $q$ in the order $\Aw(q)\Bw(q)$. Because the full trajectory is the infinite concatenation $\Aw(q)\Bw(q)\Aw(q+1)\Bw(q+1)\cdots$, regrouping the adjacent blocks as
\[
  \W_{k+1}(q)=\Bw(q)\Aw(q+1)
\]
leaves the infinite play sequence completely unchanged while shifting the phase boundaries (Figure~\ref{fig:rotation}). Under this grouping:
\begin{itemize}[leftmargin=*]
 \item \emph{New entry leader $b_+=\reset{0}$:} Phase $\W_{k+1}(q)$ begins with the first reset run $\reset{0}^{m_0(q)}$ of $\Bw(q)$ at the flat maximum score $\zeta(q)$, so $b_+=\reset{0}$ is the fixed entry-maximizing action. We choose the onset threshold large enough that this run has at least two plays in every phase.
 \item \emph{New endpoint counter $v_+=e_g^\top$:} Phase $\W_{k+1}(q)$ ends with the single play of $g$ at the end of $\Aw(q+1)$, and $g$ is never played earlier in $\Bw(q)\Aw(q+1)$. Thus, the unit coordinate row $v_+=e_g^\top$ counts exactly one remaining play of $g$ ($v_+z^{\rm rem}=1$) at every proper prefix of $\W_{k+1}(q)$ and drops to $0$ after the final play, satisfying the endpoint-detection inequality $H_{\rm prefix}-u_{b_+}+v_+z^{\rm rem}\ge1$ automatically with margin $\delta_{0,+}=1$.
\end{itemize}

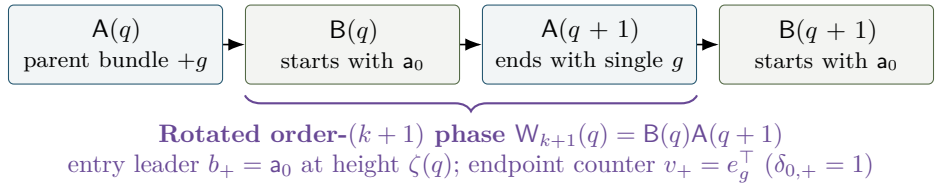
\begin{figure}[htbp]\centering
\begin{tikzpicture}[font=\small,x=1cm,scale=.95,transform shape]
\node[parbox,text width=27mm] (a) at (0,0) {$\Aw(q)$\\{\footnotesize parent bundle $+ g$}};
\node[conbox,text width=27mm] (b) at (3.3,0) {$\Bw(q)$\\{\footnotesize starts with $\reset{0}$}};
\node[parbox,text width=27mm] (c) at (6.6,0) {$\Aw(q+1)$\\{\footnotesize ends with single $g$}};
\node[conbox,text width=27mm] (k) at (9.9,0) {$\Bw(q+1)$\\{\footnotesize starts with $\reset{0}$}};
\draw[arr] (a)--(b);\draw[arr] (b)--(c);\draw[arr] (c)--(k);
\draw[decorate,decoration={brace,mirror,amplitude=6pt},thick,Accent] ($(b.south west)+(0,-.15)$)--($(c.south east)+(0,-.15)$)
  node[midway,below=7pt,align=center,text=Accent] {\textbf{Rotated order-$(k+1)$ phase $\W_{k+1}(q) = \Bw(q)\Aw(q+1)$}\\entry leader $b_+ = \reset{0}$ at height $\zeta(q)$; endpoint counter $v_+ = e_g^{\top}$ ($\delta_{0,+}=1$)};
\end{tikzpicture}
\caption{Rotating the phase boundaries from $\Aw(q)\Bw(q)$ to $\W_{k+1}(q) = \Bw(q)\Aw(q+1)$ without changing the infinite trajectory of play. In $\W_{k+1}(q)$, the endpoint action $g$ is the very last play of the phase.}\label{fig:rotation}
\end{figure}

Finally, to solve $P^\top X=Q$ with the required signs $P_E<0$ and $P_L>0$, our algebraic construction uses a positive-weight condition (Definition~\ref{def:admissible}): certain height-related polynomials must be expressible as a combination of the rows of the parent's count table $C$ with a strictly positive weight on every row. Section~\ref{sec:next} shows how to choose the free parameters of the current lift so that the new count table $C_{k+1}$ inherits this positive-weight property for the next lift, and Proposition~\ref{prop:inherit} completes the induction.

\paragraph{Action counts.}
Let $m_k$ be the number of recurrent actions in the square order-$k$ clock and let $n_k=m_k+1$ include its single setup action (distinct from the dimensions $m,n$ of a general rectangular game in Section~\ref{sec:fp}). Starting from the $3$-action RPS clock ($m_2=3$), each lift from order $k$ to order $k+1$ adds $1$ endpoint action and $2(k+1)$ reset actions, and one setup action is adjoined after reaching the target order:
\[
 m_k=3+\sum_{j=2}^{k-1}(2j+3)=(k+1)^2-6,
 \qquad n_k=m_k+1=(k+1)^2-5.
\]
For $k=3$, this gives $m_3=10$ recurrent actions ($3$ RPS, $1$ endpoint, and $6$ reset actions) and $n_3=11$ actions including setup, with quadratic phase length and entry height and cubic elapsed time.

\FloatBarrier

\section{Proof of the main theorem}\label{sec:clock}
The bridge from the RPS base to Theorem~\ref{thm:main} is the one-step Theorem~\ref{thm:lift}, stated with its hypotheses at the beginning of Section~\ref{sec:gateway} and proved at the end of Section~\ref{sec:chronology}. It preserves both the clock conditions and the additional algebraic data, so it can be iterated. We prove Theorem~\ref{thm:main} by constructing an order-$k$ clock for each $k\ge2$ and then realizing its action sequence as fictitious play from zero counts. In Section~\ref{sec:quantitative}, we state the recursive-family consequence needed for this argument and compute its growth. The more general one-step construction, Theorem~\ref{thm:lift}, is stated in the beginning of Section~\ref{sec:gateway}, with precise algebraic hypotheses in Section~\ref{sec:admissible}. We provide its proof in Sections~\ref{sec:gateway} and~\ref{sec:onestepproof}, with the concluding argument at the end of Section~\ref{sec:chronology}. In Section~\ref{sec:finish}, we remove ties and realize the required initial scores using Lemmas~\ref{lem:bias} and~\ref{lem:setup}. We then prove the gap estimate at every integer time and complete the proof in Section~\ref{sec:mainproof}.

\subsection{Definition of a clock and the bundle height bound}
We give a precise definition of an order-$k$ clock, collecting the data passed from one induction step to the next. It has three conditions. The first specifies an exact trajectory with polynomial phase boundaries. The second allows a new action to detect a phase endpoint. The third bounds the height still to be gained in terms of the number of plays still to come, so that early reset actions can wait throughout a parent bundle. For a row vector $v$, its \emph{support} is $\supp(v)=\{i:v_i\ne0\}$.

All thresholds below are obtainable by finite calculations with the rational coefficient data. In particular, when a positive rational $\epsilon$ is supplied, the threshold in condition (III) can be found in finitely many steps. We include this requirement as part of the clock data and verify it for the base and each lift; subsequent statements use the shorter phrase ``there is a threshold.''

\begin{definition}[Order-$k$ clock]\label{def:clock}
An \emph{order-$k$ clock} ($k\ge2$) consists of the data
\[
 (A,\prec,R,C,b,v,Q_k,\{\W_k(q)\}_{q\ge Q_k},\delta_0,\mathcal A_k,
   \widehat Q_k(\cdot)).
\]
Here $A$ is a rational skew-symmetric matrix $A \in \Q^{m \times m}$ on $m$ actions equipped with a fixed priority ordering $\prec$, a score table $R \in \Q^{m \times k}$, an integer count table $C \in \Z^{m \times k}$, an entry action $b \in \{0, \ldots, m-1\}$, a rational counter row $v \in \Q^{1 \times m}$, an integer phase threshold $Q_k \ge 1$, and explicitly specified words $\W_k(q)$ for all integers $q \ge Q_k$. The rational constants $\delta_0>0$ and $\mathcal A_k\ge0$, together with the computable threshold rule $\widehat Q_k(\epsilon)$, satisfy the following three properties:
\begin{enumerate}[label=\textbf{(\Roman*)},leftmargin=*]
  \item \textbf{Exact trajectory and polynomial growth.} For every integer $q \ge Q_k$, starting from the phase-entry score vector $\mathbf r(q) = R p_{k-1}(q)$, the word $\W_k(q)$ is followed by the auxiliary process~\eqref{eq:process}, has nonnegative action counts $\mathbf c(q) = C p_{k-1}(q) \ge 0$, and terminates at the exact next phase-entry score vector $\mathbf r(q+1) = R p_{k-1}(q+1)$. Algebraically,
  \[
    A C = R N_k, \qquad \rank(C) = k, \qquad \chivec := R_{b,:} \text{ satisfies } \chivec_{k-1} > 0.
  \]
  \item \textbf{Entry leadership and endpoint detection.} Every phase word $\W_k(q)$ ($q \ge Q_k$) begins with at least two consecutive plays of $b$. The counter row $v$ satisfies $v_b = 0$ and $v C = \erow{k}$, where $\erow{k}=(1,0,\ldots,0)\in\Q^{1\times k}$, and there is a rational margin $\delta_0 > 0$ such that at every proper prefix of $\W_k(q)$ ($q \ge Q_k$),
  \begin{equation}\label{eq:marker}
    H_{\rm prefix} - u_b + v z^{\rm rem} \ge \delta_0,
  \end{equation}
  Here $\mathbf u\in\Q^m$ is the score vector after the prefix, $H_{\rm prefix}=\max_i u_i$ is its height, and $z^{\rm rem}\in\Z_{\ge0}^{m}$ is the column vector of plays still remaining in $\W_k(q)$. Thus, $vz^{\rm rem}$ is a scalar. All three quantities depend on the prefix. Moreover, at every phase entry $q \ge Q_k$, the following entry inequality holds:
  \begin{equation}\label{eq:supportsafe}
    \bigl(\mathbf r(q) + A e_b - v^{\top}\bigr)_b > \bigl(\mathbf r(q) + A e_b - v^{\top}\bigr)_i
    \qquad \text{for every } i \in \supp(v).
  \end{equation}
  \item \textbf{Bound on the remaining height gain.} There is a rational constant $\mathcal A_k\ge0$ such that, for every rational $\epsilon>0$, there is a threshold $\widehat Q_k(\epsilon)\ge Q_k$ obtainable as described above, such that for all integers $q\ge\widehat Q_k(\epsilon)$ and every prefix of $\W_k(q)$ (including the empty and complete prefixes),
  \begin{equation}\label{eq:suffix}
    H_{\rm end} - H_{\rm prefix} \le \epsilon \lvert \W^{\rm rem} \rvert + \mathcal{A}_k,
  \end{equation}
  where $H_{\rm end}=\chi(q+1)=\chivec p_{k-1}(q+1)$ is the scalar phase-exit height and $\lvert \W^{\rm rem} \rvert = \one_m^{\top} z^{\rm rem}$ is the number of remaining plays.
\end{enumerate}
\end{definition}

We note that the counter row $v$ may have negative entries. Condition~\eqref{eq:supportsafe} ensures that when the endpoint action $g$ (whose play changes the parent scores by $-(A_{b,:} + v)^{\top}$) plays at the end of a parent bundle, any positive score change $-v_i > 0$ on an action $i \in \supp(v)$ with $v_i < 0$ cannot promote $i$ above the entry leader $b$. For example, the RPS base uses $v=(0,-1,1)$, so the definition must allow a negative entry. We check its entry inequalities in the base-case proof, Proposition~\ref{prop:seed}. After the first lift, the counter is $e_g^{\top}$ and simply counts the final endpoint play; we retain the more general definition to include RPS in the same induction.

We apply condition (III) to a whole bundle, rather than charging the constant $\mathcal A_k$ separately for every phase. The next lemma shows that only the first partially completed phase needs that constant.
\begin{lemma}[Remaining height gain across a parent bundle]\label{lem:bundle}
Let $(A, R, C, b, v, Q_k, \W_k)$ be an order-$k$ clock with constant $\mathcal{A}_k$ in condition (III). For every rational $\epsilon > 0$, there is a threshold above which the remaining-height bound~\eqref{eq:suffix} holds with the \emph{same} additive constant $\mathcal{A}_k$ across every bundle $\W_k(a)\W_k(a+1)\cdots\W_k(a+L-1)$ with $L\ge1$ and starting index $a$ above that threshold.
\end{lemma}
\begin{proof}
Condition (I) says that each coordinate $\mathbf c_i(q)=(Cp_{k-1}(q))_i$ is a nonnegative action count for all integers $q\ge Q_k$. A nonzero polynomial that is nonnegative at all sufficiently large integers has a positive leading coefficient. Since $\rank(C)=k$, the last column of $C$ is nonzero. Its nonzero entries are positive by the preceding observation, so the total phase length has degree $k-1$ with positive leading coefficient. The full-phase height increment $\chi(\sigma+1)-\chi(\sigma)$ has degree at most $k-2$. We can therefore choose a threshold $Q^\sharp(\epsilon)\ge\widehat Q_k(\epsilon)$ such that
\[
 \chi(\sigma+1)-\chi(\sigma)\le\epsilon|\W_k(\sigma)|
 \qquad(\sigma\ge Q^\sharp(\epsilon)).
\]
This bound for a complete phase has no additive constant. A threshold can be found, for example, by comparing the positive leading coefficient of $\epsilon|\W_k(\sigma)|-[\chi(\sigma+1)-\chi(\sigma)]$ with the sum of the absolute values of its lower coefficients.

Consider a proper prefix of the bundle $\W_k(a)\cdots\W_k(a+L-1)$, with $a\ge Q^\sharp(\epsilon)$. Let $r$ be the index of the phase containing the next play, and let $d$ be the number of plays remaining in that phase. At a boundary we take the next phase with its empty prefix. The remainder of the bundle consists of these $d$ plays followed by the complete phases $r+1,\ldots,a+L-1$. Condition (III) bounds the first part, and the preceding display bounds every later full phase. Hence
\[
\begin{aligned}
 H_{\rm bundle\text{-}end}-H_{\rm prefix}
 &=\chi(a+L)-H_{\rm prefix}\\
 &=\bigl[\chi(r+1)-H_{\rm prefix}\bigr]
   +\sum_{\sigma=r+1}^{a+L-1}\bigl[\chi(\sigma+1)-\chi(\sigma)\bigr]\\
 &\le\epsilon d+\mathcal A_k
       +\epsilon\sum_{\sigma=r+1}^{a+L-1}|\W_k(\sigma)|\\
 &=\epsilon|\W_{\rm bundle}^{\rm rem}|+\mathcal A_k.
\end{aligned}
\]
For the complete prefix, the remaining length and height gain are both zero, so the same inequality holds.
\end{proof}

\subsection{The quantitative effect of one step}\label{sec:quantitative}
The clock conditions control both the action sequence and its scale. The
following simple calculation explains the exponent before any coupling
coefficients are chosen.

\begin{lemma}[Growth of a clock and a parent bundle]\label{lem:clockgrowth}
For an order-$k$ clock there are positive constants $\mathsf L_k,h_k$ such that
\[
 |\W_k(q)|=\mathsf L_kq^{k-1}+O(q^{k-2}),\qquad
 \chi(q)=h_kq^{k-1}+O(q^{k-2}).
\]
Explicitly, 
\[
\mathsf L_k= \frac{(\one_m^{\top}C)_{k-1}}{(k-1)!} \qquad \text{ and } \qquad
h_k=\frac{\chivec_{k-1}}{(k-1)!}.
\]
For a fixed positive integer $\tau$, its parent bundle has length
\begin{equation}\label{eq:bundlegrowth}
 \sum_{\sigma=\tau q}^{2\tau q-1}|\W_k(\sigma)|
 =\frac{\mathsf L_k \tau^k(2^k-1)}{k}q^k+O(q^{k-1}).
\end{equation}
If the polynomial $\zeta$ satisfies $\zeta(q)-\zeta(q-1)=\chi(2\tau q)-\chi(\tau q)+c$, for some fixed constant $c\in\R$, then
\begin{equation}\label{eq:heightgrowth}
 \zeta(q)=\frac{h_k\tau^{k-1}(2^{k-1}-1)}{k}q^k+O(q^{k-1}).
\end{equation}
\end{lemma}
In the construction, $c$ is the same fixed height gain introduced in Section~\ref{sec:overviewblocks}, and we use $0<c<1$; the polynomial conclusion itself holds for any fixed scalar $c$.
\begin{proof}
The proof proceeds in four steps.

\smallskip\noindent\textbf{Step 1: the length of one phase.}
By condition (I) of Definition~\ref{def:clock}, each count $\mathbf c_i(q)$ is a polynomial of degree at most $k-1$ and is nonnegative for every sufficiently large integer $q$. Its coefficient of $q^{k-1}$ is therefore nonnegative. Full column rank of $C$ implies that its last column is nonzero, so at least one of these coefficients is positive. Summing the counts over the $m$ actions gives
\[
 |\W_k(q)|=\one_m^{\top}Cp_{k-1}(q)
 =\frac{(\one_m^{\top}C)_{k-1}}{(k-1)!}q^{k-1}+O(q^{k-2}).
\]
This proves the formula and positivity of $\mathsf L_k$.

\smallskip\noindent\textbf{Step 2: the height at phase entry.}
By definition, action $b$ is selected at entry, so the maximum score is $\chi(q)=R_{b,:}p_{k-1}(q)$. The assumption $\chivec_{k-1}>0$ gives the stated formula and positivity of $h_k$.

\smallskip\noindent\textbf{Step 3: the length of the parent bundle.}
For fixed $\tau$, summing the leading powers gives
\[
 \sum_{\sigma=\tau q}^{2\tau q-1}\sigma^{k-1}
 =\frac{(2\tau q)^k-(\tau q)^k}{k}+O(q^{k-1}).
\]
The sum of the lower-order errors $O(\sigma^{k-2})$ over $\tau q$ terms is $O(q^{k-1})$. Multiplying by $\mathsf L_k$ proves~\eqref{eq:bundlegrowth}.

\smallskip\noindent\textbf{Step 4: the height of the enlarged phase.}
The right side of the height-difference equation has leading term
$h_k\tau^{k-1}(2^{k-1}-1)q^{k-1}$; the constant $c$ does not change that term because $k\ge2$. The backward difference of a nonconstant degree-$d$ polynomial has degree $d-1$, so $\zeta$ has degree $k$. If its leading coefficient is $a$, then its backward difference is $ka q^{k-1}+O(q^{k-2})$. Matching coefficients gives~\eqref{eq:heightgrowth}.
\end{proof}

We now state the consequence of the construction used to prove the main theorem. We begin with the RPS clock and extend that particular family at every order. The more general extension theorem requires additional algebraic assumptions; we state it as Theorem~\ref{thm:lift} at the opening of Section~\ref{sec:gateway}.

\begin{proposition}[Recursive family starting from RPS]\label{prop:recursive-family}
There is a sequence of order-$k$ clocks, one for each integer $k\ge2$, whose order-$2$ member is the RPS clock of Proposition~\ref{prop:seed}. If its order-$k$ member has $m_k$ actions, then its order-$(k+1)$ member has $m_k+2k+3$ actions and retains the parent matrix as a principal block. Its new phase word is $\Bw(q)\Aw(q+1)$, its entry height is $\zeta(q)$, and its new endpoint counter is $e_g^{\top}$. The constants in the clock definition satisfy $\delta_{0,k+1}=1$ and $\mathcal A_{k+1}=\mathcal A_k+1$.
\end{proposition}
\begin{proof}
The base case has two ingredients: the explicit RPS phase trajectory (Proposition~\ref{prop:seed}) and strictly positive weights with an integer correction (Proposition~\ref{prop:seedadmissible}). These weights and correction are the additional data formalized in Definition~\ref{def:admissible}. The one-step theorem, Theorem~\ref{thm:lift}, extends a clock with those data and returns the same type of data for the new clock. Repeated application therefore gives every order. Proposition~\ref{prop:inherit} checks the stated phase word, height, and new constants at each step.
\end{proof}

The general lift, stated in Section~\ref{sec:admissible}, extends an order-$k$ clock when its count table admits the required strictly positive weights and integer correction. It returns both an order-$(k+1)$ clock and the same type of algebraic data, so the operation can be repeated. That section states the precise hypotheses and gives the proof roadmap; the remaining proof map below identifies each supporting result.

The next corollary gives the leading coefficients of the extended family. The reset runs add to the phase length, while the height increment is determined by the parent bundle.

\begin{corollary}[Leading coefficients after the lift]\label{cor:liftgrowth}
Consider one extension in Proposition~\ref{prop:recursive-family}, from an order-$k$ parent to an order-$(k+1)$ clock. Let $M\in\Z^{2(k+1)\times(k+1)}$ be the reset count table for that extension, and let $\tau$ be its parent scale. Then
\begin{align}
 h_{k+1}&=\frac{h_k\tau^{k-1}(2^{k-1}-1)}{k},\label{eq:heightrecurrence}\\
 \mathsf L_{k+1}&=\frac{\mathsf L_k\tau^k(2^k-1)}{k}
              +\frac{(\one^{\top}M)_k}{k!}>0.\label{eq:lengthrecurrence}
\end{align}
\end{corollary}
\begin{proof}
For the phase word $\W_{k+1}(q)=\Bw(q)\Aw(q+1)$, the entry height is $\zeta(q)$. Formula~\eqref{eq:heightgrowth} therefore gives~\eqref{eq:heightrecurrence}. The phase length is the sum of three terms:
\[
 |\W_{k+1}(q)|
 =\one_{2(k+1)}^{\top}Mp_k(q)
  +\sum_{\sigma=\tau(q+1)}^{2\tau(q+1)-1}|\W_k(\sigma)|+1.
\]
The reset-count term has coefficient $(\one^{\top}M)_k/k!$ at degree $k$. The bundle term has the coefficient in~\eqref{eq:bundlegrowth}, because replacing $q$ by $q+1$ does not change the leading coefficient. The final single play has degree zero. Adding these coefficients proves~\eqref{eq:lengthrecurrence}. The parent contribution is positive, and the construction gives positive degree-$k$ reset-count coefficients (Lemma~\ref{lem:uniform}), so $\mathsf L_{k+1}>0$.
\end{proof}

\subsection{Removing ties and realizing the initial scores}\label{sec:finish}
The clocks above use prescribed rational initial scores and may resolve ties by a fixed priority rule. The main theorem requires zero initial counts and unique best responses after the first action. We address the two issues in order: a small fixed bias makes the same action sequence strict, and one new action produces the biased scores in the first round of fictitious play.

\subsubsection{A fixed bias removes ties}
The next lemma shows that we can add rational offsets small enough to preserve every strict score comparison and ordered so that every former tie is resolved by the original priority.
\begin{lemma}[A rational bias realizes a priority rule]\label{lem:bias}
Let $A\in\Q^{m\times m}$ be skew-symmetric and let $f\in\Q^m$. There is a rational vector $\beta\in\Q^m$ such that the score-update process~\eqref{eq:process} initialized at $\mathbf u_0=f+\beta$ has a unique maximizer at every decision time and follows exactly the fixed-priority trajectory from $f$.
\end{lemma}
\begin{proof}
Let $D_0$ be the least positive integer clearing the denominators of $A$ and $f$, and let $r(i)\in\{0,\ldots,m-1\}$ be the priority rank of action $i$, with rank $0$ highest. Set
\begin{equation}\label{eq:bias}
  \beta_i=\frac{m-1-r(i)}{4mD_0}.
\end{equation}
Every unperturbed score difference is a multiple of $1/D_0$, at every integer time. Thus, a strict difference has magnitude at least $1/D_0$, while $|\beta_i-\beta_j|<1/(4D_0)$. The bias preserves every strict comparison and orders any tied scores by the original priority. Inductively, the chosen actions agree at every decision time and the score vectors continue to differ by the same fixed bias. Hence every selected maximizer in the biased process is unique.
\end{proof}

\subsubsection{One initial action produces the required scores}
We add one action whose payoff column gives the biased initial scores, up to a common translation. The translation ensures that the new action is never selected after the first round.
\begin{lemma}[One setup action]\label{lem:setup}
Any rational skew-symmetric score process with unique maximizers and initial vector $f+\beta$ can be realized, up to a common score translation, by simultaneous fictitious play from zero counts in a game with one additional action. Suppose both players choose the setup action in the first round. Then the setup action is never played again.
\end{lemma}
\begin{proof}
List the setup action first, followed by the $m$ recurrent actions, and define
\begin{equation}\label{eq:setup}
  f^*=f+\beta+\bigl(1-\min_i(f_i+\beta_i)\bigr)\one_m\ge\one_m,
  \qquad
  \widehat A=\begin{pmatrix}0&-(f^*)^{\top}\\ f^*&A\end{pmatrix}.
\end{equation}
Both players choose the new setup action at decision time $0$. After this round, the recurrent score vector is $f^*$ and the setup score is $0$. The common translation does not change comparisons among recurrent actions. Their maximum is nondecreasing by~\eqref{eq:monotone} and is initially at least $1$. Each recurrent play changes the setup score by $-f_i^*\le-1$. Therefore the setup action stays strictly below the recurrent height at every subsequent decision time. The two players have the same unique recurrent best response at each subsequent time, so their counts remain equal and realize the desired word.
\end{proof}

\subsection{The duality gap at every integer time}
The phase formulas determine the gap at phase boundaries. To prove the theorem at every integer time, we also need to control the plays between those boundaries. The next lemma uses the nondecreasing maximum score to enclose an entire phase between its endpoint heights. The fixed bias, setup translation, and starting phase change only lower-order terms.
\begin{lemma}[All integer times]\label{lem:alltime}
Let an order-$k$ clock have count table $C_k$ and entry-score table $R_k$, with entry action $b$. Apply Lemmas~\ref{lem:bias} and~\ref{lem:setup}, starting at any fixed phase $Q\ge Q_k$. Define
\[
  L=(\one^{\top}C_k)_{k-1}>0,\qquad \xi=(R_k)_{b,k-1}>0,
  \qquad a=\frac{L}{k!},\qquad b_* =\frac{\xi}{(k-1)!}.
\]
Equivalently, $a=\mathsf L_k/k$ and $b_*=h_k$ in the ordinary-coefficient
notation of Lemma~\ref{lem:clockgrowth}. For the resulting game,
\begin{equation}\label{eq:constant}
  \Gap_{\widehat A_k}(z_t/t,z_t/t)
     =c_k t^{-1/k}(1+o(1)),\qquad
  c_k=2b_*a^{-(k-1)/k}>0,
\end{equation}
as $t\to\infty$ over all integers.
\end{lemma}
\begin{proof}
The positivity of $L$ follows from eventual nonnegative counts and full column rank, as in the proof of Lemma~\ref{lem:bundle}. Let $t(q)$ be the number of completed rounds at the entry of phase $q$, including the initial setup round. Then
\[
  t(q)=1+\sum_{\sigma=Q}^{q-1}\one^{\top}C_kp_{k-1}(\sigma)
       =a q^k+O(q^{k-1}),
\]
using $\sum_{\sigma=Q}^{q-1}\binom{\sigma}{j}=\binom{q}{j+1}-\binom{Q}{j+1}$. At this phase boundary, the recurrent score vector is $R_kp_{k-1}(q)$ plus a fixed bias and a common fixed translation. The biased priority trajectory still selects $b$ at each phase entry. Hence the maximum score is
\[
  H_{t(q)}=b_*q^{k-1}+O(q^{k-2}).
\]
For any integer $t(q)\le t<t(q+1)$, monotonicity gives
\[
  H_{t(q)}\le H_t\le H_{t(q+1)}.
\]
Moreover, $t(q+1)/t(q)=1+O(q^{-1})$ and $H_{t(q+1)}/H_{t(q)}=1+O(q^{-1})$. The constants in these two estimates are fixed by the constructed game and the starting phase $Q$. The same constants can be enlarged to cover every integer $t$ in the interval $t(q)\le t<t(q+1)$, because the endpoint bounds enclose the whole phase. Therefore, uniformly over those times,
\[
  t=a q^k(1+O(q^{-1})),\qquad
  H_t=b_*q^{k-1}(1+O(q^{-1})).
\]
The setup action is never a maximizer after its initial play, so this $H_t$ is also the maximum over the entire game. Substituting into~\eqref{eq:symmetricgap} yields~\eqref{eq:constant}. In fact, the relative error obtained here is $O(t^{-1/k})$.
\end{proof}

\subsection{Proof of the main theorem}\label{sec:mainproof}
\begin{proof}[Proof of Theorem~\ref{thm:main}]
We prove the theorem in three steps.

\smallskip\noindent\textbf{Step 1: construct the clock.}
The order-$2$ RPS clock is the base case. Proposition~\ref{prop:recursive-family} extends it to every order $k\ge2$. At each extension from order $j$ to $j+1$, the action count increases by $2j+3$, so
\begin{equation}\label{eq:dimension}
 m_k=3+\sum_{j=2}^{k-1}(2j+3)=(k+1)^2-6,
 \qquad \mathcal A_k=k-1.
\end{equation}
The score and count tables are rational and integer, respectively, as required in Definition~\ref{def:clock}.

\smallskip\noindent\textbf{Step 2: obtain a strict fictitious-play trajectory.}
Choose any integer starting phase $Q\ge Q_k$, where $Q_k$ is the clock threshold in Definition~\ref{def:clock}; all prescribed phases from $Q$ onward then satisfy its trajectory and comparison conditions. Initialize the auxiliary process at $f=R_kp_{k-1}(Q)$. Lemma~\ref{lem:bias} removes ties without changing the sequence of actions. Lemma~\ref{lem:setup} produces these scores from zero counts using one additional common first action. There are therefore $n_k=m_k+1=(k+1)^2-5$ actions, and every best response after the first round is unique.

\smallskip\noindent\textbf{Step 3: compute the gap.}
Lemma~\ref{lem:alltime} gives
$\Gap_{\widehat A_k}(z_t/t,z_t/t)=c_kt^{-1/k}(1+o(1))$
at every integer time, with $c_k>0$. This implies the claimed $\Theta(t^{-1/k})$ rate.
\end{proof}

Appendix~\ref{app:bounds} provides bounds that are sufficient to choose the reset-count scale, and computes the leading gap constant. Appendix~\ref{app:examples} provides explicit payoff matrices for the cases $k=2$ and $k=3$, as an illustration.

\Needspace{24\baselineskip}
\subsection{Remaining proof map}\label{sec:proofmap}
The deduction of Theorem~\ref{thm:main} above uses the following dependencies. Theorem~\ref{thm:lift} is the induction bridge; its five proof steps appear in Section~\ref{sec:gateway} and Sections~\ref{sec:moment}--\ref{sec:chronology}. Table~\ref{tab:proofmap} separates a statement location from the location of its proof.
\begin{table}[htbp]\centering\small
\renewcommand{\arraystretch}{1.15}
\begin{tabular}{@{}>{\raggedright\arraybackslash}p{.37\linewidth}>{\raggedright\arraybackslash}p{.35\linewidth}>{\raggedright\arraybackslash}p{\dimexpr.28\linewidth-4\tabcolsep\relax}@{}}\toprule
Ingredient & Statement & Detailed proof\\\midrule
Induction bridge & Theorem~\ref{thm:lift} in Section~\ref{sec:gateway} & Section~\ref{sec:liftproof}\\
Base case $k=2$: RPS phases and the extra data for the first lift & Propositions~\ref{prop:seed}, \ref{prop:seedadmissible} & Section~\ref{sec:seed}; Section~\ref{sec:rpsadmissible}\\
Endpoint identity and parent-score comparison & Lemma~\ref{lem:gateway} & Section~\ref{sec:gateway}\\
Compatibility of both reset equations & Lemmas~\ref{lem:interpolate}, \ref{lem:cokernel}, \ref{lem:scale}; Proposition~\ref{prop:momentchoice} & Section~\ref{sec:gateway}; Sections~\ref{sec:moment}--\ref{sec:synthesis}\\
Coupling signs and waiting margins & Lemma~\ref{lem:uniform} & Section~\ref{sec:synthesis}; Appendix~\ref{app:bounds}\\
Positive weights for the next lift & Proposition~\ref{prop:nextweights} & Section~\ref{sec:next}\\
Every proposed play is a best response in the full game & Proposition~\ref{prop:chronology} & Section~\ref{sec:chronology}\\
All conditions for the next order & Proposition~\ref{prop:inherit}; Theorem~\ref{thm:lift} & Section~\ref{sec:chronology}\\
\bottomrule
\end{tabular}
\caption{Dependencies in the proof of Theorem~\ref{thm:main}, with locations for statements and proofs.}\label{tab:proofmap}
\end{table}

\section{The base case \texorpdfstring{$k=2$}{k=2}: rock--paper--scissors}\label{sec:seed}
We prove the base case of the induction: the RPS game gives an order-$2$ clock. We use the payoff matrix and polynomial tables displayed in~\eqref{eq:rpsmatrix}--\eqref{eq:rpstables} and verify all three conditions of Definition~\ref{def:clock}. Proposition~\ref{prop:seed} verifies their exact word, endpoint condition, and remaining-height bound, giving the order-two clock that starts the proof of Theorem~\ref{thm:main}. Proposition~\ref{prop:seedadmissible} later supplies the additional algebraic data for the first lift. The known square-root upper bound for the three-action RPS game is a special case of \citet*[Theorem 11]{lazarsfeld25a}. The setup-augmented game is handled by our own construction; here we check the stronger phase information needed for the induction.

The shift $\sigma=q+1$ makes the additional algebraic data needed for the first lift especially simple.

\begin{proposition}[RPS clock]\label{prop:seed}
Equations~\eqref{eq:rpsmatrix}--\eqref{eq:rpstables}, with entry action $b=0$
and counter $v=(0,-1,1)$, define an order-$2$ clock with $m=3$ actions in Definition~\ref{def:clock}. The game starting from zero counts has one additional setup action, hence four actions in total. One may take
\[
 Q_2=1,\qquad\delta_0=1,\qquad\mathcal A_2=1,\qquad
 \widehat Q_2(\epsilon)=1\quad(\epsilon>0).
\]
Every best response in its phase words is unique.
\end{proposition}
\begin{proof}
We check the exact trajectory first, then the endpoint-counter inequalities, and finally the bound on the remaining height gain. For the trajectory calculation, put $\sigma=q+1$ and write
$\mathbf r_{\rm RPS}(\sigma)=(\sigma,-2\sigma-5/3,\sigma-1/3)^\top$.
Starting from this state, the three runs have lengths
$3\sigma+2,3\sigma+3,3\sigma+4$. Their complete run-entry score vectors and next winning
heights are
\[
\begin{array}{c|c|c|c}
\text{selected action}&\text{run-entry scores}&\text{length}&\text{next height}\\\hline
0&(\sigma,-2\sigma-5/3,\sigma-1/3)&3\sigma+2&\sigma+1/3\\
1&(\sigma,\sigma+1/3,-2\sigma-7/3)&3\sigma+3&\sigma+2/3\\
2&(-2\sigma-3,\sigma+1/3,\sigma+2/3)&3\sigma+4&\sigma+1
\end{array}
\]
During run $0$, coordinate $1$ increases at rate one and coordinate $2$
decreases at rate one. At the last proper integer offset, coordinate $1$
is $2/3$ below the leader; coordinate $2$ stays at least $1/3$ below it.
During run $1$, coordinate $2$ increases at rate one and coordinate $0$
decreases at rate one, with the same two margins. During run $2$,
coordinate $0$ increases and coordinate $1$ decreases, again with these
margins. Each next action overtakes by $1/3$ exactly at its stated handoff.
The endpoint is $\mathbf r_{\rm RPS}(\sigma+1)$, so this proves the whole infinite strict word.

To check the coefficient-table identity in Definition~\ref{def:clock}, recall explicitly that
\[
 A_{\mathrm{RPS}}=\begin{pmatrix}0&-1&1\\1&0&-1\\-1&1&0\end{pmatrix},\quad
 R_2=\begin{pmatrix}1&1\\-11/3&-2\\2/3&1\end{pmatrix},\quad
 C_2=\begin{pmatrix}5&3\\6&3\\7&3\end{pmatrix},\quad
 N_2=\begin{pmatrix}0&0\\1&0\end{pmatrix}.
\]
Here $A_{\mathrm{RPS}}$ is the payoff matrix, $R_2,C_2$ are the width-two score and count tables in~\eqref{eq:rpstables}, and $N_2$ takes the forward difference of a coefficient row. Direct multiplication gives
\[
 A_{\mathrm{RPS}}C_2=R_2N_2=\begin{pmatrix}1&0\\-2&0\\1&0\end{pmatrix}.
\]
The first two rows of $C_2$ have determinant $-3$, so $\rank(C_2)=2$. All counts are positive for $q\ge1$, the first run
has at least eight plays, and the entry-height row is $\chivec=(1,1)$.
The counter satisfies $v_0=0$ and $vC_2=(1,0)$.
Within a parent phase, after $h$ plays of the current run, its remaining
counter is respectively
\[
 v z^{\rm rem}=1,\qquad 1+h,\qquad 3\sigma+4-h
\]
on runs $0,1,2$. It is therefore at least one at every proper prefix.
Since the current maximum is at least coordinate $0$, this proves
\eqref{eq:marker} with margin one. The two supported entry gaps in
$\mathbf r_2(q)+A_{\mathrm{RPS}}e_0-v^\top$ are $3q+8/3$ and $7/3$, proving
\eqref{eq:supportsafe} despite the negative counter entry.
Finally, the total height gain of one phase is one and the maximum is
nondecreasing. The remaining gain is at most one at every prefix, giving
\eqref{eq:suffix} with additive constant one and constant threshold.
\end{proof}

The same remaining-height estimate must also hold across a bundle of many phases, with the additive constant $\mathcal A_2=1$ used only once. The unfinished first phase contributes at most one to the increase of the maximum score. Every later complete phase contributes another one, which we bound by $\epsilon$ times that phase's length rather than by another additive constant. Since $|\W_2(q)|=9q+18$, we can choose
\[
 Q^\sharp(\epsilon)=\max\{1,\lceil1/(9\epsilon)\rceil\}
\]
to work in Lemma~\ref{lem:bundle}: for $q\ge Q^\sharp(\epsilon)$, each complete phase satisfies $1\le\epsilon(9q+18)$. Thus, the single-phase threshold $\widehat Q_2(\epsilon)=1$ and the bundle threshold $Q^\sharp(\epsilon)$ serve different purposes; the latter absorbs the height increases of all later complete phases into the remaining-length term.

\paragraph{The standard square-root rate.}
For this particular trajectory, summing the phase lengths gives elapsed time $(9/2)q^2+O(q)$ at height $q+O(1)$; monotonicity between phase boundaries therefore gives gap $\Theta(t^{-1/2})$ at all integer times. The displayed scores differ by a common constant from the scores generated by the following identical nonnegative initial count vectors for the two players:
\[
 R_2p_1(q)+\frac23\one_3
 =A_{\mathrm{RPS}}\begin{pmatrix}0\\q+4/3\\2q+3\end{pmatrix}.
\]
Adding the same positive multiple of $\one_3$ makes every coordinate of these initial count vectors positive without changing the scores.
This observation identifies the parent trajectory; the zero-count
initialization of the final games is supplied later by
Lemma~\ref{lem:setup}.

\paragraph{Remark: Why three actions are a natural starting point.}
A two-action skew-symmetric matrix has the form
$A=\bigl(\begin{smallmatrix}0&-a\\a&0\end{smallmatrix}\bigr)$.
For any initial score vector, the difference $u_1-u_0$ increases by $a$ at
every update, whichever action is chosen. If $a\ne0$, then the maximizing action
can change at most once; if $a=0$, then the scores are constant. Thus, a
nontrivial recurrent clock cannot start with a two-action synchronized
skew-symmetric game. In ordinary non-skew-symmetric $2\times2$ games,
fictitious play can switch repeatedly between actions, but importing their two separate count vectors is unnecessary for
this construction. RPS supplies the required clock directly within the
skew-symmetric setting used throughout the proof.

\section{The one-step theorem and the endpoint/reset equations}\label{sec:gateway}
The goal of Sections~\ref{sec:gateway} and~\ref{sec:onestepproof} is to lift an order-$k$ clock to an order-$(k+1)$ clock while retaining the parent payoff matrix as a principal block and adding only $2k+3$ new actions. The clock conditions in Definition~\ref{def:clock} control the parent trajectory, but they do not by themselves provide all of the algebraic data needed to make this lift repeatable. We therefore carry additional data---an integer scale, a rational reference gain, an integer count correction, and strictly positive parent weights---whose precise requirements are formalized in Definition~\ref{def:admissible} below.

Our main technical result shows that, whenever an order-$k$ clock is equipped with these additional data, one step of the construction produces an order-$(k+1)$ clock carrying the same type of data and hence ready for the next lift. Here $g$ denotes the new endpoint action, and $e_g$ is its standard basis vector.

\begin{theorem}[One-step construction]\label{thm:lift}
Every order-$k$ clock satisfying Definition~\ref{def:admissible}, $k\ge2$, on $m$ actions can be extended to an order-$(k+1)$ clock on $m+2k+3$ actions satisfying the same definition. The parent matrix is retained as a principal block. The new clock has endpoint counter $v_+=e_g^{\top}$, prefix margin $\delta_{0,+}=1$, remaining-height constant $\mathcal A_{k+1}=\mathcal A_k+1$, and reference gain $c_{{\rm ref},+}=1/2$.
\end{theorem}

The forward reference to Definition~\ref{def:admissible} encodes two additional algebraic ingredients: a correction satisfying the higher compatibility conditions, and a strictly positive representation of a target row by the parent count rows. Their precise formulation requires several coefficient tables associated with the parent bundle and the old-action score restoration. We introduce these tables in Section~\ref{sec:lift-tables}, state the additional hypotheses in Section~\ref{sec:admissible}, and then carry out the construction through Sections~\ref{sec:gateway}--\ref{sec:next}; Section~\ref{sec:chronology} verifies that the resulting fixed game generates the prescribed best-response trajectory and closes the induction.

\subsection{Parent and reference tables for the one-step theorem}\label{sec:lift-tables}
We now define the coefficient tables needed to state Definition~\ref{def:admissible} and formulate the reset equations. We first fix the widths and transfer matrices, then define the old-action scores and counts for an arbitrary gain. These are definitions needed for the reference-gain hypotheses in Section~\ref{sec:admissible}; we choose a compatible gain and prove validity later.

Fix an order-$k$ parent clock ($k\ge2$) on $m$ actions, with data $(A,R,C,b,v)$ from Definition~\ref{def:clock}, and write $w=r=k+1$. Here $w$ is the coefficient-table width, whereas $r$ is the number of actions in each reset group; these distinct roles happen to require the same integer. For brevity in the matrix formulas, set
\[
 \mathsf{T}:=\mathsf{T}_w,\qquad B:=B_w=\mathsf{T}_w^{-1},\qquad N:=N_w=\mathsf{T}_w-I_w,
 \qquad \erow{w}=(1,0,\ldots,0)\in\Q^{1\times w}.
\]
Thus, $\mathsf{T},B,N\in\Z^{w\times w}$ always act on the new width-$w$ tables in Sections~\ref{sec:gateway} and~\ref{sec:onestepproof}; the parent-width matrices and rows keep their explicit subscripts $N_k$ and $\erow{k}$. For any row $f$, the rows $f\mathsf{T},fB,fN$ represent respectively the forward shift, backward shift, and forward difference of $fp_k(q)$. The transfer matrices below evaluate a parent table at the bundle entrance and exit, and sum it over the bundled phases. For a positive integer scale factor $\tau\ge 1$, define the integer coefficient transfer matrices $E_\tau, E_{2\tau}, F_\tau\in \Z^{k \times w}$ by
\begin{equation}\label{eq:maps}
  p_{k-1}(\tau q) = E_\tau p_k(q), \qquad
  p_{k-1}(2\tau q) = E_{2\tau} p_k(q), \qquad
  \sum_{\sigma=\tau q}^{2\tau q-1} p_{k-1}(\sigma) = F_\tau p_k(q).
\end{equation}
These matrices are obtained by taking forward differences at $q=0$ of the integer-valued polynomials in~\eqref{eq:maps}.\footnote{For a polynomial $f$ of degree at most $r$, $f(q)=\sum_{j=0}^r(\Delta^j f)(0)\binom qj$, where $\Delta f(q)=f(q+1)-f(q)$ and $(\Delta^j f)(0)=\sum_{h=0}^j(-1)^{j-h}\binom jh f(h)$. Integer values $f(0),\ldots,f(r)$ therefore give integer coefficients.} Hence, their coefficients in the binomial basis are integers. For the sum, we use the identity
\[
 \sum_{\sigma=a}^{b-1}\binom \sigma j=\binom b{j+1}-\binom a{j+1}
 \qquad(a\le b).
\]
The polynomial $\binom{2\tau q}{j+1}-\binom{\tau q}{j+1}$ defines the left-hand side of~\eqref{eq:maps} also at nonpositive $q$, where we use this polynomial extension rather than a reversed finite sum. Because the sum is empty at $q=0$, column $0$ of $F_\tau$ is zero. Columns are indexed $0,1,\ldots,k$, so columns $1,\ldots,k$ are the second through last columns. They form a lower-triangular invertible $k\times k$ block with strictly positive main diagonal
\begin{equation}\label{eq:Fdiag}
  (F_\tau)_{j, j+1} = \tau^{j+1}(2^{j+1} - 1) > 0 \qquad (0 \le j < k).
\end{equation}
Indeed, row $j$ of $F_\tau$ represents a polynomial of degree $j+1$ with zero constant term. Its ordinary leading coefficient is $\tau^{j+1}(2^{j+1}-1)/(j+1)!$, which gives both the triangular support and~\eqref{eq:Fdiag}. For instance, at $k=2$, the choices $\tau=1$ and $\tau=2$ give
\[
 F_1=\begin{pmatrix}0&1&0\\0&1&3\end{pmatrix},\qquad
 F_2=\begin{pmatrix}0&2&0\\0&5&12\end{pmatrix}.
\]
For the second example, the sums of $1$ and $\sigma$ over $\sigma=2q,\ldots,4q-1$ are $2q$ and $6q^2-q=5q+12\binom q2$, respectively. The invertible blocks here are the last two columns, with diagonals $(1,3)$ and $(2,12)$.

\paragraph{Outer height and old-action tables.}
The parent bundle raises its own entry height from $\chi(\tau q)$ to $\chi(2\tau q)$. Indeed, it concatenates the parent phases indexed $\tau q,\ldots,2\tau q-1$, each ending at the next phase-entry score vector by Definition~\ref{def:clock}(I). We choose an outer height polynomial whose one-step increase equals that rise plus a positive gain $c$; this makes consecutive outer phases fit together while leaving room for the endpoint and reset actions. Its additive constant is free, and we fix the baseline $H=\mathcal A_k+3$ to exceed the parent's remaining-height allowance $\mathcal A_k$ by three, as needed in the later waiting estimates.

Given a rational height gain $c\in(0,1)$, let $\zetavec\in\Q^{1\times w}$ be a coefficient row, and write $\zeta(q):=\zetavec p_k(q)$ for its scalar evaluation. Define $\zetavec$ uniquely by its backward difference and its scalar value at $q=0$, and define the common parent-score translation row $\hvec\in\Q^{1\times w}$, as follows:
\begin{equation}\label{eq:y}
\begin{aligned}
 \zetavec(I-B)&=\chivec(E_{2\tau}-E_\tau)+c\erow{w}, & \zetavec p_k(0)&=H,\\
 \hvec&=\zetavec B-\chivec E_\tau. &&
\end{aligned}
\end{equation}
Evaluating~\eqref{eq:y} gives the scalar relations
\[
 \zeta(q)-\zeta(q-1)=\chi(2\tau q)-\chi(\tau q)+c,
 \qquad h(q)=\zeta(q-1)-\chi(\tau q).
\]
We likewise write $h(q):=\hvec p_k(q)\in\Q$ for the scalar common score translation represented by the row $\hvec$. Here $\zetavec$ and $\hvec$ are fixed coefficient rows, whereas $\zeta(q)$ and $h(q)$ are scalar functions. In particular, $\zeta(0)=\zetavec p_k(0)=\zetavec_0$, because $p_k(0)=e_0$. The right-hand side of the first equation in~\eqref{eq:y} has degree at most $k-1$, because $\chi$ has degree $k-1$, so it has a discrete antiderivative of degree at most $k$. The row solutions of $f(I-B)=0$ are precisely the constant-polynomial rows, and the value at zero therefore makes $\zetavec$ unique. Adjoin the endpoint action $g$ as action $m$ (so the old actions are $\{0, \ldots, m-1, g\}$ of total size $m+1$), and set
\begin{equation}\label{eq:old}
\begin{gathered}
 G_O=\begin{pmatrix}A&-(A_{b,:}+v)^{\top}\\A_{b,:}+v&0\end{pmatrix}
       \in\Q^{(m+1)\times(m+1)},\\
 X=\begin{pmatrix}CF_\tau\\\erow{w}\end{pmatrix}\in\Z^{(m+1)\times w},\qquad
 S=\begin{pmatrix}RE_\tau+\one_m\hvec\\\zetavec B-vCF_\tau+\delta\erow{w}\end{pmatrix}
       \in\Q^{(m+1)\times w},\\
 S_{\mathrm{mid}}=S+G_OX,\qquad
 D=S\mathsf{T}-S_{\mathrm{mid}}=SN-G_OX,\\
 0<\delta<\min\{c,\delta_0\}.
\end{gathered}
\end{equation}
The first $m$ rows of $X$ count the parent bundle, and its last row counts the single play of $g$. The first $m$ rows of $S$ give the parent entry scores plus the common translation $h(q)$; its last row places $g$ below $b$ by the remaining counter value minus $\delta$. Thus, we define the count table $X$ so that $Xp_k(q)$ is the vector of counts in the old-block word
\[
  \Aw(q) = \W_k(\tau q)\W_k(\tau q+1)\cdots\W_k(2\tau q-1)\, g.
\]
Moreover, $\rank(X) = w$ because $C F_\tau$ has rank $k$ on columns $1, \ldots, k$ while the endpoint row $X_{g,:}=\erow{w}=(1,0,\ldots,0)$ is nonzero only in column $0$. For the evaluated vectors in $\Q^{m+1}$, write
\[
 \mathbf s(q):=Sp_k(q),\qquad \mathbf x(q):=Xp_k(q),\qquad
 \mathbf s_{\mathrm{mid}}(q):=S_{\mathrm{mid}}p_k(q),\qquad
 \mathbf d(q):=Dp_k(q).
\]
The old-action scores start at $\mathbf s(q)$ and reach the intermediate state $\mathbf s_{\mathrm{mid}}(q)=\mathbf s(q)+G_O\mathbf x(q)$ immediately after $\Aw(q)$. The next phase must instead start at $\mathbf s(q+1)=S\mathsf{T}p_k(q)$. Thus, the reset word $\Bw(q)$ must supply exactly $\mathbf d(q)=\mathbf s(q+1)-\mathbf s_{\mathrm{mid}}(q)$.

For clarity, the four old-action tables have the same width $w=k+1$:
\begin{center}\small
\begin{tabular}{@{}cll@{}}\toprule
Table & Size & Meaning at parameter $q$\\\midrule
$X$ & $(m+1)\times w$ & counts in the parent bundle and the one play of $g$\\
$S$ & $(m+1)\times w$ & old scores before that word\\
$S_{\mathrm{mid}}$ & $(m+1)\times w$ & intermediate old scores after that word\\
$D$ & $(m+1)\times w$ & required change $\mathbf s(q+1)-\mathbf s_{\mathrm{mid}}(q)$\\\bottomrule
\end{tabular}
\end{center}

\subsection{The additional data required by the one-step theorem}\label{sec:admissible}
The clock conditions control the parent trajectory, but the lift also needs two algebraic inputs: a correction to the total reset counts that satisfies specified compatibility tests, and a strictly positive representation of a target height polynomial by the parent count polynomials. We define both inputs here to make the hypotheses of Theorem~\ref{thm:lift} precise. Their roles in the proof are described in the roadmap below.

Let $c_{\rm ref}\in(0,1)\cap\Q$ denote a rational \emph{reference gain}. It fixes the reusable correction and positive-weight data before the reset scale is chosen. The actual gain $c$ may later differ slightly to remove the lowest compatibility obstruction; the superscript $\circ$ always denotes evaluation at this fixed reference gain. For a proposed integer scale $\tau\ge1$, let $\zetavec^\circ\in\Q^{1\times w}$ and $D^\circ\in\Q^{(m+1)\times w}$ be the tables in~\eqref{eq:y}--\eqref{eq:old} with $c=c_{\rm ref}$, and put $\psi^\circ:=\zetavec^\circ B-\zetavec^\circ\in\Q^{1\times w}$. As before, $w=k+1$, $H=\mathcal A_k+3$, and $\chivec=R_{b,:}\in\Q^{1\times k}$. We use the same scalar polynomial already defined in~\eqref{eq:parenttables}:
\[
 \chi(\sigma)=\chivec p_{k-1}(\sigma)=(Rp_{k-1}(\sigma))_b
        =\max_i(Rp_{k-1}(\sigma))_i\qquad(\sigma\ge Q_k,\ \sigma\in\Z).
\]
In particular, $\chivec_{k-1}$ is the last coefficient of the fixed row $\chivec$, whereas $\chi(\sigma)$ is its evaluation at the phase index $\sigma$. Its ordinary leading coefficient is $\chivec_{k-1}/(k-1)!$.

We specify the compatibility-test matrices now; Lemma~\ref{lem:cokernel} in Section~\ref{sec:compatibilitytests} will prove that they form a complete set of tests. Write
\[
 \overline p_k(q):=\left(1,q,\frac{q^2}{2!},\ldots,\frac{q^k}{k!}\right)^{\top},
 \qquad p_k(q)=S_w\overline p_k(q),
\]
where $S_w\in\Q^{w\times w}$ is the lower-triangular change-of-basis matrix. Its diagonal entries are one. The subscript $\mathrm{DP}$ below stands for divided powers, the basis $q^i/i!$. These are \emph{scaled-monomial coordinates}: coefficient rows multiply $\overline p_k(q)$ instead of the binomial-basis vector $p_k(q)$. Let $N_0\in\Q^{w\times w}$ have entries $(N_0)_{i,i-1}=1$ for $1\le i\le k$ and zero elsewhere. For each even $h\in[k,2k]$, define
\begin{equation}\label{eq:cokernel}
 \begin{gathered}
 V^{(h)}_{ij}=(-1)^i\mathbf1_{\{i+j=h\}}\quad(0\le i,j\le k),\qquad
 Z_{\rm DP}^{(h)}=e^{-N_0/2}V^{(h)},\\
 Z^{(h)}=S_wZ_{\rm DP}^{(h)}S_w^{\top}\in\Q^{w\times w}.
 \end{gathered}
\end{equation}
The exponential is a finite rational polynomial in $N_0$, since $N_0^{k+1}=0$. Indeed, the derivative with respect to $q$ is $(\overline p_k)'(q)=N_0\overline p_k(q)$, so its finite Taylor formula makes the backward shift equal to $e^{-N_0}$. Thus, these matrices are explicit rational data determined by $k$, not additional unknowns in the construction.

For the positive-weight input, we want a weighted sum of old-block counts to equal the sum of two consecutive reference heights. This identity will produce the negative early and positive late reference couplings in~\eqref{eq:P0}. The endpoint row can supply only a constant; subtracting its weight isolates the nonconstant part that the parent counts must represent. Evaluating~\eqref{eq:y} at $q=0$ gives
$\zeta^\circ(-1)=\zetavec^\circ Bp_k(0)=H-c_{\rm ref}$.
Together with $\zeta^\circ(0)=H$, this gives constant entry $H+(H-c_{\rm ref})-(2H-c_{\rm ref})=0$ for $\zetavec^\circ+\zetavec^\circ B-(2H-c_{\rm ref})\erow{w}$. Since columns $1,\ldots,k$ of $F_\tau$ form an invertible block, there is a unique target row $t_{\tau,c_{\rm ref}}(\chivec)\in\Q^{1\times k}$ satisfying
\begin{equation}\label{eq:target}
 \begin{aligned}
 t_{\tau,c_{\rm ref}}(\chivec)F_\tau
   &=\zetavec^\circ+\zetavec^\circ B-(2H-c_{\rm ref})\erow{w},\\
 \bigl(t_{\tau,c_{\rm ref}}(\chivec)\bigr)_{k-1}
   &=\frac{2\chivec_{k-1}(2^{k-1}-1)}{\tau(2^k-1)}>0.
 \end{aligned}
\end{equation}
The last equality follows by comparing degree-$k$ coefficients and using~\eqref{eq:Fdiag}. Equation~\eqref{eq:target} only defines the target row; it does not assert that positive parent weights represent it. That additional requirement is item (ii) below. The argument $\chivec$ specifies the entire coefficient row, not the single number $\chi(\sigma)$ at a chosen phase index $\sigma$; together with $\tau,c_{\rm ref},H$, this row determines $\zetavec^\circ$ and hence the target.

\begin{definition}[Additional data for the next step]\label{def:admissible}
In addition to the clock data in Definition~\ref{def:clock}, assume that we have an integer $\tau\ge1$, a rational reference gain $c_{\rm ref}\in(0,1)$, an integer correction row $\ell_*\in\Z^{1\times w}$, and a strictly positive weight vector $a_P\in\Q_{>0}^m$ with the following properties:
\begin{enumerate}[label=(\roman*),leftmargin=*]
 \item The correction row satisfies the compatibility tests with indices above $k$:
 \begin{equation}\label{eq:higher}
  \ip{Z^{(h)}}{(D^\circ)^{\top}X-\ell_*^{\top}\psi^\circ}=0
  \qquad\text{for every even }k<h\le2k.
 \end{equation}
 \item The weights satisfy $a_P^{\top}C=t_{\tau,c_{\rm ref}}(\chivec)$, where the target row is uniquely defined by~\eqref{eq:target}. Equivalently, with the positive endpoint weight $a_g=2H-c_{\rm ref}$,
 \begin{equation}\label{eq:coneidentity}
  a=\begin{pmatrix}a_P\\2H-c_{\rm ref}\end{pmatrix}\in\Q_{>0}^{m+1},
  \qquad a^{\top}X=\zetavec^\circ+\zetavec^\circ B.
 \end{equation}
\end{enumerate}
\end{definition}
The equivalence in (ii) follows from $X=\col(CF_\tau,\erow{w})$ and~\eqref{eq:target}, since right multiplication by $F_\tau$ is injective on width-$k$ rows. Thus, the target is a strictly positive weighted sum of the old-block count polynomials. In Section~\ref{sec:synthesis}, this will make the leading coupling to the two reset groups equal to $[-a\one_r^{\top}\mid a\one_r^{\top}]$, with opposite signs in the two groups.

\paragraph{Proof roadmap for Theorem~\ref{thm:lift}.}
The input is an order-$k$ parent clock together with $(\tau,c_{\rm ref},\ell_*,a_P)$ from Definition~\ref{def:admissible}. The goal is an order-$(k+1)$ clock that satisfies both the trajectory conditions and this same additional-data definition, so another lift is possible. We prove the theorem across Sections~\ref{sec:gateway} and~\ref{sec:onestepproof}, as follows.
\begin{enumerate}[leftmargin=*]
 \item Section~\ref{sec:gateway} introduces the parent and old-action coefficient tables. Lemma~\ref{lem:gateway} verifies the endpoint action and derives the key identities for the required score change $D$; Lemmas~\ref{lem:resetidentity}--\ref{lem:interpolate} then derive the two reset equations and reduce exact score restoration to the matrix compatibility equation~\eqref{eq:moment}.
 \item Section~\ref{sec:moment} characterizes the compatibility tests in Lemma~\ref{lem:cokernel}, constructs the integer correction required by Definition~\ref{def:admissible}(i) in Lemma~\ref{lem:scale}, and verifies the additional RPS data in Proposition~\ref{prop:seedadmissible}.
 \item Section~\ref{sec:synthesis} chooses the actual gain $c$ and the total reset counts. If $k$ is even, then a small change from $c_{\rm ref}$ cancels the lowest test $h=k$; all higher tests stay satisfied. It then constructs $M,J,P$, solves the two reset equations, and proves the sign and waiting bounds.
 \item Section~\ref{sec:next} supplies the strictly positive weights required by Definition~\ref{def:admissible}(ii) at the next order and chooses the remaining count parameter so that the construction can be repeated.
 \item Section~\ref{sec:chronology} proves that every action in the proposed word is a best response, moves the phase boundary to the required entry action, and checks the new counter, prefix margin, and remaining-height bound. Section~\ref{sec:liftproof} gives the concluding proof of Theorem~\ref{thm:lift} by collecting these ingredients.
\end{enumerate}
The theorem therefore does not assume that solving the matrix equation alone proves best-response validity. We provide that final dynamical verification and pass the additional data to the next order later.

\subsection{Verifying the endpoint action}\label{sec:endpoint-proof}
The next lemma verifies that $g$ waits at every proper bundle prefix and wins exactly at the endpoint. It also identifies the intermediate scores after the single play of $g$ and shows that the required reset $D$ has no degree-$k$ term. The latter fact will remove one compatibility test automatically.

\begin{lemma}[Endpoint-action identities]\label{lem:gateway}
During the parent bundle $\W_k(\tau q)\cdots\W_k(2\tau q-1)$ starting from $\mathbf s(q)$, with remaining bundled parent counts $z^{\rm rem}$, the endpoint action score relative to $b$ satisfies
\begin{equation}\label{eq:gatewaygap}
  u_g - u_b = \delta - v z^{\rm rem}.
\end{equation}
Consequently, the endpoint action loses to the active parent maximum by at least $\delta_0 - \delta > 0$ at every proper prefix of the parent bundle, and exceeds the parent maximum by $\delta > 0$ at the exact bundle endpoint. After its single play, the parent maximum is $u_b = \zeta(q) - c$ and the endpoint action score is $u_g = \zeta(q) - c + \delta$. Moreover,
\begin{equation}\label{eq:top}
\begin{gathered}
 S_{b,:}=\zetavec B,\qquad (S_{\mathrm{mid}})_{b,:}=\zetavec-c\erow{w},\qquad
 (S_{\mathrm{mid}})_{g,:}=\zetavec-(c-\delta)\erow{w},\\
 S_{:,k}=(S_{\mathrm{mid}})_{:,k}=\zetavec_k\one_{m+1},\qquad D_{:,k}=0,
\end{gathered}
\end{equation}
where $\zetavec_k = \chivec_{k-1} \tau^{k-1}(2^{k-1} - 1) > 0$, and the required score change table $D$ is independent of $\delta$.
\end{lemma}
\begin{proof}
We first establish the invariant score difference~\eqref{eq:gatewaygap}, then use the parent's prefix condition to locate the handoff, and finally check the coefficient identities. At the start, the $b$ row of $S$ and the definition of $h$ in~\eqref{eq:y} give
\[
 u_b=\chi(\tau q)+h(q)=\zeta(q-1),\qquad
 u_g=\zeta(q-1)-vCF_\tau p_k(q)+\delta.
\]
All parent plays are still remaining, so $z^{\rm rem}=CF_\tau p_k(q)$; subtracting the two scores proves~\eqref{eq:gatewaygap} initially. A parent play $i$ adds $A_{bi}$ to $u_b$ and $A_{bi}+v_i$ to $u_g$, and replaces $z^{\rm rem}$ by $z^{\rm rem}-e_i$. Hence the score difference after one play is
\[
 (u_g-u_b)^{\rm new}=(u_g-u_b)+v_i
 =\delta-v(z^{\rm rem}-e_i),
\]
which preserves the identity~\eqref{eq:gatewaygap} with the updated remaining count vector. At a proper prefix, let $\sigma$ be the index of the current parent phase, so $\tau q\le \sigma<2\tau q$. At an internal phase boundary, use the next phase with its empty prefix. Let $z_\sigma^{\rm rem}$ count the plays remaining in this phase alone. There are $2\tau q-1-\sigma$ later complete phases. Each contributes $vCp_{k-1}(s)=\erow{k}p_{k-1}(s)=1$, so
\[
 vz^{\rm rem}=vz_\sigma^{\rm rem}+2\tau q-1-\sigma.
\]
Substitute this decomposition into~\eqref{eq:gatewaygap}. The parent prefix condition~\eqref{eq:marker} therefore gives
\[
 H_{\rm parent}-u_g
 =\bigl(H_{\rm parent}-u_b+vz_\sigma^{\rm rem}\bigr)
   +(2\tau q-1-\sigma)-\delta
 \ge\delta_0-\delta>0.
\]
This uses the parent maximum, without requiring $b$ to remain a maximizer or $vz_\sigma^{\rm rem}$ to be positive.
At the bundle endpoint, $z^{\rm rem}=0$. The consecutive parent phase identities take the parent scores from $\mathbf r(\tau q)+\one_mh(q)$ to $\mathbf r(2\tau q) + \one_m h(q)$, whose maximum is attained at $b$ by the parent entry-leadership condition. Its value is $\chi(2\tau q)+h(q)=\zeta(q)-c$ by the scalar height equation following~\eqref{eq:y}. Setting $z^{\rm rem}=0$ in~\eqref{eq:gatewaygap} then gives $u_g=u_b+\delta=\zeta(q)-c+\delta$.

Playing $g$ once adds $(G_O)_{0:m, g} = A e_b - v^{\top}$ to the parent scores. Because $\W_k(2\tau q)$ begins with at least two plays of $b$, action $b$ remains a maximizer of $\mathbf r(2\tau q) + A e_b$. For actions outside $\supp(v)$, the update agrees with that first play of $b$. On $\supp(v)$, apply~\eqref{eq:supportsafe} at parent phase $2\tau q$: it states that each coordinate of $\mathbf r(2\tau q)+Ae_b-v^{\top}$ on that support is strictly below its $b$ coordinate. Adding the common translation $h(q)$ preserves these comparisons. Since $A_{bb} = v_b = 0$, the score of $b$ remains $\zeta(q) - c$. To check the top column, telescope the parent differences over the bundle:
\[
 N_kF_\tau=E_{2\tau}-E_\tau,\qquad
 ACF_\tau=R(E_{2\tau}-E_\tau).
\]
The evaluated entrance and exit tables have degree at most $k-1$, so the last column of each right-hand side is zero. Also $vCF_\tau=\erow{k}F_\tau$ represents $\tau q$, of degree $1<k$. The contribution from the single play of $g$ is constant. Thus $(G_OX)_{:,k}=0$ in both the parent and endpoint rows. The parent table $RE_\tau$ has degree at most $k-1$, while $\hvec$ and $\zetavec B$ both have degree-$k$ coefficient
$\zetavec_k=\chivec_{k-1}\tau^{k-1}(2^{k-1}-1)>0$.
Consequently all old actions share this top coefficient before and after the old word. A shift preserves the top coefficient, so $D=S\mathsf{T}-S_{\mathrm{mid}}$ has zero top column. Changing $\delta$ shifts $S_{g,:}$ by a constant multiple of $\erow{w}$, which satisfies $\erow{w} N = 0$ and leaves $D = SN - G_O X$ unchanged.
\end{proof}

\subsection{Restoring both score blocks with one coupling matrix}\label{sec:reset-equations}
After the old-block word $\Aw(q)$, the old scores are at $\mathbf s_{\mathrm{mid}}(q)$ and must be changed to $\mathbf s(q+1)$. The reset actions also need their own next-entry scores: otherwise the same reset schedule need not work again in phase $q+1$. We therefore impose simultaneous restoration of both blocks, using one coupling matrix $P$ whose transpose controls the reverse interactions by skew-symmetry.

For this lift, choose $n=2r=2(k+1)$ reset actions. This is the number used by our two-group construction. Let $M\in\Z^{n\times w}$ be their as-yet-undetermined count table, let $J\in\Q^{n\times n}$ be a strictly lower-triangular interaction matrix, and let $P\in\Q^{(m+1)\times n}$ be the coupling from reset actions to old actions. The desired reset-entry state after the old word is $\one_n\zetavec-JM$: earlier runs will raise each later reset action to the common ceiling $\zeta(q)$. At the end of all runs it becomes $\one_n\zetavec-J^{\top}M$, which must be $V\mathsf{T}$. This determines $V$; the old word changes reset scores by $-P^{\top}X$, determining $Q$. Define the skew-symmetric block $K$ for interactions among reset actions, the full lifted skew-symmetric matrix $G$, the initial reset-action score table $V$, and the table of score changes from old actions to reset actions $Q$ by
\begin{equation}\label{eq:controllerdata}
\begin{gathered}
  K = J - J^{\top}, \qquad G = \begin{pmatrix}G_O & P \\ -P^{\top} & K\end{pmatrix}, \\
  V = \one_n \zetavec B - J^{\top} M B, \qquad Q = \one_n(\zetavec B - \zetavec ) + JM - J^{\top} M B.
\end{gathered}
\end{equation}
Here $K\in\Q^{n\times n}$, $G\in\Q^{(m+1+n)\times(m+1+n)}$, and $V,Q\in\Q^{n\times w}$. Write $\mathbf m(q):=Mp_k(q)\in\Z^n$ for the reset-count vector, $\mathbf v(q):=Vp_k(q)\in\Q^n$ for the reset-entry score vector, and $\mathbf u(q):=Up_k(q)\in\Q^{m+1+n}$ for the full entry score vector. The prescribed word of reset actions is $\Bw(q) = \reset{0}^{\mathbf m_0(q)} \reset{1}^{\mathbf m_1(q)} \cdots \reset{n-1}^{\mathbf m_{n-1}(q)}$, and the full phase-entry score table of the unrotated outer phase $\Aw(q)\Bw(q)$ is $U = \col(S, V)$.

The entries of $M,J,P$ will be constructed in Section~\ref{sec:synthesis}. One requirement is $J_{ij}>0$ for every $i>j$: a play of an earlier reset action then raises the score of a later reset action, whereas a play of the later one lowers the earlier action's score. Together with the entry scores in~\eqref{eq:controllerdata}, positive run counts, and the coupling bounds proved later, these signs must make each $\reset{i}$ remain a maximizer for exactly $\mathbf m_i(q)$ plays and hand off to $\reset{i+1}$ in the prescribed order. Proposition~\ref{prop:chronology} verifies this best-response property, including competition from the old actions.

For the proposed score tables, the simultaneous restoration requirement is precisely the following pair of linear equations for $P$:
\begin{equation}\label{eq:twosided}
  PM = D, \qquad P^{\top} X = Q.
\end{equation}
We state explicitly what these equations accomplish before solving them.
\begin{lemma}[The two coupling equations give the next phase-entry scores]\label{lem:resetidentity}
With the tables in~\eqref{eq:old} and~\eqref{eq:controllerdata}, equations~\eqref{eq:twosided} hold if and only if the combined word with count table $\col(X,M)$ changes the full score table from $U=\col(S,V)$ to $U\mathsf{T}$. Equivalently,
\begin{equation}\label{eq:fullreset}
 G\begin{pmatrix}X\\M\end{pmatrix}=UN.
\end{equation}
\end{lemma}
We note that~\eqref{eq:fullreset} is an identity of coefficient tables; whether the word is actually played is a separate best-response question.
\begin{proof}
We verify the old and reset blocks separately. After the old-block word, the old scores are $S_{\mathrm{mid}}=S+G_OX$. If $PM=D$, then the reset word changes them to
\[
 S_{\mathrm{mid}}+PM=S_{\mathrm{mid}}+D=S\mathsf{T}.
\]
Here the last equality is the definition $D=S\mathsf{T}-S_{\mathrm{mid}}$ in~\eqref{eq:old}, so the reset word supplies exactly the missing old-score change.
For the reset block, $P^{\top}X=Q$ and~\eqref{eq:controllerdata} give the intermediate scores
\[
 V-P^{\top}X=V-Q=\one_n\zetavec-JM.
\]
The reset word then gives
\[
 (\one_n\zetavec-JM)+(J-J^{\top})M
 =\one_n\zetavec-J^{\top}M=V\mathsf{T},
\]
where the last identity uses $B\mathsf{T}=I_w$. Thus, both blocks equal their next-entry tables, proving~\eqref{eq:fullreset}.
Conversely, the old block of~\eqref{eq:fullreset} gives $PM=SN-G_OX=D$. Its reset block gives
\[
 P^{\top}X=V+(J-J^{\top})M-V\mathsf{T}
 =\one_n(\zetavec B-\zetavec)+JM-J^{\top}MB=Q.
\]
Hence the exact full reset is equivalent to both coupling equations.
\end{proof}

The next lemma isolates the only algebraic obstruction to finding such a coupling $P$: both equations must give the same matrix when paired with the opposite count table. Full column rank makes this condition sufficient as well as necessary.

\begin{lemma}[Solving $PM=D$ and $P^{\top}X=Q$ simultaneously]\label{lem:interpolate}
Suppose $X \in \Q^{(m+1)\times w}$ and $M \in \Q^{n \times w}$ have full column rank $w$. The two equations~\eqref{eq:twosided} are solvable for $P \in \Q^{(m+1)\times n}$ if and only if the compatibility identity holds:
\begin{equation}\label{eq:work}
  D^{\top} X = M^{\top} Q.
\end{equation}
Moreover, let $L_X\in\Q^{(m+1)\times w}$ and $R_M\in\Q^{w\times n}$ be one-sided inverse matrices satisfying $X^{\top}L_X=I_w$ and $R_MM=I_w$. Then, for any reference matrix $P_0\in\Q^{(m+1)\times n}$, a solution to~\eqref{eq:twosided} is given explicitly by
\begin{equation}\label{eq:interpolate}
\begin{aligned}
 D_R&=D-P_0M\in\Q^{(m+1)\times w},\qquad
 Q_R=Q-P_0^{\top}X\in\Q^{n\times w},\\
 P&=P_0+D_RR_M+L_X\bigl(Q_R^{\top}-X^{\top}D_RR_M\bigr).
\end{aligned}
\end{equation}
\end{lemma}
\begin{proof}
We prove first that compatibility is necessary and then check that formula~\eqref{eq:interpolate} solves both equations. Its first correction $D_RR_M$ fixes $PM=D$. Its second fixes $X^{\top}P=Q^{\top}$ without changing $PM$, because compatibility makes its product with $M$ zero. If $PM=D$ and $P^{\top}X=Q$, then
\[
 D^{\top}X=(PM)^{\top}X=M^{\top}(P^{\top}X)=M^{\top}Q,
\]
which is~\eqref{eq:work}. Conversely, assume~\eqref{eq:work}. Transposing it and using the residual definitions gives
\[
 Q_R^{\top}M=(Q^{\top}-X^{\top}P_0)M
 =X^{\top}(D-P_0M)=X^{\top}D_R.
\]
Consequently, right-multiplying~\eqref{eq:interpolate} by $M$ yields
\[
 PM=P_0M+D_R+L_X(Q_R^{\top}M-X^{\top}D_R)=P_0M+D_R=D.
\]
For the equation $P^{\top}X=Q$ in~\eqref{eq:twosided}, $L_X^{\top}X=I_w$ gives
\[
 \begin{aligned}
 P^{\top}X
 &=P_0^{\top}X+R_M^{\top}D_R^{\top}X
   +(Q_R-R_M^{\top}D_R^{\top}X)L_X^{\top}X\\
 &=P_0^{\top}X+Q_R=Q.
 \end{aligned}
\]
Thus, the same matrix $P$ satisfies both equations, establishing sufficiency and the claimed formula.
\end{proof}
For definiteness, we always choose $L_X=X(X^{\top}X)^{-1}$, which is rational because $X$ has full column rank. The left inverse $R_M$ will be specified by the construction of the reset actions.

Write $\ell=\one_n^{\top}M\in\Z^{1\times w}$ for the total reset-action count row and $E=M^{\top}JM\in\Q^{w\times w}$ for the aggregate interaction matrix. Define the coefficient row $\psi:=\zetavec B-\zetavec\in\Q^{1\times w}$ and its scalar evaluation $\psi(q):=\psi p_k(q)=\zeta(q-1)-\zeta(q)$, the change in the height polynomial when the phase index decreases by one. Substituting $Q=\one_n\psi+JM-J^{\top}MB$ gives
\[
 M^{\top}Q=\ell^{\top}\psi+M^{\top}JM-M^{\top}J^{\top}MB
 =\ell^{\top}\psi+E-E^{\top}B.
\]
Thus, under the full-column-rank hypotheses of Lemma~\ref{lem:interpolate}, solvability of~\eqref{eq:twosided} is equivalent to the $w\times w$ linear matrix equation
\begin{equation}\label{eq:moment}
  \LL(E) := E - E^{\top} B = D^{\top} X - \ell^{\top} \psi.
\end{equation}
Equation~\eqref{eq:moment} is the \emph{matrix compatibility equation}: solving it ensures that the required changes to old scores and reset-action scores agree with skew-symmetry. In Sections~\ref{sec:moment} and~\ref{sec:synthesis}, we prove that suitable count and interaction tables solve this equation with the required positive reset interactions. Section~\ref{sec:chronology} then uses these signs and the coupling bounds to prove that every prescribed action in each outer phase is a best response.

\section{Proof of the one-step theorem}\label{sec:onestepproof}
\subsection{Compatibility tests and the integer correction}\label{sec:moment}
Section~\ref{sec:gateway} reduced exact score restoration to the matrix compatibility equation
\[
 \LL(E):=E-E^{\top}B=D^{\top}X-\ell^{\top}\psi.
\]
In this section, we characterize exactly when this equation is solvable and use that characterization to construct the integer scale and correction required by Definition~\ref{def:admissible}(i). Lemma~\ref{lem:cokernel} gives a complete family of scalar compatibility tests for $\LL(E)=F$, and Lemma~\ref{lem:scale} chooses the scale $\tau$ and correction row $\ell_*$ so that every test with index strictly above $k$ is satisfied. This leaves at most the lowest-index test $h=k$, which will be removed by the gain adjustment in Section~\ref{sec:synthesis}. At this stage we address only algebraic compatibility; payoff signs and best-response validity are handled later.

Section~\ref{sec:compatibilitytests} derives the complete family of compatibility tests, Section~\ref{sec:scale} constructs the integer scale and correction satisfying the higher-index tests, and Section~\ref{sec:rpsadmissible} verifies Definition~\ref{def:admissible} for the RPS clock in Proposition~\ref{prop:seedadmissible}.

Recall that $w=k+1$ is the coefficient-table width, $B=\mathsf{T}_w^{-1}$ is the backward shift, and
\[
 \LL(E)=E-E^{\top}B
\]
is the operator in~\eqref{eq:moment}. The lower-triangular matrix $S_w$ changes scaled-monomial coordinates to binomial coordinates:
\[
 p_k(q)=S_w(1,q,q^2/2!,\ldots,q^k/k!)^{\top}.
\]
The tables $D,X\in\Q^{(m+1)\times w}$ encode the required old-score change and old-word counts in~\eqref{eq:old}. Below, $F$ denotes a generic $w\times w$ right-hand side, distinct from the transfer matrix $F_\tau$ in~\eqref{eq:maps}.

\subsubsection{Why the listed linear tests characterize solvability}\label{sec:compatibilitytests}
We now characterize the image of $\LL$, and hence all obstructions to solving $\LL(E)=F$. For comparison, if the backward shift $B$ were replaced by the identity, the equation would become $E-E^{\top}=F$, so $F$ would have to be skew-symmetric, imposing $\binom{k+2}{2}$ scalar conditions. The backward shift changes these dependencies substantially: Lemma~\ref{lem:cokernel} will show that the actual operator $\LL$ has only $\lfloor k/2\rfloor+1$ independent compatibility tests. We identify them by computing the kernel of the adjoint operator.

For the matrix inner product $\ip{U}{V}=\sum_{i,j}U_{ij}V_{ij}$, define the adjoint $\LL^*$ by $\ip{Z}{\LL(E)}=\ip{\LL^*(Z)}{E}$. Here $\LL^*(Z)=Z-BZ^{\top}$. Finite-dimensional linear algebra gives $\operatorname{im}(\LL)=\ker(\LL^*)^{\perp}$, hence
\[
 F\in\operatorname{im}(\LL)
 \quad\Longleftrightarrow\quad
 \ip{Z}{F}=0\ \text{for every }Z\text{ with }\LL^*(Z)=0.
\]
Indeed, these $Z$ are precisely the matrices orthogonal to every output of $\LL$; taking the orthogonal complement twice recovers its output space. Therefore, to prove that $\LL(E)=F$ has a solution, it is enough to find a basis of $\ker(\LL^*)$ and show that the weighted sum $\sum_{i,j}Z_{ij}F_{ij}$ vanishes for each basis matrix.

In the scaled-monomial coordinates defined in Section~\ref{sec:admissible}, the backward shift $q\mapsto q-1$ acts by the matrix exponential $e^{-N_0}$, because
\[
 \frac{(q-1)^i}{i!}=\sum_{j=0}^i\frac{(-1)^{i-j}}{(i-j)!}\frac{q^j}{j!},
 \qquad e^{-N_0}=\sum_{j=0}^k\frac{(-N_0)^j}{j!}.
\]
The next lemma shows that the explicit matrices in~\eqref{eq:cokernel} form such a basis. In particular, there are no additional hidden consistency requirements.

\begin{lemma}[All linear compatibility tests for $\LL$]\label{lem:cokernel}
In scaled-monomial coordinates, the matrices $Z_{\rm DP}^{(h)}$ in~\eqref{eq:cokernel}, indexed by even $h\in[k,2k]$, form a basis of the kernel of $\LL_{\rm DP}^*(Z)=Z-e^{-N_0}Z^{\top}$. In binomial coordinates, the corresponding basis is $Z^{(h)}=S_wZ_{\rm DP}^{(h)}S_w^{\top}$. Consequently, $\LL(E)=F$ is solvable over $\Q$ if and only if $\ip{Z^{(h)}}{F}=0$ for every even integer $h\in[k,2k]$.
\end{lemma}
\begin{proof}
We first reduce the adjoint-kernel equation to a symmetry condition and a recurrence along anti-diagonals. We then solve that recurrence and translate the resulting basis back to binomial coordinates. In scaled-monomial coordinates, $\LL_{\rm DP}^*(Z_{\rm DP}) = Z_{\rm DP} - e^{-N_0} Z_{\rm DP}^{\top}$. Any solution satisfies $Z_{\rm DP}^{\top} = e^{N_0} Z_{\rm DP}$ and hence $Z_{\rm DP} = e^{-N_0} Z_{\rm DP} e^{-N_0^{\top}}$. Define the linear operator $\mathcal D:\Q^{w\times w}\to\Q^{w\times w}$ by
\[
 \mathcal D(Z):=N_0Z+ZN_0^{\top}.
\]
Left multiplication by $N_0$ and right multiplication by $N_0^{\top}$ commute, and $\mathcal D^{2k+1}=0$ because $N_0^{k+1}=0$. The operator exponential is therefore the finite sum
\[
 e^{-\mathcal D}:=\sum_{j=0}^{2k}\frac{(-\mathcal D)^j}{j!},
 \qquad e^{-\mathcal D}(Z)=e^{-N_0}Ze^{-N_0^{\top}}.
\]
 Moreover, $I-e^{-\mathcal D}=\mathcal D \widetilde P(\mathcal D)$, where $\widetilde P(x)=\sum_{j=1}^{2k}(-1)^{j+1}x^{j-1}/j!$ has constant term one. Thus, $\widetilde P(\mathcal D)$ is invertible: the operator $\widetilde P(\mathcal D)$ has the form $I+\mathcal K$, where $\mathcal K$ is a polynomial in $\mathcal D$ with zero constant term. This $\mathcal K$ is nilpotent, and the inverse is the finite sum $(I+\mathcal K)^{-1}=\sum_{j=0}^{2k}(-\mathcal K)^j$. It follows that $e^{-\mathcal{D}}(Z_{\rm DP}) = Z_{\rm DP}$ is equivalent to $\mathcal{D}(Z_{\rm DP}) = 0$.

Setting $\widetilde V = e^{N_0/2} Z_{\rm DP}$ (left multiplication by $e^{N_0/2}$ commutes with $\mathcal{D}$), the two conditions $Z_{\rm DP}^{\top} = e^{N_0} Z_{\rm DP}$ and $\mathcal{D}(Z_{\rm DP}) = 0$ become
\[
  \widetilde V^{\top} = \widetilde V \qquad \text{and} \qquad N_0 \widetilde V + \widetilde V N_0^{\top} = 0.
\]
For the symmetry step, $\mathcal D(Z_{\rm DP})=0$ implies $Z_{\rm DP}e^{N_0^{\top}/2}=e^{-N_0/2}Z_{\rm DP}$. Hence
\[
 \widetilde V^{\top}=Z_{\rm DP}^{\top}e^{N_0^{\top}/2}
 =e^{N_0}Z_{\rm DP}e^{N_0^{\top}/2}
 =e^{N_0/2}Z_{\rm DP}=\widetilde V.
\]
In coordinates, $N_0 \widetilde V + \widetilde V N_0^{\top} = 0$ reads $\widetilde V_{i-1,j} + \widetilde V_{i,j-1} = 0$ for all $0 \le i,j \le k$ (with $\widetilde V_{-1,j} = \widetilde V_{i,-1} = 0$). On any anti-diagonal $i + j = h < k$, the equation at $(i,j)=(0,h+1)$ forces $\widetilde V_{0,h}=0$ and inductively $\widetilde V_{i,h-i} = 0$. For $h\ge k$, that boundary equation is outside the matrix. The recurrence instead relates successive entries on the same anti-diagonal by a change of sign, with one free scalar. Thus, on each anti-diagonal $i + j = h \in [k, 2k]$, $\widetilde V_{i,h-i} = (-1)^i c_h$ is determined up to a scalar, and symmetry $\widetilde V_{i,h-i} = \widetilde V_{h-i,i}$ holds if and only if $(-1)^i = (-1)^{h-i}$, i.e., $h$ is even.

The symmetric solutions are therefore exactly the span of the $V^{(h)}$ in~\eqref{eq:cokernel}; they are linearly independent because they occupy distinct anti-diagonals. The preceding transformations are reversible: if $\widetilde V^{\top}=\widetilde V$ and $N_0\widetilde V+\widetilde V N_0^{\top}=0$, then $Z_{\rm DP}=e^{-N_0/2}\widetilde V$ satisfies $Z_{\rm DP}=e^{-N_0}Z_{\rm DP}^{\top}$. Hence these matrices give a basis of the full adjoint kernel.

The shift relation is $BS_w=S_we^{-N_0}$. Under the changes of coordinates $E_{\rm DP}=S_w^{\top}ES_w$ and $F_{\rm DP}=S_w^{\top}FS_w$, the matrix equation $E-E^{\top}B=F$ becomes $E_{\rm DP}-E_{\rm DP}^{\top}e^{-N_0}=F_{\rm DP}$. Hence, we have $\ip{Z_{\rm DP}^{(h)}}{F_{\rm DP}} = \ip{S_w Z_{\rm DP}^{(h)} S_w^{\top}}{F} = \ip{Z^{(h)}}{F}$. The orthogonality criterion preceding the lemma now gives the stated necessary and sufficient tests. All coefficients are rational, so a consistent system has a rational solution by Gaussian elimination.
\end{proof}

Left multiplication by the lower-triangular matrices $e^{-N_0/2}$ and $S_w$ can only increase a row index; right multiplication by $S_w^{\top}$ can only increase a column index. Starting from $i+j=h$ therefore gives support only where $i+j\ge h$. All three factors have diagonal entries one, so the entries on the original anti-diagonal retain their values $(-1)^i$.

For the right-hand side $F=D^{\top}X-\ell^{\top}\psi$ of~\eqref{eq:moment}, the top test vanishes automatically, the intermediate tests determine the correction row, and the possible lowest test determines the gain adjustment. We verify these three claims in turn. First, $D_{:,k}=0$ by~\eqref{eq:top}, and $\psi_k=0$ because $\psi$ is a difference of degree-$k$ polynomials. Thus,
\[
 F_{kk}=(D_{:,k})^{\top}X_{:,k}-\ell_k\psi_k=0.
\]
However, the rest of column $k$ need not vanish. The top test $h=2k$ is supported solely at $(k,k)$, so this one entry makes that test automatic. Second, for each even $h\in(k,2k)$, the support of $Z^{(h)}$ is contained in $\{(i,j):i+j\ge h\}$, which excludes the first row and column. We satisfy these higher-index tests in Section~\ref{sec:scale}. Finally, when $k$ is even, the lowest test $h=k$ has corner entries $Z^{(k)}_{0,k}=Z^{(k)}_{k,0}=1$; this is the single test left for the gain adjustment in Section~\ref{sec:synthesis}. When $k$ is odd there is no test with index $h=k$.

\subsubsection{Choosing the integer scale and correction for tests indexed above $k$}\label{sec:scale}
We now construct the integer data required in Definition~\ref{def:admissible}(i). The top test $h=2k$ is already automatic, and we will solve the other tests with even indices $k<h<2k$ one at a time. The potential lowest test $h=k$ is intentionally left to the gain adjustment in Section~\ref{sec:synthesis}. Keep the rational reference gain $c_{\rm ref}$ fixed, with reference tables $\zetavec^\circ,D^\circ,\psi^\circ$ as defined in Sections~\ref{sec:lift-tables} and~\ref{sec:admissible}.

\begin{lemma}[Choosing the integer scale factor]\label{lem:scale}
For the algebraic parent data $(R^{\top}C,\one_m^{\top}C,\chivec,vC=\erow{k},c_{\rm ref})$, there exist an integer scale $\tau\ge1$ and an integer correction row $\ell_*\in\Z^{1\times w}$ satisfying~\eqref{eq:higher}. They can be computed in finitely many steps using exact rational arithmetic. Moreover, this choice of $(\tau,\ell_*)$ does not depend on row $0$ of $R^{\top}C$ or on the constant coefficient $\chivec_0$.
\end{lemma}
\begin{proof}
We solve first for rational correction coefficients with the scale $\tau$ left symbolic, and then choose an integer $\tau$ that clears all denominators. The support pattern of the tests makes the coefficient equations triangular. Define
\[
  L_\tau= E_\tau \mathsf{T}-E_{2\tau}, \qquad f=c_{\rm ref}\erow{w}-\chivec L_\tau.
\]
To make the dependence on $\tau$ explicit, first expand the parent and endpoint rows of $D^\circ=S^\circ N-G_OX$. The bundle consists of phases $\sigma=\tau q,\ldots,2\tau q-1$. Summing their successive differences gives
\[
 \begin{aligned}
 N_kF_\tau p_k(q)
 &=\sum_{\sigma=\tau q}^{2\tau q-1}
       \bigl(p_{k-1}(\sigma+1)-p_{k-1}(\sigma)\bigr)\\
 &=p_{k-1}(2\tau q)-p_{k-1}(\tau q)
 =(E_{2\tau}-E_\tau)p_k(q).
 \end{aligned}
\]
Equality for all positive integers $q$ implies the coefficient identity $N_kF_\tau=E_{2\tau}-E_\tau$. Using $AC=RN_k$ and $vC=\erow{k}$ then gives
\[
 \begin{aligned}
 D^\circ_{0:m,:}&=RL_\tau+\one_m f+(A_{b,:}+v)^{\top}\erow{w},\\
 D^\circ_{g,:}&=c_{\rm ref}\erow{w}-\erow{k}F_\tau \mathsf{T}.
 \end{aligned}
\]
Here $D^\circ_{0:m,:}$ denotes the first $m$ rows. Multiplying by $X=\col(CF_\tau,\erow{w})$ and using $A_{b,:}C=\chivec N_k$ gives
\begin{equation}\label{eq:summary}
\begin{split}
  (D^\circ)^{\top} X ={}& L_\tau^{\top} (R^{\top} C) F_\tau+ f^{\top} (\one_m^{\top} C) F_\tau\\
  &+ (\erow{w})^{\top}\bigl\{\chivec(E_{2\tau} - E_\tau) + \erow{k} F_\tau\bigr\}
  + \bigl\{c_{\rm ref}\erow{w} - \erow{k} F_\tau \mathsf{T}\bigr\}^{\top} \erow{w}.
\end{split}
\end{equation}
Use a hat to denote the corresponding scaled-monomial coordinates:
\[
 \widehat F=S_w^{\top}(D^\circ)^{\top}XS_w,\qquad
 \widehat\psi=\psi^\circ S_w,\qquad \widehat\ell=\ell_*S_w.
\]
For $i,j\ge1$, the entry $\widehat F_{ij}$ is a rational polynomial in $\tau$ divisible by $\tau^{i+j}$. To see this directly, a parent monomial $\sigma^a$ contributes to $L_\tau$ the polynomial
\[
 \tau^a\bigl((q+1)^a-(2q)^a\bigr).
\]
Its coefficient of $q^i$ is zero for $a<i$ and otherwise divisible by $\tau^i$. Similarly, summing a parent polynomial over the bundle gives a polynomial antiderivative evaluated at $2\tau q$ and $\tau q$, whose coefficient of $q^j$ is divisible by $\tau^j$. These facts apply to the first two products in~\eqref{eq:summary}. The last two products remain supported in row $0$ or column $0$, because $\erow{w}S_w=\erow{w}$ and $S_w^{\top}(\erow{w})^{\top}=(\erow{w})^{\top}$. They contribute nothing for $i,j\ge1$.

Since $\psi^\circ=-\chivec(E_{2\tau}-E_\tau)-c_{\rm ref}\erow{w}$, each $\widehat\psi_j$ for $j\ge1$ is a fixed rational multiple of $\tau^j$, with
\[
 \widehat\psi_k=0,\qquad
 \widehat\psi_{k-1}=-\chivec_{k-1}\tau^{k-1}(2^{k-1}-1)\ne0.
\]
Initialize every coordinate of $\widehat\ell$ to zero. For each even $h\in(k,2k)$, in decreasing order, set $i=h-k+1$ and solve the $h$th test by
\begin{equation}\label{eq:scalepivot}
 (-1)^i\widehat\psi_{k-1}\widehat\ell_i
 =\ip{Z_{\rm DP}^{(h)}}{\widehat F}
 -\sum_{\substack{i<a\le k\\0\le b\le k-1}}
 (Z_{\rm DP}^{(h)})_{ab}\widehat\ell_a\widehat\psi_b.
\end{equation}
These equations determine the selected coefficients one at a time, in decreasing index order. The coefficient of $\widehat\ell_i$ in~\eqref{eq:scalepivot} is $(-1)^i\widehat\psi_{k-1}\ne0$. Coefficients with index below $i$ cannot occur because every nonzero test entry has $a+b\ge h$. Here $b\le k-1$ because $\widehat\psi_k=0$; at the smallest possible row $a=i=h-k+1$, only $b=k-1$ is possible, and that anti-diagonal entry is $(-1)^i$. All higher-index coefficients are already fixed. Later choices have smaller indices and therefore cannot alter a test already solved. Inductively, each nonzero summand on the right is divisible by $\tau^{a+b}$ and hence by $\tau^h$; the first term has the same divisibility because $\widehat F_{ab}$ is divisible by $\tau^{a+b}$ and the test has support only where $a+b\ge h$, as shown above. Dividing by this nonzero coefficient, a constant multiple of $\tau^{k-1}$, gives a rational polynomial $\widehat\ell_i(\tau)$ divisible by $\tau^{h-k+1}=\tau^i$. The unselected coordinates, including $0$, $1$, and $k$, remain zero.

To make the choice explicit, write the solved polynomials in a formal variable $x$ as $\widehat\ell_i(x)=\sum_{r\ge i}a_{ir}x^r$, with $a_{ir}\in\Q$. Define
\begin{equation}\label{eq:scalechoice}
 D_*:=\operatorname{lcm}_{i,r}\operatorname{den}(a_{ir}),
 \qquad \tau:=D_*k!,
\end{equation}
Here $\operatorname{lcm}$ means least common multiple, $\operatorname{den}(a)$ is the positive denominator of $a$ in lowest terms, and the least common multiple of an empty collection is $1$. This is the scale factor used by the general algorithm, fixed solely from the parent data. For this choice and each $i\le k$, $\widehat\ell_i(\tau)$ is an integer multiple of $(k!)^i$ and hence a multiple of $i!$, so the ordinary polynomial coefficients of $\ell_*$ are integers, and its binomial coefficients are integers by the finite-difference formula following~\eqref{eq:maps}. Thus $\ell_*\in\Z^{1\times w}$. Finally, $\erow{k} L_\tau= 0$ (since $\erow{k} E_\tau \mathsf{T}=\erow{w} \mathsf{T}=\erow{w}=\erow{k} E_{2\tau}$) and $\chivec_0$ cancels from $\chivec(E_{2\tau} - E_\tau)$ and $\chivec L_\tau$, proving independence of row $0$ of $R^{\top} C$ and of $\chivec_0$. The explicit RPS witness in Proposition~\ref{prop:seedadmissible} uses $\tau=1$ instead of this general denominator-clearing choice.
\end{proof}

\subsubsection{Verifying the additional data for the RPS base}\label{sec:rpsadmissible}
To apply Theorem~\ref{thm:lift} for the first time, from order $2$ to order $3$, the RPS clock must supply both ingredients of Definition~\ref{def:admissible}: integer scale and correction data satisfying the higher-index tests, and strictly positive weights representing the target height polynomial. The next proposition gives explicit values and checks each requirement. The subsequent example explains why the first lift still adjusts the actual gain to satisfy the lowest test.

\begin{proposition}[Additional data for the RPS clock]\label{prop:seedadmissible}
The clock of Proposition~\ref{prop:seed} satisfies
Definition~\ref{def:admissible} with
\[
 \tau=1,\qquad c_{\rm ref}=1/2,\qquad\ell_*=(0,0,0),\qquad
 a_P=(2,2,2)^\top/27.
\]
\end{proposition}
\begin{proof}
We must verify the reference height and target row, the positivity and count identity for $a_P$, and the higher-index compatibility tests. Proposition~\ref{prop:seed} gives the remaining-height constant $\mathcal A_2=1$, so $H=\mathcal A_2+3=4$. With $\chivec=(1,1)$, $\tau=1$, and $c_{\rm ref}=1/2$, equation~\eqref{eq:y} gives
\[
 \zeta^\circ(q)=q^2/2+q+4,\qquad
 t_{1,1/2}(\chivec)p_1(q)=4/3+2q/3.
\]
Indeed, the target polynomial $4/3+2\sigma/3$ sums over $\sigma=q,\ldots,2q-1$ to
\[
 \frac43q+\frac23\frac{3q^2-q}{2}
 =q^2+q=\zeta^\circ(q)+\zeta^\circ(q-1)-\frac{15}{2},
\]
which verifies~\eqref{eq:target}, since $2H-c_{\rm ref}=15/2$. For the proposed positive weights,
\[
 a_P^\top C_2=(4/3,2/3)=t_{1,1/2}(\chivec),
\]
with all three weights equal to $2/27>0$, so the positive-weight requirement holds. It remains to check the higher-index compatibility tests~\eqref{eq:higher} required by Definition~\ref{def:admissible}(i). For $k=2$, the only higher compatibility test in~\eqref{eq:higher} is
$h=4$. It is the bottom-right entry and vanishes automatically, so
$\ell_*=0$ satisfies all required higher tests.
\end{proof}

\paragraph{Adjusting the gain to remove the lowest RPS compatibility obstruction}
Proposition~\ref{prop:seedadmissible} supplies all input data for the first lift, but exact score restoration also requires the lowest compatibility test $h=2$. That test is deliberately excluded from Definition~\ref{def:admissible}(i), which imposes only~\eqref{eq:higher}, and at the reference gain $c_{\rm ref}=1/2$ it does not vanish. We compute the obstruction here and show how the later reset-count scale $T_1$ allows a small change of the actual gain $c$ to remove it. This illustrates the general even-order adjustment that will be proved in Proposition~\ref{prop:momentchoice}.

For the RPS data,
\[
 X=\begin{pmatrix}0&8&9\\0&9&9\\0&10&9\\1&0&0\end{pmatrix},\qquad
 (D^\circ)^\top X=
 \begin{pmatrix}-1/2&-23/2&-27/2\\-1&27&27\\0&0&0\end{pmatrix},
 \qquad \psi^\circ=(-1/2,-1,0).
\]
Here the width-three reference-height row is $\zetavec^\circ=(4,3/2,1)$.
The matrix for the lowest compatibility test in~\eqref{eq:cokernel} is
\[
 Z^{(2)}=\begin{pmatrix}0&0&1\\0&-1&0\\1&1&-1/8\end{pmatrix},
 \qquad\ip{Z^{(2)}}{(D^\circ)^\top X}=-81/2.
\]
The test value $-81/2$ is nonzero, so keeping $c=1/2$ would leave the two reset equations inconsistent. To correct it, let $T_1\ge1$ be the integer reset-count scale to be chosen in Section~\ref{sec:synthesis}. The least positive integer clearing the denominators of $\zetavec^\circ$ is $\lambda=2$. With $r=k+1=3$ and $\ell_*=0$, the total reset-count row prescribed in~\eqref{eq:baseline} is $\ell=2r\lambda T_1\zetavec^\circ$, so $\ell_2=12T_1$, while $(\one_4^{\top}X)_2=27$. Changing the gain by $\Delta c=c-1/2$ changes the lowest test by $\Delta c(27+12T_1)$. Setting this equal to $81/2$ gives the correction in~\eqref{eq:parity}:
\begin{equation}\label{eq:rpscorrection}
 c=\frac12+\frac{27}{8T_1+18}
   =\frac{2(T_1+9)}{4T_1+9}.
\end{equation}
The correction is positive and of order $T_1^{-1}$. Moreover, $c<1$ is equivalent to $T_1>9/2$, so every integer $T_1\ge5$ makes $c\in(0,1)$ and cancels this algebraic obstruction. This bound on $T_1$ guarantees compatibility and an admissible gain; the coupling signs and waiting margins require the additional lower bounds on $T_1$ proved in Section~\ref{sec:uniform}, applied here with $k=2$. We choose $T_1$ large enough for all these requirements. The reference gain used to define the reusable positive-weight data is still $c_{\rm ref}=1/2$. The calculations in this section establish the $k=2$ instance. Proposition~\ref{prop:momentchoice} will justify the analogous correction at every even parent order and show that $c=c_{\rm ref}$ suffices at odd parent orders.

\subsection{Constructing the reset actions with the required signs}\label{sec:synthesis}
Sections~\ref{sec:gateway}--\ref{sec:moment} reduced exact score restoration to the compatibility equation and constructed the higher-test correction. We now build the $2(k+1)$ reset actions required by Theorem~\ref{thm:lift}: their count table $M$, interaction matrix $J$, and coupling $P$ must satisfy $PM=D$ and $P^{\top}X=Q$ in~\eqref{eq:twosided}, with negative early and positive late old-action couplings.

The integer $T_1$ controls the common scale of the reset counts. We choose it first to secure compatibility, signs, and waiting margins uniformly over every permitted value of a second integer $T_2$. This second parameter changes lower-degree count coefficients while preserving the group totals and leading coefficients. Section~\ref{sec:next} chooses $T_2$ to obtain positive weights for the next lift; the uniform estimates proved here ensure that this later choice preserves all the required signs.

The outputs are the exact compatibility solution of Proposition~\ref{prop:momentchoice} and the uniform sign and waiting bounds of Lemma~\ref{lem:uniform}. Section~\ref{sec:lowesttest} chooses the total reset counts and removes the possible lowest obstruction, Section~\ref{sec:resetinteractions} constructs $J$ and $P$, and Section~\ref{sec:uniform} proves the bounds. Best-response validity is checked afterward in Section~\ref{sec:chronology}.

\paragraph{Parameters, tables, and norms.}
Throughout this section we fix the order $k$, the parent clock, and its additional data $(\tau,c_{\rm ref},\ell_*,a_P)$. In particular, $X$ is the old-action count table from~\eqref{eq:old}; $H=\mathcal A_k+3$ and $\zetavec^\circ$ are the baseline height and the reference height row from~\eqref{eq:y}, evaluated at $c=c_{\rm ref}$; $a$ is the positive weight vector in~\eqref{eq:coneidentity}; and $L_X=X(X^{\top}X)^{-1}$ is the right inverse of $X^{\top}$ chosen after Lemma~\ref{lem:interpolate}. The shift matrices and the auxiliary matrices defined solely from this parent data do not depend on $T_1$ or $T_2$. The actual gain $c$, height row $\zetavec$, score-change table $D$, and interaction table $E$ may depend on $T_1$, but not on $T_2$. The count tables $W,W_E,M$ and interaction blocks $J,Q,P$ may depend on both parameters.

For a matrix $Z$, we use the transpose-invariant, submultiplicative norm
\begin{equation}\label{eq:coefficientnorm}
 \norm{Z}_*:=\max\{\norm{Z}_1,\norm{Z}_\infty\}
 =\max\left\{\max_j\sum_i|Z_{ij}|,\ \max_i\sum_j|Z_{ij}|\right\}.
\end{equation}
Here $\norm{Z}_1$ and $\norm{Z}_\infty$ are the induced column-sum and row-sum norms. Thus, $\norm{YZ}_*\le\norm{Y}_*\norm{Z}_*$ for compatible matrices and $\norm{Z^{\top}}_*=\norm{Z}_*$. For a row or column vector $f$, write $\norm{f}_1=\sum_j|f_j|$ and $\norm{f}_\infty=\max_j|f_j|$.

The statement $Z=O(T_1^{-\alpha})$ uniformly in $T_2$ means that there are constants $K,T_0$, depending only on the fixed parent data, such that
\[
 \norm{Z}_*\le K T_1^{-\alpha}
 \quad\text{for every integer }T_1\ge T_0
 \text{ and every integer }T_2\ge\max\{T_1,\upsilon_k\}.
\]
The coefficient $\upsilon_k=T_1\lambda\zetavec_k^\circ$ is defined in~\eqref{eq:baseline} below, so its dependence on $T_1$ is included in this quantifier. For a coefficient table of width $k+1$, the corresponding evaluated-vector bound is
\[
 \norm{Zp_k(q)}_\infty
 \le \norm{Z}_*\norm{p_k(q)}_\infty
 \le K T_1^{-\alpha}\norm{p_k(q)}_\infty.
\]
The parameter limit $T_1\to\infty$ compares different candidate games. Once $T_1$ and $T_2$ have both been chosen, all payoffs are constant; later growth estimates as $q\to\infty$ concern the phases of that one fixed game.

\subsubsection{Choosing the total counts and satisfying the lowest compatibility test}\label{sec:lowesttest}
The correction row $\ell_*$ is supplied by Lemma~\ref{lem:scale} and satisfies the higher-index tests of Definition~\ref{def:admissible}(i). We now choose the total reset counts and, when necessary, adjust the gain to satisfy the lowest-index test as well. Recall $r=w=k+1$, and let $\lambda\ge1$ be the least positive integer for which the reference height row $\lambda\zetavec^\circ$ has integer entries. For an integer $T_1\ge1$, define the rows $\upsilon,\ell_0,\ell\in\Z^{1\times w}$ and the matrix $E_0\in\Q^{w\times w}$ by
\begin{equation}\label{eq:baseline}
  \upsilon= T_1 \lambda \zetavec^\circ, \qquad
  \ell_0 = 2r \lambda \zetavec^\circ, \qquad
  \ell = T_1 \ell_0 + \ell_*, \qquad
  E_0 = 2r \lambda (\zetavec^\circ)^{\top} \zetavec^\circ.
\end{equation}
Lemma~\ref{lem:cokernel} indexes the tests by the even integers $h\in[k,2k]$. If $k$ is odd, then every such $h$ is strictly greater than $k$, so Lemma~\ref{lem:scale} has already satisfied all of them and we set $c=c_{\rm ref}$. If $k$ is even, then only the test with $h=k$ remains. We make its residual zero by choosing
\begin{equation}\label{eq:parity}
  c = c_{\rm ref} - \frac{\ip{Z^{(k)}}{(D^\circ)^{\top} X - \ell_*^{\top} \psi^\circ}}{(\one_{m+1}^{\top} X)_k + \ell_k}.
\end{equation}
The column $C_{:,k-1}$ is nonzero by full column rank and entrywise nonnegative because the count polynomials are eventually nonnegative. Multiplication by the positive diagonal entry $(F_\tau)_{k-1,k}$ therefore gives $(\one_{m+1}^{\top}X)_k>0$. Because this quantity is positive and $\ell_k = 2r\lambda \zetavec_k^\circ T_1 + \ell_{*,k} \to +\infty$ as $T_1 \to \infty$, the denominator of~\eqref{eq:parity} is strictly positive and $c = c_{\rm ref} + O(T_1^{-1}) \in (0,1)$ for all sufficiently large $T_1$. Using this $c$, define $\zetavec$ by the height equation~\eqref{eq:y} from Section~\ref{sec:lift-tables}: $\zeta(q)-\zeta(q-1)=\chi(2\tau q)-\chi(\tau q)+c$ with $\zeta(0)=H$. Define $D=S\mathsf{T}-S_{\mathrm{mid}}$ by~\eqref{eq:old}, the old-score change that the reset word must supply, and fix the endpoint action margin
\[
  \delta = \frac{1}{2}\min\{\delta_0, c\},
\]
independently of $T_2$.

To select a solution of $\LL(E)=F$ consistently, fix a rational linear map
\[
 \RR:\operatorname{im}(\LL)\longrightarrow\Q^{w\times w},
 \qquad \LL(\RR(F))=F.
\]
For example, row-reduce the fixed linear system for $\LL$, with matrix entries listed row by row, and set its free variables to zero. This gives a solution $E=\RR(F)$ that depends linearly on each compatible right-hand side $F$; the map depends only on $k$. We retain the leading term $T_1E_0$ and use this map to correct its residual:
\begin{equation}\label{eq:E}
  E = T_1 E_0 + \RR\bigl(D^{\top} X - \ell^{\top} \psi - T_1 \LL(E_0)\bigr).
\end{equation}
The next proposition proves that the argument of $\RR$ is compatible, making this construction a solution of the required matrix equation.
\begin{proposition}[Exact solution of the scalar compatibility tests]\label{prop:momentchoice}
For all sufficiently large $T_1$, the gain $c$ in~\eqref{eq:parity} (or $c=c_{\rm ref}$ for odd $k$) belongs to $(0,1)$. The argument of $\RR$ in~\eqref{eq:E} satisfies every compatibility test, so that equation defines a rational solution of $E-E^{\top}B=D^{\top}X-\ell^{\top}\psi$, the compatibility equation~\eqref{eq:moment}. Moreover, as $T_1\to\infty$ with the parent data fixed,
\begin{equation}\label{eq:Esize}
  c - c_{\rm ref} = O(T_1^{-1}), \qquad E - T_1 E_0 = O(1).
\end{equation}
\end{proposition}
\begin{proof}
The proof has two parts: the gain correction preserves the tests already solved, and the residual left for $\RR$ stays bounded. First, the definition of $E_0$ gives
\[
 \LL(E_0)=2r\lambda(\zetavec^\circ)^{\top}
              (\zetavec^\circ-\zetavec^\circ B)
          =-\ell_0^{\top}\psi^\circ.
\]
Write $d_c:=c-c_{\rm ref}$ for the scalar gain change, and let $e_1^{\top}=(0,1,0,\ldots,0)\in\Q^{1\times w}$, the coefficient row of the polynomial $q=\binom q1$. The height equation~\eqref{eq:y}, with its value at zero unchanged, yields
\[
 \zetavec-\zetavec^\circ=d_c e_1^{\top},\qquad
 D-D^\circ=d_c\one_{m+1}\erow{w},\qquad
 \psi-\psi^\circ=-d_c\erow{w}.
\]
Using $\ell=T_1\ell_0+\ell_*$ from~\eqref{eq:baseline}, we can therefore expand the residual as
\begin{equation}\label{eq:compatibilityresidual}
 D^{\top}X-\ell^{\top}\psi-T_1\LL(E_0)
 =(D^\circ)^{\top}X-\ell_*^{\top}\psi^\circ
   +d_c\bigl((\erow{w})^{\top}\one_{m+1}^{\top}X+\ell^{\top}\erow{w}\bigr).
\end{equation}
For every even $h>k$, the test matrix $Z^{(h)}$ has zero first row and first column. Its inner product with the last term is therefore zero, and its inner product with the first two terms is zero by~\eqref{eq:higher}. These are exactly the higher-index tests of Lemma~\ref{lem:cokernel}. When $k$ is even, $Z^{(k)}_{0,k}=Z^{(k)}_{k,0}=1$, with all other entries in its first row and first column zero. Thus, the remaining test equals
\[
 \ip{Z^{(k)}}{(D^\circ)^{\top}X-\ell_*^{\top}\psi^\circ}
 +d_c\bigl((\one_{m+1}^{\top}X)_k+\ell_k\bigr)=0
\]
by~\eqref{eq:parity}. The subtracted term $T_1\LL(E_0)$ has zero inner product with every test because it belongs to the image of $\LL$. Hence~\eqref{eq:E} solves the matrix compatibility equation exactly.

Finally, $d_c=O(T_1^{-1})$, $\ell=T_1\ell_0+\ell_*$, and all other quantities on the right of~\eqref{eq:compatibilityresidual} are independent of $T_1$. The residual matrix on the left of~\eqref{eq:compatibilityresidual} therefore has norm $O(1)$. The fixed linear map $\RR$ used in~\eqref{eq:E} maps this bounded residual to an $O(1)$ correction, proving~\eqref{eq:Esize}.
\end{proof}

Proposition~\ref{prop:momentchoice} has solved~\eqref{eq:moment} for the total count row $\ell$ chosen in~\eqref{eq:baseline}, but we still need individual reset runs whose interactions can realize $E$. For this sufficient construction we use two groups of $r=k+1$ reset actions: an early group, which will lower every old-action score, and a late group, which will restore those scores. We give each group an invertible count table so that its interactions can be solved explicitly. The adjustment matrix $V_0$ below has zero column sums, $\one_r^{\top}V_0=0$, and a zero last column. It lets $T_2$ change the lower count coefficients without changing the total counts or the leading coefficients. Define $V_0\in\Z^{r\times w}$ and $M\in\Z^{2r\times w}$ by
\begin{equation}\label{eq:W}
  V_0 = \begin{pmatrix}I_k & 0 \\ -\one_k^{\top} & 0\end{pmatrix}, \qquad
  W = \one_r \upsilon+ T_2 V_0, \qquad
  W_E = W + e_0 \ell_*, \qquad
  M = \begin{pmatrix}W_E \\ W\end{pmatrix},
\end{equation}
where $T_2\ge\max\{T_1,\upsilon_k\}$ is an integer controlling the lower-degree count coefficients and $e_0\in\Q^r$ in $W_E$ selects its first row. The early table $W_E$ and late table $W$ each have size $r\times w$, with $r=w$. As polynomials, the first $k$ rows of $W$ are $\upsilon p_k(q)+T_2\binom qi$ for $i=0,\ldots,k-1$; its last row is $\upsilon p_k(q)-T_2\sum_{i=0}^{k-1}\binom qi$. Their adjustments sum to zero and have degree below $k$. Their combined table has exactly the total count row already used in the compatibility equation:
\[
 \one_{2r}^{\top}M=2r\upsilon+\ell_*=\ell.
\]
Its leading coefficient column is
\[
 \gamma:=M_{:,k}=\upsilon_k\one_{2r}+\ell_{*,k}e_0>0
\]
for sufficiently large $T_1$, where here $e_0\in\Q^{2r}$. Consequently every individual reset count is an integer-valued degree-$k$ polynomial that is positive at all sufficiently large phase indices. These two facts preserve the solved total-count equation and provide positive run lengths; we next choose interactions realizing $E$.

\subsubsection{Constructing the reset interactions and the old-action coupling}\label{sec:resetinteractions}
We now construct $J$ so that $M^{\top}JM=E$, and then construct $P$ solving both reset equations. The matrix $J$ describes how a reset action raises the scores of later reset actions; its lower-left block $C_\times$ links the early group to the late group. For $P$, we begin from a matrix $P_0$ with negative early columns and positive late columns, and correct it using Lemma~\ref{lem:interpolate}. The next section proves that these corrections are small enough to preserve all required signs.
For the inverse formulas we need three matrices determined by the parent data, independently of $T_1$ and $T_2$: $\Pi\in\Q^{k\times r}$ and $\Phi_0,\Phi_1\in\Q^{w\times r}$. Define
\begin{equation}\label{eq:phi}
  \Pi = \bigl[I_k\;\big|\;\mathbf 0_k\bigr] - \frac{1}{r}\one_k \one_r^{\top}, \qquad
  \Phi_0 = \frac{e_k \one_r^{\top}}{r \lambda \zetavec_k^\circ}, \qquad
  \Phi_1 = \begin{pmatrix}I_k \\ -\zetavec_{0:(k-1)}^\circ / \zetavec_k^\circ\end{pmatrix} \Pi.
\end{equation}
Here $\zetavec_{0:(k-1)}^\circ=(\zetavec_0^\circ,\ldots,\zetavec_{k-1}^\circ)\in\Q^{1\times k}$, and $\bigl[I_k\;\big|\;\mathbf 0_k\bigr]\in\Q^{k\times w}$ appends the zero column $\mathbf 0_k\in\Q^k$ to the identity. The identities $\Pi\one_r=0$, $\Pi V_0=\bigl[I_k\;\big|\;\mathbf 0_k\bigr]$, and $\one_r^{\top}W=r\upsilon=rT_1\lambda\zetavec^\circ$ give the exact closed-form inverses
\begin{equation}\label{eq:inverse}
  W^{-1} = T_1^{-1} \Phi_0 + T_2^{-1} \Phi_1, \qquad
  W_E^{-1} = W^{-1} - \frac{W^{-1} e_0 \ell_* W^{-1}}{1 + \ell_* W^{-1} e_0}.
\end{equation}
The first formula separates the common row-sum component, recovered by $\Phi_0$, from the centered component, recovered using $\Pi$ in $\Phi_1$. Multiplying it by $W$ and using the three preceding identities gives $I_w$. The second is the Sherman--Morrison formula for the rank-one update $W_E=W+e_0\ell_*$. Since $\Phi_0,\Phi_1,\ell_*$ are fixed and $T_2^{-1}\le T_1^{-1}$, it also gives $\ell_*W^{-1}e_0=O(T_1^{-1})$ uniformly for $T_2\ge T_1$. Thus, the denominator in the second formula is nonzero for every such $T_2$ once $T_1$ is sufficiently large. In particular, there is a parent-dependent constant $K$ such that
\[
 \norm{W^{-1}}_*,\ \norm{W_E^{-1}}_*\le K/T_1
\]
throughout this range.

Let $L_r\in\Z^{r\times r}$ have ones strictly below the diagonal and zeros elsewhere. We use the positive scalar $\eta$ to create handoffs within each group and the positive scalar $\kappa$ for the leading interactions between groups. The matrix $E_{\rm base}$ is the contribution of those leading cross-group interactions. The residual $F_{\rm res}\in\Q^{w\times w}$ accounts for what remains after subtracting this contribution and the within-group interactions. The choice $\eta=(T_1T_2^2)^{-1}$ offsets the quadratic growth of those internal contributions in $T_2$. The choice of $\kappa$ makes the leading cross-group contribution exactly $T_1E_0$ plus its fixed correction. Lemma~\ref{lem:uniform} checks the required uniform error bounds. Define
\begin{equation}\label{eq:J}
\begin{gathered}
  \eta = \frac{1}{T_1 T_2^2}, \qquad
  \kappa = \frac{2}{r \lambda T_1}, \qquad
  E_{\rm base} = T_1 E_0 + 2(\zetavec^\circ)^{\top} \ell_*, \\
  F_{\rm res} = E - E_{\rm base} - \eta\bigl(W_E^{\top} L_r W_E + W^{\top} L_r W\bigr), \\
  C_\times = \kappa \one_r \one_r^{\top} + W^{-\top} F_{\rm res} W_E^{-1}, \qquad
  J = \begin{pmatrix}\eta L_r & 0 \\ C_\times & \eta L_r\end{pmatrix}.
\end{gathered}
\end{equation}
Here $C_\times\in\Q^{r\times r}$ and $J\in\Q^{2r\times2r}$, while $E_{\rm base}$ has size $w\times w$. The definition of the cross-group block gives
\[
 \kappa W^{\top}\one_r\one_r^{\top}W_E
 =\kappa(r\upsilon)^{\top}(r\upsilon+\ell_*)=E_{\rm base},\qquad
 W^{\top}C_\times W_E=E_{\rm base}+F_{\rm res}.
\]
Consequently,
\[
 \begin{aligned}
 M^{\top}JM
 &=\eta(W_E^{\top}L_rW_E+W^{\top}L_rW)
      +W^{\top}C_\times W_E\\
 &=\eta(W_E^{\top}L_rW_E+W^{\top}L_rW)
      +E_{\rm base}+F_{\rm res}=E.
 \end{aligned}
\]
We have therefore realized the solution of the compatibility equation. To obtain the coupling $P$, we apply the formula~\eqref{eq:interpolate} from Lemma~\ref{lem:interpolate} with the following reference matrix $P_0\in\Q^{(m+1)\times2r}$ and left inverse $R_M\in\Q^{w\times2r}$:
\begin{equation}\label{eq:P0}
  P_0 = \bigl[-a \one_r^{\top} \;\big|\; a \one_r^{\top}\bigr], \qquad
  R_M = \bigl[0 \;\big|\; W^{-1}\bigr].
\end{equation}
The brackets in $R_M$ concatenate a $w\times r$ zero block and $W^{-1}$ horizontally, so $R_MM=W^{-1}W=I_w$. The rank of $X$ was established after~\eqref{eq:old}; the invertible late block $W$ gives full column rank of $M$. Moreover, $M^{\top}JM=E$ and~\eqref{eq:moment} imply $M^{\top}Q=D^{\top}X$, precisely the compatibility hypothesis of Lemma~\ref{lem:interpolate}. Lemma~\ref{lem:interpolate} therefore proves $PM=D$ and $P^{\top}X=Q$ exactly.

\subsubsection{Bounds that preserve the reset interaction and coupling signs}\label{sec:uniform}
The exact reset equations $PM=D$ and $P^{\top}X=Q$ alone do not guarantee that the proposed plays are best responses. We next show that $C_\times$ stays entrywise positive and that $P$ retains the opposite signs of its early and late reference columns. The uniformity in $T_2$ is essential: we will choose that parameter only in Section~\ref{sec:next}, after fixing $T_1$ here.

\begin{lemma}[Reset action bounds for both parameters]\label{lem:uniform}
As $T_1 \to \infty$, uniformly over all integers $T_2 \ge \max\{T_1, \upsilon_k\}$,
\begin{equation}\label{eq:uniform}
  C_\times = \kappa \one_r \one_r^{\top} + O(T_1^{-2}), \qquad
  Q = P_0^{\top} X + O(T_1^{-1}), \qquad
  P = P_0 + O(T_1^{-1}).
\end{equation}
Consequently, there is an integer $T_0$ depending only on the parent data such that, for every integer $T_1\ge T_0$ and every integer $T_2\ge\max\{T_1,\upsilon_k\}$, the following properties hold:
\begin{enumerate}[label=\textup{(\roman*)},leftmargin=*]
 \item The actual gain satisfies $0<c<1$, and the leading reset-count column satisfies $\gamma=M_{:,k}>0$ entrywise.
 \item Every entry of $J$ strictly below the diagonal is positive: $C_\times>0$ entrywise and $\eta>0$.
 \item The early columns of $P$ are strictly negative and its late columns are strictly positive, both entrywise: $P_E<0<P_L$.
 \item With $a_{\min}=\min_{0\le i\le g}a_i>0$ and $\epsilon=a_{\min}/4$, every early column satisfies the waiting margins
 \begin{equation}\label{eq:reserve}
 -P_{ij}>\epsilon\quad(i<m,\ 0\le j<r),\qquad
 -P_{gj}-c>\mathcal A_k\quad(0\le j<r).
 \end{equation}
\end{enumerate}
\end{lemma}
\begin{proof}
The waiting margins have a backward-time interpretation: each remaining parent play separates an early reset action from the parent maximum by more than $\epsilon$, while the final endpoint play leaves more than $\mathcal A_k$ of reserve. Proposition~\ref{prop:chronology} makes this comparison explicit. We use the norm defined in~\eqref{eq:coefficientnorm}. The first two steps bound the inverses and the residual defining $C_\times$; the next two cancel the potentially large count tables before estimating $Q$ and $P$; the last step turns those bounds into the stated signs and margins. Every estimate holds for all integers $T_2\ge\max\{T_1,\upsilon_k\}$.

\paragraph{Step 1: inverse bounds independent of $T_2$.}
Set $K_W=\norm{\Phi_0}_*+\norm{\Phi_1}_*$, which is a constant determined only by the parent data. Since $T_2\ge T_1$, the first formula in~\eqref{eq:inverse} gives
\[
 \norm{W^{-1}}_*\le K_W/T_1=:b_W,
 \qquad |\ell_*W^{-1}e_0|\le b_W\norm{\ell_*}_1=:\nu.
\]
For all sufficiently large $T_1$, $\nu<1/2$, so the denominator in the second inverse formula stays away from zero for every integer $T_2\ge\max\{T_1,\upsilon_k\}$. The second formula in~\eqref{eq:inverse} yields
\[
 \norm{W_E^{-1}}_*
 \le b_W+\frac{b_W^2\norm{\ell_*}_1}{1-\nu}
 =\frac{b_W}{1-\nu}=O(T_1^{-1}).
\]
The same bounds hold for the transposed inverses because $\norm{Z^{\top}}_*=\norm{Z}_*$.

\paragraph{Step 2: a bounded residual and a smaller cross-block error.}
The matrices $W$ and $W_E$ themselves can grow with $T_2$. However, their normalized versions are bounded:
\[
 \frac{W}{T_2}=\frac{T_1}{T_2}\one_r\lambda \zetavec^\circ+V_0,
 \qquad
 \frac{W_E}{T_2}=\frac{W}{T_2}+\frac{e_0\ell_*}{T_2}.
\]
Thus, both norms are at most
$b_C:=\norm{\one_r\lambda \zetavec^\circ}_*+\norm{V_0}_*+\norm{\ell_*}_1/T_1=O(1)$.
Since $\eta=(T_1T_2^2)^{-1}$,
\[
 \norm{\eta(W_E^{\top}L_rW_E+W^{\top}L_rW)}_*
 \le \frac{2\norm{L_r}_*b_C^2}{T_1}=O(T_1^{-1}).
\]
Also $E-E_{\rm base}=E-T_1E_0-2(\zetavec^\circ)^{\top}\ell_*=O(1)$ by~\eqref{eq:Esize}. Hence $F_{\rm res}=O(1)$ uniformly in $T_2$. Submultiplicativity of the matrix norm now gives
\[
 \norm{C_\times-\kappa\one_r\one_r^{\top}}_*
 \le\norm{W^{-\top}}_*\norm{F_{\rm res}}_*\norm{W_E^{-1}}_*
 =O(T_1^{-2}),
\]
which is the first estimate in~\eqref{eq:uniform}.

\paragraph{Step 3: cancel the count tables from the score-change errors.}
Multiplying the last error bound by $W$ or $W_E$ directly would give a bound proportional to $T_2/T_1^2$, which is not uniform over the allowed values of $T_2$. Instead, the definition of $C_\times$ gives the exact identities
\begin{align*}
 C_\times W_E
 &=\kappa\one_r\one_r^{\top}W_E+W^{-\top}F_{\rm res}
 =2\one_r \zetavec^\circ+\kappa\one_r\ell_*+W^{-\top}F_{\rm res},\\
 C_\times^{\top}WB
 &=\kappa\one_r\one_r^{\top}WB+W_E^{-\top}F_{\rm res}^{\top}B
 =2\one_r \zetavec^\circ B+W_E^{-\top}F_{\rm res}^{\top}B.
\end{align*}
Here $\one_r^{\top}W_E=r\upsilon+\ell_*$, $\one_r^{\top}W=r\upsilon$, and $\kappa r \upsilon=2\zetavec^\circ$. Note that the factors $W_E$ and $W$, whose norms can grow with $T_2$, have canceled from the error terms; only their inverse transposes remain. The leading terms also depend only on the row sums $r\upsilon$ and $r\upsilon+\ell_*$, which are independent of $T_2$.

Write $Q=\col(Q_E,Q_L)$ for the early and late rows. Expanding $JM-J^{\top}MB$ block by block gives
\begin{align*}
 Q_E&=\one_r \psi+\eta L_rW_E-\eta L_r^{\top}W_EB-C_\times^{\top}WB,\\
 Q_L&=\one_r \psi+C_\times W_E+\eta L_rW-\eta L_r^{\top}WB.
\end{align*}
For either $Z=W$ or $W_E$, the internal terms satisfy
\[
 \norm{\eta L_rZ}_*+\norm{\eta L_r^{\top}ZB}_*
 \le\frac{\norm{L_r}_*b_C(1+\norm{B}_*)}{T_1T_2}
 =O(T_1^{-2}).
\]
Further, $\psi-\psi^\circ=-(c-c_{\rm ref})\erow{w}=O(T_1^{-1})$, $\kappa\ell_*=O(T_1^{-1})$, and the remaining inverse-times-residual terms are $O(T_1^{-1})$ by Steps 1--2. The leading cancellations are
\[
 \psi^\circ-2\zetavec^\circ B=-(\zetavec^\circ+\zetavec^\circ B),
 \qquad \psi^\circ+2\zetavec^\circ=\zetavec^\circ+\zetavec^\circ B.
\]
Therefore
\[
 Q=\begin{pmatrix}-\one_r(\zetavec^\circ+\zetavec^\circ B)\\
                         \one_r(\zetavec^\circ+\zetavec^\circ B)\end{pmatrix}
       +O(T_1^{-1})=P_0^{\top}X+O(T_1^{-1}),
\]
using the positive-weight identity $a^{\top}X=\zetavec^\circ+\zetavec^\circ B$. This proves the second coefficient-table estimate.

\paragraph{Step 4: reconstruct the coupling matrix.}
The common baseline count row $r\upsilon=rT_1\lambda\zetavec^\circ$, which grows with $T_1$, also cancels between the early and late groups:
\[
 P_0M=-a(r\upsilon+\ell_*)+a(r\upsilon)=-a\ell_*.
\]
Thus, $D_R=D+a\ell_*=O(1)$, since $D=D^\circ+(c-c_{\rm ref})\one_{m+1}\erow{w}$ is bounded. Step 3 gives $Q_R:=Q-P_0^{\top}X=O(T_1^{-1})$, while $R_M=[0\mid W^{-1}]$ is the horizontal block matrix in~\eqref{eq:P0} and satisfies $\norm{R_M}_*=\norm{W^{-1}}_*=O(T_1^{-1})$. The matrices $L_X$ and $X$ depend only on the parent data, not on $T_1$ or $T_2$. Formula~\eqref{eq:interpolate} consequently implies
\[
 \norm{P-P_0}_*
 \le\norm{D_R}_*\norm{R_M}_*
 +\norm{L_X}_*\bigl(\norm{Q_R}_*+
           \norm{X}_*\norm{D_R}_*\norm{R_M}_*\bigr)
 =O(T_1^{-1}).
\]
This proves the final estimate with a constant independent of $T_2$.

\paragraph{Step 5: choose one scale that guarantees every sign.}
Every entry of the leading term $\kappa\one_r\one_r^{\top}$ in $C_\times$ is $\kappa=2/(r\lambda T_1)>0$, whereas its entrywise error is $O(T_1^{-2})$. More explicitly, Steps 2 and 4 provide constants $K_\times,K_P$ independent of both parameters such that
\[
 \norm{C_\times-\kappa\one_r\one_r^{\top}}_*
       \le K_\times/T_1^2,
 \qquad
 \norm{P-P_0}_*\le K_P/T_1.
\]
Thus, it is enough to choose $T_1$ so large that
\begin{equation}\label{eq:signcriterion}
 \frac{K_\times}{T_1^2}<\frac{\kappa}{2},\qquad
 \frac{K_P}{T_1}<
 \min\left\{\frac{a_{\min}}4,
             \frac{a_g-\mathcal A_k-1}{2}\right\}.
\end{equation}
The second right-hand side is positive because $a_{\min}>0$ and
$a_g=2(\mathcal A_k+3)-c_{\rm ref}>\mathcal A_k+1$. The first inequality gives $C_\times>0$ entrywise, and $\eta>0$ gives all remaining positive entries below the diagonal of $J$. The second gives $P_E<0<P_L$ entrywise. In terms of the family of matrices indexed by $(T_1,T_2)$, each early column differs from $-a$ and each late column differs from $a$ by at most $K_P/T_1$ in every entry, uniformly in $T_2$. In particular,
\[
 \begin{aligned}
 -P_{ij}&\ge a_i-K_P/T_1>3a_{\min}/4>\epsilon
                   &&(i<m,\ 0\le j<r),\\
 -P_{gj}-c&>a_g-K_P/T_1-1>\mathcal A_k
                   &&(0\le j<r).
 \end{aligned}
\]
These inequalities prove~\eqref{eq:reserve}. In Appendix~\ref{app:bounds}, we provide explicit rational bounds for the same errors, which can be used to choose a numerical value of $T_1$.

Since $c=c_{\rm ref}+O(T_1^{-1})$ with $0<c_{\rm ref}<1$, $\nu=O(T_1^{-1})$, and $\upsilon_k=T_1\lambda\zetavec_k^\circ\to+\infty$, sufficiently large $T_1$ also gives $0<c<1$, $\nu<1/2$, and $\upsilon_k+\ell_{*,k}>0$; the other entries of $\gamma$ equal $\upsilon_k>0$. All these inequalities hold for every sufficiently large $T_1$, uniformly over the permitted values of $T_2$. We may therefore fix $T_1$ before choosing $T_2$. The exact reset also gives $P\gamma=(PM)_{:,k}=D_{:,k}=0$.
\end{proof}

We have now solved both reset equations $PM=D$ and $P^{\top}X=Q$ from~\eqref{eq:twosided} and proved the signs and margins of Lemma~\ref{lem:uniform} for every integer $T_2\ge\max\{T_1,\upsilon_k\}$. In Section~\ref{sec:next}, we use the remaining freedom in $T_2$ to obtain the positive weights required at the next order.

\subsection{Preserving the additional data for the next order}\label{sec:next}
We now supply the positive-weight data required by Definition~\ref{def:admissible}(ii) for order $k+1$, completing the algebraic part of Theorem~\ref{thm:lift}. Section~\ref{sec:synthesis} fixed $T_1$ and secured both reset equations, signs, and waiting margins for every integer $T_2\ge\max\{T_1,\upsilon_k\}$. We use the remaining freedom to make the next-order weights positive.

The key point in Proposition~\ref{prop:nextweights} is that the next scale $\tau_+$ and target row can be computed independently of $T_2$. Increasing $T_2$ then makes the late-reset weights positive. The uniform estimates of Lemma~\ref{lem:uniform}, not merely preservation of group totals, ensure that this choice does not invalidate the reset signs and margins.

The output is a fixed permitted integer $T_2$ together with the next correction and positive weights. Section~\ref{sec:weight-preservation} obtains them using algebraic identities only. Section~\ref{sec:chronology} will establish the remaining best-response and clock-inheritance conditions for the resulting fixed game.

Recall that the target row $t_{\tau,c_{\rm ref}}(\chivec)\in\Q^{1\times k}$ is defined by~\eqref{eq:target}; the positive-weight conditions in Definition~\ref{def:admissible} are $a_P>0$ and $a_P^{\top}C=t_{\tau,c_{\rm ref}}(\chivec)$. We need this same identity for the next clock, whose height-coefficient row will be $\zetavec$. The reset-count table $M=\col(W_E,W)$ and interaction table $J$ are those of~\eqref{eq:W} and~\eqref{eq:J}.

\subsubsection{Choosing the remaining parameter and preserving positive weights}\label{sec:weight-preservation}
We form the rotated tables, isolate the quantities independent of the free parameter, and then make the remaining weights positive. This supplies the reusable algebraic input for Proposition~\ref{prop:inherit} and the next application of Theorem~\ref{thm:lift}.

Recall the old count and intermediate score tables $X,S_{\mathrm{mid}}$ from~\eqref{eq:old}, the forward shift $\mathsf{T}$, and the reset tables $M,J$ from~\eqref{eq:W} and~\eqref{eq:J}. With $T_1$ fixed by Lemma~\ref{lem:uniform} and $T_2$ still free, define the rotated order-$(k+1)$ phase word and tables
\begin{equation}\label{eq:rotate}
  \W_{k+1}(q) = \Bw(q)\Aw(q+1), \qquad
  R_+ = \begin{pmatrix}S_{\mathrm{mid}} \\ \one_{2r} \zetavec  - JM\end{pmatrix}, \qquad
  C_+ = \begin{pmatrix}X\mathsf{T} \\ M\end{pmatrix}.
\end{equation}
The word begins after $\Aw(q)$, at old scores $S_{\mathrm{mid}}$ and reset scores $\one_{2r}\zetavec-JM$. It then uses the current reset counts $M$ and the next old-word counts $X\mathsf{T}$, where $\mathsf{T}$ shifts $q$ to $q+1$. These are algebraic tables at this point; Section~\ref{sec:chronology} proves their trajectory interpretation.

\begin{proposition}[Positive weights for the next order]\label{prop:nextweights}
Fix a sufficiently large scale $T_1$ as in Lemma~\ref{lem:uniform}. There is a sufficiently large integer $T_2\ge\max\{T_1,\upsilon_k\}$ such that the rotated tables in~\eqref{eq:rotate} satisfy the positive-weight condition of Definition~\ref{def:admissible}(ii) at order $k+1$, with reference gain $c_{{\rm ref},+}=1/2$. This choice preserves the reset equations, signs, and margins proved in Section~\ref{sec:synthesis}.
\end{proposition}
\begin{proof}

By~\eqref{eq:twosided}, $G C_+ = R_+ N$. The four quantities used in Lemma~\ref{lem:scale} to compute the next scale factor and correction row $(\tau_+,\ell_{*,+})$ are independent of $T_2$:
\begin{equation}\label{eq:nextsummary}
  R_+^{\top} C_+ = S_{\mathrm{mid}}^{\top} X \mathsf{T} + \zetavec^{\top} \ell - E^{\top}, \qquad
  \one^{\top} C_+ = \one^{\top} X \mathsf{T} + \ell, \qquad
  (R_+)_{\reset{0},:} = \zetavec , \qquad (C_+)_{g,:} = \erow{w}.
\end{equation}
Indeed, the reset-action contribution is
\[
 (\one_{2r}\zetavec-JM)^{\top}M
 =\zetavec^{\top}(\one_{2r}^{\top}M)-(M^{\top}JM)^{\top}
 =\zetavec^{\top}\ell-E^{\top}.
\]
Thus, $S_{\mathrm{mid}}, X, \zetavec , \ell, E$ depend only on $T_1$ (since $\delta = \min\{\delta_0, c\}/2$ was fixed before $T_2$).

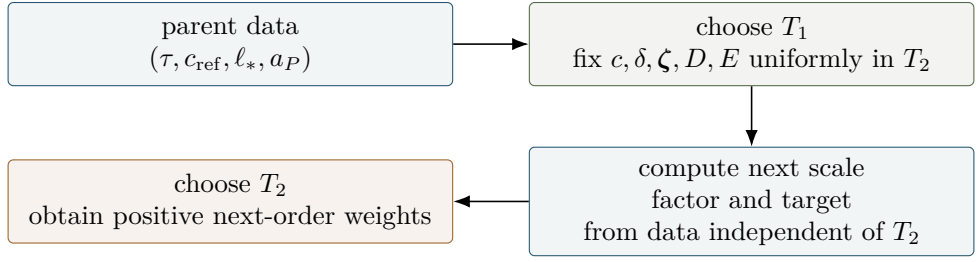
\begin{figure}[tb]\centering
\begin{tikzpicture}[font=\small,node distance=8mm and 10mm]
\node[parbox,text width=56mm] (a) {parent data\\$(\tau,c_{\rm ref},\ell_*,a_P)$};
\node[conbox,right=of a,text width=56mm] (b) {choose $T_1$\\fix $c,\delta,\zetavec ,D,E$ uniformly in $T_2$};
\node[parbox,below=of b,text width=56mm] (c) {compute next scale factor and target\\from data independent of $T_2$};
\node[gatbox,left=of c,text width=56mm] (k) {choose $T_2$\\obtain positive next-order weights};
\draw[arr] (a)--(b);\draw[arr] (b)--(c);\draw[arr] (c)--(k);
\end{tikzpicture}
\caption{The order of parameter choices. Increasing $T_2$ in the last step preserves the signs: the bounds used to choose $T_1$ hold for every permitted value of $T_2$.}\label{fig:parameters}
\end{figure}

Lemma~\ref{lem:scale} uses only these four quantities, the reset identity, the counter identity, and the positive leading entry coefficient. It does not use trajectory validity. Thus, it can already be applied to the $T_2$-independent quantities in~\eqref{eq:nextsummary} with reference gain $c_{{\rm ref},+} = 1/2$ to produce the next-order integer scale factor $\tau_+ \ge 1$, count correction row $\ell_{*,+} \in \Z^{1 \times (k+2)}$, and target $t_+ = t_{\tau_+, 1/2}(\zetavec ) \in \Q^{1 \times w}$, which satisfies $(t_+)_k > 0$ by~\eqref{eq:target}. We assign a common small positive weight $\mu$ to all old and early-reset rows. Their sum is $r_0$ below. The late rows must represent the residual target, so their weights are $t_{\rm rem}W^{-1}$. Define
\begin{equation}\label{eq:nextweights}
  r_0 = \one_{m+1}^{\top} X \mathsf{T} + r \upsilon+ \ell_*, \qquad
  \mu = \frac{(t_+)_k}{4(r_0)_k} > 0, \qquad
  t_{\rm rem} = t_+ - \mu r_0, \qquad
  \alpha = \frac{(t_{\rm rem})_k}{r \upsilon_k} > 0.
\end{equation}
Here $r_0,t_{\rm rem}\in\Q^{1\times w}$ are rows and $\mu,\alpha$ are positive rational scalars. The row $r_0=\one_{m+1}^{\top}X\mathsf{T}+\one_r^{\top}W_E$ sums the old-action and early-reset rows of $C_+$. Its last entry is positive: the old count polynomials have nonnegative leading coefficients, and every early reset count has a positive degree-$k$ coefficient. Thus, $\mu$ is well-defined, and $(t_{\rm rem})_k=\frac34(t_+)_k>0$. Multiplying $t_{\rm rem}$ by the closed-form inverse $W^{-1}$ in~\eqref{eq:inverse} gives the late reset-action weight vector
\[
  t_{\rm rem} W^{-1} = \alpha \one_r^{\top} + T_2^{-1} t_{\rm rem} \Phi_1.
\] The constant term is $t_{\rm rem}(T_1^{-1}\Phi_0)=\alpha\one_r^{\top}$ because $\Phi_0$ selects the last target coefficient and $\upsilon_k=T_1\lambda\zetavec_k^\circ$.
Choosing the integer $T_2$ so that
\begin{equation}\label{eq:radius}
  T_2 > \max\Bigl\{T_1,\, \upsilon_k,\, 1,\, \frac{2\norm{t_{\rm rem} \Phi_1}_\infty}{\alpha}\Bigr\}
\end{equation}
ensures $t_{\rm rem} W^{-1} > \frac{1}{2}\alpha \one_r^{\top} > 0$. Assigning weight $\mu > 0$ to every old row and every row of $C_+$ for an early reset action, and assigning weights $(t_{\rm rem} W^{-1})^{\top} > 0$ to the $r$ rows for late reset actions of $C_+$, yields a strictly positive vector $a_{P,+} \in \Q_{>0}^{m+1+2r}$ satisfying
\[
  a_{P,+}^{\top} C_+ = \mu r_0 + (t_{\rm rem} W^{-1})W = \mu r_0 + t_{\rm rem} = t_+.
\]
This supplies the positive-weight and scalar-test data $(\tau_+, 1/2, \ell_{*,+}, a_{P,+})$ for the next lift.

\end{proof}
For definiteness, we choose $T_2$ to be the least integer strictly larger than the bound in~\eqref{eq:radius}. At this point every entry of the new payoff matrix $G$ is fixed. To complete the proof of Theorem~\ref{thm:lift}, Section~\ref{sec:chronology} shows that the proposed words are generated by best responses.

\subsection{Best-response verification and proof of the one-step theorem}\label{sec:chronology}
We complete the proof of Theorem~\ref{thm:lift} by showing that the fixed game constructed in Sections~\ref{sec:synthesis}--\ref{sec:next} generates the prescribed best-response trajectory. The exact score identities and quantitative margins compare each proposed play with every competing action. In Proposition~\ref{prop:chronology}, we prove that the full game follows the word $\Aw(q)\Bw(q)$. In Proposition~\ref{prop:inherit}, we move the phase boundary to the start of the reset block and verify all next-order clock conditions. With these ingredients, we give the proof of Theorem~\ref{thm:lift} in Section~\ref{sec:liftproof}.

\subsubsection{Checking every proposed run against every action}\label{sec:valid-runs}
Recall the fixed data available at this point. The old actions are the $m$ parent actions $0,\ldots,m-1$ and the endpoint action $g=m$. There are $2r$ reset actions $\reset{0},\ldots,\reset{2r-1}$, where $r=k+1$. Their full payoff matrix is $G$ from~\eqref{eq:controllerdata}, with old-action block $G_O$, coupling block $P$, and reset block $K=J-J^{\top}$. The matrices $P,M,J$ satisfy the two reset equations~\eqref{eq:twosided} and the signs and margins of Lemma~\ref{lem:uniform}.

The tables $S$ from~\eqref{eq:old} and $V$ from~\eqref{eq:controllerdata} prescribe the scores at the start of outer phase $q$:
\[
 \mathbf s(q)=Sp_k(q)\in\Q^{m+1},\qquad
 \mathbf v(q)=Vp_k(q)\in\Q^{2r},\qquad
 \mathbf u(q)=\col(\mathbf s(q),\mathbf v(q)).
\]
Thus, $\mathbf s(q)$ lists the old-action scores and $\mathbf v(q)$ lists the reset-action scores. The vector $\mathbf v(q)$ is distinct from the parent counter row $v$ in Definition~\ref{def:clock}. The count vector $\mathbf m(q)=Mp_k(q)$ gives the reset run lengths. The prescribed word first plays the parent bundle followed by $g$,
\[
 \Aw(q)=\W_k(\tau q)\cdots\W_k(2\tau q-1)\,g,
\]
and then the reset block
\[
 \Bw(q)=\reset{0}^{\mathbf m_0(q)}\cdots
          \reset{2r-1}^{\mathbf m_{2r-1}(q)}.
\]
Lemma~\ref{lem:resetidentity} already shows that adding the payoff columns for this entire word takes $\mathbf u(q)$ to $\mathbf u(q+1)$. We must now show that its successive actions are the best responses selected by~\eqref{eq:process}.

Give the actions the priority order
\begin{equation}\label{eq:priority}
  \text{parent actions in their inherited order}
  \ \prec\ \reset{2r-1}\ \prec\ \reset{2r-2}\ \prec\cdots\prec\ \reset{0}\ \prec\ g.
\end{equation}
Here $x\prec y$ means that $x$ wins a tie against $y$. In particular, any parent action wins a tie against any added action; among reset actions, the later reset action wins a tie. These choices determine the boundary handoffs $\reset{j}\to\reset{j+1}$ and $\reset{2r-1}\to b$.

\begin{proposition}[Every proposed run is valid]\label{prop:chronology}
There is an integer threshold $Q_{\rm out}$ such that, for every integer $q\ge Q_{\rm out}$, the best-response updates starting from $\mathbf u(q)=\col(\mathbf s(q),\mathbf v(q))$ follow the exact word $\Aw(q)\Bw(q)$ and end at $\mathbf u(q+1)$. During the parent bundle no added action interrupts the parent. Action $g$ plays once at the bundle endpoint, followed by exactly $\mathbf m_j(q)$ consecutive plays of each reset action $\reset{j}$ in increasing order of $j$. Throughout this reset block, including its entry and endpoint, the maximum score among all actions equals $\zeta(q)$. Here $q$ is fixed, so this maximum does not change as reset actions are played.
\end{proposition}
\begin{proof}
We follow the sequence in Figure~\ref{fig:word}: the parent bundle, the single play of $g$, the early reset runs, and the late reset runs. First we choose $Q_{\rm out}$ to make all run lengths positive and to allow the bundle-height bound. We then show that the added actions wait during the parent bundle, that $g$ hands play to the first reset action, and that the reset actions execute in order while every old action waits. All comparisons are evaluated after a proposed prefix; the final paragraph turns them into an induction on the number of plays.

\paragraph{Choose the phase threshold $Q_{\rm out}$.}
Use $\epsilon=a_{\min}/4$ from Lemma~\ref{lem:uniform}. Let $Q_{\rm bun}(\epsilon)$ be a threshold from Lemma~\ref{lem:bundle}. Each reset-action count polynomial has a positive degree-$k$ coefficient. Hence the following explicit estimate gives an integer $Q_{\rm len}\ge1$ such that $\mathbf m_j(q)\ge2$ for every reset action and all integers $q\ge Q_{\rm len}$. Explicitly, a polynomial $f(q)=a_Dq^D+\cdots+a_0$ with $D\ge1$ and $a_D>0$ is positive when $q\ge1$ and $q>\sum_{i<D}|a_i|/a_D$, because $f(q)\ge q^{D-1}(a_Dq-\sum_{i<D}|a_i|)>0$. Apply this to $\mathbf m_j(q)-2$. Recall that the integer $\tau\ge1$ is the parent scale from Definition~\ref{def:admissible}, fixed before $T_1$ and $T_2$, and that the bundle in $\Aw(q)$ starts at parent phase $\tau q$. Choose
\begin{equation}\label{eq:Qout}
  Q_{\rm out}=\max\left\{1,Q_{\rm len}+1,
               \left\lceil Q_{\rm bun}(\epsilon)/\tau\right\rceil\right\}.
\end{equation}
Then every entry of $\mathbf m(q)$ and $\mathbf m(q-1)$ is at least two. Also $\tau q\ge Q_{\rm bun}(\epsilon)\ge Q_k$. Therefore the parent clock conditions in Definition~\ref{def:clock} hold at every phase index used in the bundle and at its exit index $2\tau q$. Lemma~\ref{lem:gateway} gives the endpoint-action comparisons, and Lemma~\ref{lem:bundle} gives the remaining-height bound across the entire bundle.

\paragraph{The parent bundle: early reset actions.}
Within the parent block the payoff matrix is still $A$, and the initial score shift is the common scalar $h(q)$. Thus, provided that no added action is selected, condition~(I) of Definition~\ref{def:clock} makes the parent's internal sequence valid. Consider the score vector obtained after any prefix of the prescribed parent bundle, including the empty and complete prefixes. Let $z_i^{\rm rem}$ be the number of plays of parent action $i$, $0\le i<m$, remaining in that bundle.

Fix an early reset action $\reset{j}$, where $0\le j<r$. Its score rises during parent plays because $-P_{ij}>0$, so an initial-score bound would not suffice. We instead work backward from its exact score after $\Aw(q)$. By $P^{\top}X=Q$ in~\eqref{eq:twosided} and the calculation in Lemma~\ref{lem:resetidentity}, its score immediately after the complete word $\Aw(q)$ is $\zeta(q)-(J\mathbf m(q))_j$. From the present prefix to that point, every remaining play of parent action $i$ changes its score by $G_{\reset{j},i}=-P_{ij}$, and the final play of $g$ changes it by $-P_{gj}$. Consequently,
\[
 \zeta(q)-(J\mathbf m(q))_j
 =u_{\reset{j}}-\sum_{i=0}^{m-1}P_{ij}z_i^{\rm rem}-P_{gj}.
\]
Solving this identity for its score at the present prefix gives
\begin{equation}\label{eq:earlybackward}
  u_{\reset{j}}=\zeta(q)-(J\mathbf m(q))_j+P_{gj}
                  +\sum_{i<m}P_{ij}z_i^{\rm rem}.
\end{equation}
The parent endpoint maximum is $\zeta(q)-c$ by Lemma~\ref{lem:gateway}. Lemma~\ref{lem:bundle}, applied with the common shift $h(q)$, gives
\[
  H_{\rm parent}\ge \zeta(q)-c-\mathcal A_k
                         -\epsilon\sum_{i<m}z_i^{\rm rem}.
\]
Subtracting~\eqref{eq:earlybackward} gives
\begin{equation}\label{eq:earlywait}
 H_{\rm parent}-u_{\reset{j}}
 \ge (J\mathbf m(q))_j+(-P_{gj}-c-\mathcal A_k)
       +\sum_{i<m}(-P_{ij}-\epsilon)z_i^{\rm rem}>0.
\end{equation}
Here $J\mathbf m(q)\ge0$ because $J$ is entrywise nonnegative and every run count is positive, $-P_{gj}-c-\mathcal A_k>0$, and $-P_{ij}-\epsilon>0$ by~\eqref{eq:reserve}. Therefore each early reset action has score strictly below the parent maximum at every decision point within the bundle. The same comparison holds after the entire bundle has been played and immediately before $g$ is selected: at that prefix, all $z_i^{\rm rem}=0$ and the strictly positive constant term in~\eqref{eq:earlywait} remains.

\paragraph{The parent bundle: late reset actions and the endpoint action.}
A late reset action $\reset{j}$, $r\le j<2r$, starts with
\[
  \mathbf v_j(q)=\zeta(q-1)-(J^{\top}\mathbf m(q-1))_j\le \zeta(q-1)=\mathbf s_b(q).
\]
For $j<2r-1$, the inequality is strict because $(J^{\top}\mathbf m(q-1))_j$ includes $J_{j+1,j}\mathbf m_{j+1}(q-1)>0$. For $j=2r-1$, it is an equality, and the parent wins the tie by~\eqref{eq:priority}. Each play of a parent action $i$ changes a late reset action's score by $-P_{ij}<0$, while the parent maximum is nondecreasing by~\eqref{eq:monotone}. Therefore every late reset action is strictly below the parent maximum after the first parent play and remains so through the bundle endpoint.

Lemma~\ref{lem:gateway} keeps $g$ strictly below the parent at every proper prefix. At the endpoint, $g$ exceeds the parent maximum by $\delta$ and is the unique maximizer, so it plays once. Immediately after this play, the reset-action scores are $\one \zeta(q)-J\mathbf m(q)$. Thus, $\reset{0}$ has score $\zeta(q)$ and every other reset action is strictly below it. Every old action has score at most $\zeta(q)-c+\delta<\zeta(q)$ by Lemma~\ref{lem:gateway}. This proves the transition to the reset block.

\paragraph{Reset action handoffs.}
During a play of $\reset{j}$, $K=J-J^{\top}$ raises every later reset score by $J_{ij}>0$, lowers every earlier reset score by $J_{ji}>0$, and leaves its own score fixed. Immediately before the reset block, its score vector is $\one\zeta(q)-J\mathbf m(q)$. Consider a prefix consisting of all runs $\reset{0},\ldots,\reset{j-1}$ and then $h$ plays of $\reset{j}$, where $0\le j<2r$ and $0\le h<\mathbf m_j(q)$. Adding the corresponding columns of $K=J-J^{\top}$ shows that the active reset action has score
\[
 u_{\reset{j}}=\zeta(q)-(J\mathbf m(q))_j
                   +\sum_{\ell<j}J_{j\ell}\mathbf m_\ell(q)
              =\zeta(q).
\]
Its own plays leave this score unchanged because $K_{jj}=0$. A previously played reset action $\reset{i}$, $0\le i<j$, has score
\[
  u_{\reset{i}}=\zeta(q)-\sum_{i<\ell<j}J_{\ell i}\mathbf m_\ell(q)-J_{ji}h\le \zeta(q)
\]
and loses any tie to $\reset{j}$. A future reset action $\reset{i}$, $i>j$, has deficit
\[
  \zeta(q)-u_{\reset{i}}=J_{ij}(\mathbf m_j(q)-h)
                        +\sum_{j<\ell<i}J_{i\ell}\mathbf m_\ell(q)>0.
\]
For $h<\mathbf m_j(q)$, the first term in this deficit is strictly positive, so no future reset action can be selected before the current run is complete. If $j<2r-1$, then at $h=\mathbf m_j(q)$ the next reset action $\reset{j+1}$ reaches $\zeta(q)$ and wins the tie against $\reset{j}$ and all earlier reset actions under~\eqref{eq:priority}. Every action $\reset{i}$ with $i>j+1$ is still strictly below $\zeta(q)$, because its remaining deficit includes $J_{i,j+1}\mathbf m_{j+1}(q)>0$. Hence the reset actions must be selected consecutively, in increasing index order, for their full prescribed run lengths, provided that no old action intervenes. We verify that remaining comparison next.

\paragraph{Old actions during the reset block.}
We bound falling old scores forward from the reset entry and rising old scores backward from the reset endpoint. Every old score starts strictly below $\zeta(q)$. All old scores decrease throughout the early half because $P_E<0$. To compare them during the late half, fix an old action $i\in\{0,\ldots,m-1,g\}$. Here $i$ labels a row of $P$, while $j\in\{r,\ldots,2r-1\}$ will label a column of $P$, corresponding to late reset action $\reset{j}$; there is no ordering relation required between $i$ and $j$.

Consider a prefix of $\Bw(q)$ that has completed all early runs but has not yet completed all late runs. For each $r\le j<2r$, let $z_j^{\rm rem}$ be the number of plays of $\reset{j}$ still remaining. All these numbers are nonnegative integers, and
\[
 \sum_{j=r}^{2r-1}z_j^{\rm rem}\ge1,
\]
because at least one prescribed reset play has not yet occurred and all early plays are already complete. The old-score reset identity $\mathbf s_{\mathrm{mid}}(q)+P\mathbf m(q)=\mathbf s(q+1)$ follows from $PM=D$ in~\eqref{eq:twosided}. Each remaining play of $\reset{j}$ adds $P_{ij}$ to the old score $u_i$. Thus, the current score and its prescribed endpoint value satisfy
\[
 \mathbf s_i(q+1)=u_i+\sum_{j=r}^{2r-1}P_{ij}z_j^{\rm rem}.
\]
Equivalently,
\begin{equation}\label{eq:oldlate}
  u_i=\mathbf s_i(q+1)-\sum_{j=r}^{2r-1}P_{ij}z_j^{\rm rem}
                         <\mathbf s_i(q+1)\le \zeta(q).
\end{equation}
The strict inequality holds because every late coupling coefficient $P_{ij}$ is positive and at least one $z_j^{\rm rem}$ is positive. For the last inequality, the parent entries of $\mathbf s(q+1)$ are a common shift of $\mathbf r(\tau(q+1))$, whose selected entry leader is $b$. Therefore all of them are at most $\mathbf s_b(q+1)=\zeta(q)$. The remaining old action is $g$. Its entry row in~\eqref{eq:old} subtracts the number of parent phases still to be played and adds the endpoint margin $\delta$, giving
\[
  \mathbf s_g(q+1)=\zeta(q)-\tau(q+1)+\delta<\zeta(q).
\]
Here $vCF_\tau p_k(q+1)=\tau(q+1)$ because $vC=\erow{k}$ counts one per parent phase, and $0<\delta<c<1$. Thus, no old action can be selected before the reset block is complete. Together with the reset-score comparisons, this proves that its maximum score remains $\zeta(q)$ throughout.

At the final endpoint, $b$ is the selected parent entry leader and wins its tie with the last reset action by~\eqref{eq:priority}. Lemma~\ref{lem:resetidentity} identifies the complete score vector there as $\mathbf u(q+1)$.

\paragraph{Induction over the plays of the proposed word.}
Let $L=|\Aw(q)\Bw(q)|$, and write the proposed word as $a_1\cdots a_L$. We now use induction on $h=0,\ldots,L$ to prove that the actual first $h$ best-response plays are $a_1,\ldots,a_h$ and that their score vector is
\[
 \mathbf u(q)+\sum_{\ell=1}^{h}Ge_{a_\ell}.
\]
The assertion holds at $h=0$, with an empty sum. If it holds for some $h<L$, then the actual state equals the score state at the corresponding proposed prefix. The comparisons above show that $a_{h+1}$ is a maximizer and wins every possible tie under~\eqref{eq:priority}. The next best response is therefore $a_{h+1}$, and adding its payoff column $Ge_{a_{h+1}}$ proves the assertion at $h+1$. At $h=L$, the reset identities give exactly $\mathbf u(q+1)$. Since $q+1\ge Q_{\rm out}$, the same argument applies to the next phase and, successively, to every later phase.
\end{proof}

\subsubsection{Starting the new phase at the first reset action}\label{sec:inherit}
The valid word from Proposition~\ref{prop:chronology} begins with a parent bundle. For the next lift, we want a phase whose endpoint is marked by exactly one final play of $g$. We therefore regard the same infinite trajectory as the concatenation of $\Bw(q)\Aw(q+1)$, beginning at the first reset action; see Figure~\ref{fig:rotation}. Proposition~\ref{prop:inherit} verifies that the rotated tables already used in Section~\ref{sec:next} satisfy every clause of the clock definition and carry the additional data needed to repeat Theorem~\ref{thm:lift}.

\begin{proposition}[The new clock and its additional data]\label{prop:inherit}
With phase words and tables from~\eqref{eq:rotate}, entry action $b_+=\reset{0}$, and counter $v_+=e_g^{\top}$, the enlarged skew-symmetric payoff matrix $G$ from~\eqref{eq:controllerdata} defines an order-$(k+1)$ clock satisfying Definition~\ref{def:admissible}. One may take $Q_{k+1}=Q_{\rm out}$, $\delta_{0,+}=1$, and $\mathcal A_{k+1}=\mathcal A_k+1$.
\end{proposition}
\begin{proof}
We check the three clauses of Definition~\ref{def:clock}, then the additional data. Recall $R_+=\col(S_{\mathrm{mid}},\one_{2r}\zetavec-JM)$ from~\eqref{eq:rotate}; these are the scores just after $\Aw(q)$.

\paragraph{Condition (I): exact trajectory and polynomial growth.}
Proposition~\ref{prop:chronology} makes the infinite concatenation $\Aw(q)\Bw(q)\Aw(q+1)\Bw(q+1)\cdots$ valid. Moving the phase boundary to just after $\Aw(q)$ gives exactly $\W_{k+1}(q)=\Bw(q)\Aw(q+1)$, starting from $R_+p_k(q)$ and ending at $R_+p_k(q+1)$. Its count table is $C_+=\col(X\mathsf{T},M)$, with integer coefficients and nonnegative evaluations for $q\ge Q_{\rm out}$. The identity $GC_+=R_+N$ says that the payoff change over one rotated phase is precisely the difference between its next and current entry-score tables. Indeed, evaluation at $p_k(q)$ gives
\[
 GC_+p_k(q)=R_+\bigl(p_k(q+1)-p_k(q)\bigr).
\]
The late reset-action submatrix of $C_+$ is $W$ from~\eqref{eq:W}. Formula~\eqref{eq:inverse}, valid for the chosen parameters, gives its inverse, so $\rank(C_+)=k+1$. The coefficient row of the new entry height is $\zetavec\in\Q^{1\times(k+1)}$; its scalar last entry $\zetavec_k>0$ is $k!$ times the ordinary degree-$k$ coefficient.

\paragraph{Condition (II): entry leadership and endpoint detection.}
The phase begins with $\mathbf m_0(q)\ge2$ plays of $\reset{0}$. Its only play of $g$ is the last symbol of $\Aw(q+1)$. Therefore $v_{+,b_+}=0$, $v_+C_+=\erow{w}$, and $v_+z^{\rm rem}=1$ at every proper prefix. Since the current maximum is at least the current score of $\reset{0}$, this proves~\eqref{eq:marker} with margin $1$. At entry, $\reset{0}$ remains a maximizer after one play of itself because its run has length at least two. The support of $v_+$ is exactly $\{g\}$. At this coordinate the vector required by~\eqref{eq:supportsafe} has value
\[
 \bigl(R_+p_k(q)+Ge_{\reset{0}}-e_g\bigr)_g
 =\zeta(q)-c+\delta+P_{g0}-1<\zeta(q),
\]
since $\delta<c$ and $P_{g0}<0$. Its $\reset{0}$ coordinate remains $\zeta(q)$, because $G$ is skew-symmetric and $g\ne\reset{0}$. All other coordinates agree with the valid state after the first reset play, so the required maximum and strict support comparison both hold.

\paragraph{Condition (III): bound on the remaining height gain.}
Consider any prefix of $\Bw(q)\Aw(q+1)$. If it lies in the reset block, then the current maximum is $\zeta(q)$ and the remaining word contains all of $\Aw(q+1)$. This latter word has length of degree $k$ with positive leading coefficient, while $\zeta(q+1)-\zeta(q)$ has degree $k-1$. For each rational $\epsilon>0$, the polynomial $\epsilon|\Aw(q+1)|-[\zeta(q+1)-\zeta(q)]$ has positive leading coefficient. The coefficient bound used in the proof of Proposition~\ref{prop:chronology} therefore gives a threshold above which the remaining height gain is at most $\epsilon$ times the remaining length.

If the prefix lies in the parent bundle of $\Aw(q+1)$, then Lemma~\ref{lem:bundle} bounds the remaining gain to its endpoint $\zeta(q+1)-c$ by $\epsilon$ times the remaining parent length plus $\mathcal A_k$. The full phase ends at maximum $\zeta(q+1)$, so its remaining gain is at most
\[
  \epsilon|\W_{k+1}^{\rm rem}|+\mathcal A_k+c
  \le\epsilon|\W_{k+1}^{\rm rem}|+\mathcal A_k+1.
\]
At the prefix immediately before $g$, the remaining gain is $c-\delta<1$; at the complete prefix it is zero. These cases also include the empty prefix. Take the maximum of $Q_{\rm out}$, the coefficient-based polynomial threshold for the reset-prefix case, and $\lceil Q_{\rm bun}(\epsilon)/\tau\rceil$. For every larger outer index, the bundled phases in $\Aw(q+1)$ meet the parent bundle threshold. This gives a terminating threshold algorithm using exact rational arithmetic for each rational $\epsilon>0$.

\paragraph{Additional data for the next lift.}
Proposition~\ref{prop:nextweights} supplies the integer scale factor, correction row, reference value $1/2$, and strictly positive weights for these same tables. The correction satisfies~\eqref{eq:higher} for the rotated algebraic tables, and its weights satisfy $a_{P,+}^{\top}C_+=t_{\tau_+,1/2}(\zetavec)$ with every entry positive. These are precisely items (i) and (ii) of Definition~\ref{def:admissible} for the new height row $\zetavec$, so that definition holds.
\end{proof}

\subsubsection{Proof of the one-step construction theorem}\label{sec:liftproof}
We are ready to prove Theorem~\ref{thm:lift}. Starting from an order-$k$ clock with the additional data of Definition~\ref{def:admissible}, we must retain its matrix as a principal block, add exactly $2k+3$ actions, and verify the three order-$(k+1)$ clock conditions together with the additional data for another lift. We must also obtain the stated endpoint counter, margins, and reference gain.

\begin{proof}[Proof of Theorem~\ref{thm:lift}]
Starting from a parent with the data in Definition~\ref{def:admissible}, the formulas of Section~\ref{sec:gateway} construct its endpoint extension. Proposition~\ref{prop:momentchoice} solves the matrix compatibility equation, and the formulas of Section~\ref{sec:synthesis} realize it as $M^{\top}JM$ and solve both reset equations. Lemma~\ref{lem:uniform} guarantees the required signs and margins for all sufficiently large scales, uniformly in $T_2$. Proposition~\ref{prop:nextweights} applies Lemma~\ref{lem:scale} to the invariant rotated tables to obtain the next higher-test correction, proving Definition~\ref{def:admissible}(i); it then fixes $T_2$ and constructs the strictly positive target representation required by item (ii). Finally, Propositions~\ref{prop:chronology} and~\ref{prop:inherit} prove validity and all inductive conditions. Exactly $1+2(k+1)=2k+3$ actions have been added, and the original matrix is unchanged in its principal block. The degree arguments here require only $k\ge2$: the top coefficient of $\erow{k}F_\tau$ is zero, while $\chivec_{k-1}\tau^{k-1}(2^{k-1}-1)$ is positive. The reset equations, their compatibility tests, the sign bounds, the positive-weight construction in Proposition~\ref{prop:nextweights}, and the best-response comparisons in Proposition~\ref{prop:chronology} all apply for $k=2$. For the RPS parent, the parity adjustment is the explicit correction~\eqref{eq:rpscorrection}, and Proposition~\ref{prop:seedadmissible} supplies the positive-weight data. Proposition~\ref{prop:inherit} gives $v_+=e_g^{\top}$, $\delta_{0,+}=1$, and $\mathcal A_{k+1}=\mathcal A_k+1$, while Proposition~\ref{prop:nextweights} gives $c_{{\rm ref},+}=1/2$. This completes the proof of the theorem.
\end{proof}

\appendix
\section{Explicit reset-scale bounds and rate constants}\label{app:bounds}
Appendix~\ref{app:reset-scale} gives explicit sufficient bounds for choosing the reset-count scale $T_1$. Appendix~\ref{app:constant-size} gives the leading gap constant from the coefficient tables and explains payoff normalization.
\subsection{Sufficient bounds for the reset-count scale}\label{app:reset-scale}
Lemma~\ref{lem:uniform}, Step 5, proves that a sufficiently large reset-count scale $T_1$ gives every required sign and waiting margin, uniformly over the later choice of $T_2$. We provide explicit bounds for choosing such a $T_1$ from the parent coefficient tables. The compatibility equations themselves have already been solved in Proposition~\ref{prop:momentchoice} and Section~\ref{sec:resetinteractions}; only inequalities for the signs remain to be checked.

Recall from~\eqref{eq:coefficientnorm} that $\norm{Z}_*=\max\{\norm{Z}_1,\norm{Z}_\infty\}$ and $\norm{f}_1=\sum_j|f_j|$ for a vector $f$. For a candidate positive integer value of $T_1$, define
\begin{align*}
  b_W &= (\norm{\Phi_0}_* + \norm{\Phi_1}_*)/T_1, &
  \nu &= b_W \norm{\ell_*}_1, \\
  b_E &= b_W / (1 - \nu), &
  b_C &= \norm{\one_r \lambda \zetavec^\circ}_* + \norm{V_0}_* + \norm{\ell_*}_1 / T_1, \\
  b_R &= \norm{E - E_{\rm base}}_* + 2\norm{L_r}_* b_C^2 / T_1, &
  b_\times &= b_W b_R b_E,
\end{align*}
and
\begin{align*}
  b_Q &= 2r\norm{\psi-\psi^\circ}_1 + \kappa r \norm{\ell_*}_1 + \frac{4\norm{L_r}_* b_C (1 + \norm{B}_*)}{T_1^2} + b_R(b_E \norm{B}_* + b_W), \\
  b_P &= \norm{D + a\ell_*}_* b_W + \norm{L_X}_*\bigl(b_Q + \norm{X}_* \norm{D + a\ell_*}_* b_W\bigr).
\end{align*}
For any $T_1$ with $\nu < 1/2$ and $T_2 \ge \max\{T_1, \upsilon_k\}$, these quantities bound $\norm{C_\times - \kappa \one_r \one_r^{\top}}_* \le b_\times$, $\norm{Q - P_0^{\top} X}_* \le b_Q$, and $\norm{P - P_0}_* \le b_P$. To choose $T_1$, it is sufficient to check
\begin{equation}\label{eq:explicit-sign-tests}
\begin{gathered}
 0<c<1,\qquad \upsilon_k+\ell_{*,k}>0,\qquad \nu<\tfrac12,\\
 b_\times<\kappa/2,\qquad
 b_P<\min\{a_{\min}/4,(a_g-\mathcal A_k-1)/2\}.
\end{gathered}
\end{equation}
When $k$ is even, if the denominator in~\eqref{eq:parity} is nonpositive, then discard that candidate integer $T_1$.  Define $b_E$ and the bounds depending on it only after verifying $\nu<1/2$. Every sufficiently large $T_1$ passes all these tests, by Lemma~\ref{lem:uniform}. One may therefore test successive integers, or double $T_1=1,2,4,\ldots$ for convenience, until a value satisfies all the requirements.

\subsection{The leading constant and payoff normalization}\label{app:constant-size}
Lemma~\ref{lem:alltime} proves the sharper form of the rate,
\[
 \Gap_{\widehat A_k}(z_t/t,z_t/t)
   =c_k t^{-1/k}(1+o(1)),\qquad c_k>0.
\]
Here $c_k$ is the leading constant multiplying $t^{-1/k}$. We give its exact formula and explain the effect of payoff normalization. For a fixed order-$k$ clock, put
\[
 \Lambda_k=(\one^{\top}C_k)_{k-1}>0,\qquad
 \xi_k=(R_k)_{b,k-1}>0,\qquad
 M_k=\max_{i,j}|(\widehat A_k)_{ij}|>0.
\]
Lemma~\ref{lem:alltime} gives the exact positive rational power
\begin{equation}\label{eq:constant-rational-power}
 c_k^k=\frac{2^k k^{k-1}\xi_k^k}{(k-1)!\Lambda_k^{k-1}}.
\end{equation}
Thus, the constants can be computed from the exact coefficient tables and
payoff matrix including setup. The constant for the normalized matrix
$\widehat A_k/M_k$ is $c_k/M_k$. The fixed bias,
common score translation, and choice of starting phase do not change
$c_k$, because they do not change the leading phase-length and height
coefficients. They can change $M_k$, and hence the normalized constant.
For the displayed order-two game, $c_2=2\sqrt{2}/3$ and $M_2=313/36$, so $c_2/M_2=24\sqrt{2}/313$.

For the particular games used in Appendix~\ref{sec:exact-computations}, the normalized constants beside this exact formula are
\[
\begin{array}{c|c}
k&\widetilde c_k=c_k/M_k\\\hline
2&0.108438101\\
3&3.22228726\times10^{-11}\\
4&1.70811213\times10^{-56}\\
5&1.37443391\times10^{-316}\\
6&8.09268843\times10^{-1596}
\end{array}
\]
These rounded values were independently recomputed at 100-digit precision from~\eqref{eq:constant-rational-power}, including the phase-to-time conversion. They describe these particular games; the constants depend on the chosen matrices.

The asymptotic rate $t^{-1/k}$ need not be accurate near the start of the trajectory. Lower-order terms and the chosen starting phase can produce a long transient, although Lemma~\ref{lem:alltime} guarantees convergence to the stated law for each fixed game. Appendix~\ref{sec:exact-computations} gives concrete numerical examples.

\section{Explicit matrices for orders two and three}\label{app:examples}
We display explicit order-$2$ and order-$3$ matrices to illustrate one application of the lift. The order-$3$ entries span a wide range because we use conservative sufficient parameter choices. The matrices here include the setup action, listed first and denoted by $\omega$. Both players start from zero counts and choose $\omega$ in the first round. Thereafter the best responses are unique. The matrices are displayed before the optional normalization to $[-1,1]$; multiplying every entry by the same positive number preserves the sequence of play.

\subsection{Order two: rock--paper--scissors with setup}
The recurrent block is the familiar three-action RPS matrix $A_{\mathrm{RPS}}$ in Section~\ref{sec:seed}. At phase $q=1$, its entry score is $(2,-17/3,5/3)^{\top}$. The bias in Lemma~\ref{lem:bias} is $(2,1,0)^{\top}/36$. Adding the common constant $239/36$ makes its smallest coordinate one. In the action order $(\omega,0,1,2)$, the resulting zero-start game is therefore
\[
 \widehat A_2=\frac1{36}
 \begin{pmatrix}
 0&-313&-36&-299\\
 313&0&-36&36\\
 36&36&0&-36\\
 299&-36&36&0
 \end{pmatrix}.
\]
The setup play gives the three recurrent scores $(313/36,1,299/36)^{\top}$. The recurrent word starts with $0^8 1^9 2^{10}$ and then follows the successive RPS phases. Thus, three actions generate the square-root clock, and a fourth produces its initial scores from zero counts. To normalize this displayed matrix to $[-1,1]$, replace its denominator $36$ by $313$.

\subsection{Order three: the first application of the general construction}
Here the action order is
\[
 (\omega,\underbrace{0,1,2,g}_{\text{RPS and endpoint actions~}},
       \underbrace{\reset{0},\reset{1},\reset{2}}_{\text{early reset actions~}},
       \underbrace{\reset{3},\reset{4},\reset{5}}_{\text{late reset actions}}).
\]
One set of parameter choices satisfying the construction is $\tau=1$, $T_1=524288$, $T_2=456247107$, and the starting phase is $q=1733$. The actual height gain is $c=1048594/2097161$, as given by~\eqref{eq:rpscorrection}. The integer scale below is
\[
 \Gamma_3=518388790765424836153943053839756165120\approx5.184\times10^{38}.
\]
The large integers below come from conservative sufficient parameter bounds; they are not a proved necessity for a cubic rate. We display the matrix produced by the same parameter-selection procedure used at all orders.

Write $f^{(3)}=\col(f_O,f_E,f_L)$, with the old actions
$O=(0,1,2,g)$, early reset actions $E=(\reset{0},\reset{1},\reset{2})$, and late reset actions
$L=(\reset{3},\reset{4},\reset{5})$. Split $U=[U_E\mid U_L]$ into its first and last three columns.
The integer numerator $F_3$ and the payoff matrix are
\begin{equation}\label{eq:explicitcubic}
 \widehat A_3=\frac{F_3}{\Gamma_3},\qquad
 F_3=\begin{pmatrix}
 0&-f_O^{\top}&-f_E^{\top}&-f_L^{\top}\\
 f_O&\Gamma_3H&41943040\,U_E&41943040\,U_L\\
 f_E&-41943040\,U_E^{\top}&\Lambda_3 S&-120\,C\\
 f_L&-41943040\,U_L^{\top}&120\,C^{\top}&\Lambda_3 S
 \end{pmatrix}.
\end{equation}
This is the old/early/late block structure of Figure~\ref{fig:matrix}, with
one setup row and column added. Its block sizes are $1,4,3,3$, totaling eleven.
In particular,
\[
 f_O\in\Z^{4\times1},\quad f_E,f_L\in\Z^{3\times1},\quad
 H\in\Z^{4\times4},\quad U_E,U_L\in\Z^{4\times3},\quad S,C\in\Z^{3\times3}.
\]
Here
\[
 H=\begin{pmatrix}0&-1&1&0\\1&0&-1&2\\-1&1&0&-2\\0&-2&2&0\end{pmatrix},\qquad
 S=\begin{pmatrix}0&-1&-1\\1&0&-1\\1&1&0\end{pmatrix},\qquad
 \Lambda_3=4749910616309760 \approx 4.75 \times 10^{15}.
\]
The upper-left $3\times3$ block of $H$ is $A_{\mathrm{RPS}}$, and its fourth action is $g$.
Tables~\ref{tab:explicitsetup}--\ref{tab:explicitcross} specify every entry of
$f^{(3)}\in\Z^{10\times1}$, $U\in\Z^{4\times6}$, and $C$.
Thus, $F_3\in\Z^{11\times11}$ is completely specified by these tables and the
displayed signs and transposes. The common denominator $\Gamma_3$ clears the
rational construction coefficients. Every actual payoff is the corresponding entry of $F_3$ divided by $\Gamma_3$.

\paragraph{The sizes of the actual payoffs.}
Table~\ref{tab:cubicmagnitudes} gives enclosing intervals for the nonzero absolute entries of $\widehat A_3=F_3/\Gamma_3$. Each lower decimal bound is rounded down and each upper bound up, then verified by exact rational comparison. Thus, the integer numerators in the preceding display do not by themselves describe the payoff scale.
\begin{table}[htbp]\centering\small
\begin{tabular}{@{}lll@{}}\toprule
Block of $\widehat A_3$ & Lower bound & Upper bound\\\midrule
Setup row and column & $1$ & $3.006778\times10^6$\\
Old actions, $H$ & $1$ & $2$\\
Parent / reset & $0.0740742$ & $0.0740746$\\
Endpoint $g$ / reset & $7.4999934$ & $7.4999939$\\
Within either reset half & $9.1628\times10^{-24}$ & $9.1629\times10^{-24}$\\
Between the two reset halves & $6.3578\times10^{-7}$ & $6.3579\times10^{-7}$\\\bottomrule
\end{tabular}
\caption{Magnitudes after division by $\Gamma_3$, before any optional normalization to $[-1,1]$. Zeros are excluded; the bounds need not be attained.}\label{tab:cubicmagnitudes}
\end{table}
The largest absolute entry is the setup payoff $f^{(3)}_{\reset{0}}/\Gamma_3\approx3.006777059\times10^6$. The largest recurrent payoff has magnitude approximately $7.499993800$, in the $g,\reset{1}$ pair. The smallest nonzero magnitude is exactly
\[
 \frac{\Lambda_3}{\Gamma_3}
 =\frac{1}{109136535956157601677312}
 \approx9.162834345\times10^{-24}.
\]
To apply the optional normalization in Remark~\ref{rem:normalization}, divide every actual payoff, including the setup row and column, by the common maximum $f^{(3)}_{\reset{0}}/\Gamma_3$.

The setup column already includes the bias and common score translation from Lemmas~\ref{lem:bias} and~\ref{lem:setup}. In recurrent action order, the bias is
\[
 (9,8,7,0,1,2,3,4,5,6)^{\top}/\Gamma_3.
\]
After setup, the score vector is exactly $f^{(3)}/\Gamma_3$ and $\reset{0}$ is the unique maximizer. No additional bias is to be applied to the displayed game.

\begin{table}[H]\centering\small
\begin{tabular}{@{}cr@{}}\toprule
Action & $f^{(3)}_i$\\\midrule
$0$ & $1558679264526474926416054178185188930087813124$\\
$1$ & $1553287675509987331170162401129886239466455043$\\
$2$ & $1558678054952629807091436485651396637323427842$\\
$g$ & $1558679394125341123117187434527396943217295355$\\
$\reset{0}$ & $1558679523724207319818320690869604956346777596$\\
$\reset{1}$ & $1558679523724207312328348104430061311848611837$\\
$\reset{2}$ & $1558679523724207301084901201480264859873443838$\\
$\reset{3}$ & $748574605883120608380675033545425551359$\\
$\reset{4}$ & $518388790765424836153943053839756165120$\\
$\reset{5}$ & $541477418431223810551440457239908843521$\\
\bottomrule\end{tabular}
\caption{The setup column numerator $f^{(3)}$, in recurrent action order.}\label{tab:explicitsetup}
\end{table}

\begin{table}[H]\centering\footnotesize
\setlength{\tabcolsep}{4pt}
\begin{tabular}{@{}ccrr@{}}\toprule
$i$ & $j$ & $U_{i,j}$ & $U_{i,j+3}$\\\midrule
$0$ & $0$ & $-915510046385366433976412132953$ & $915510055415203425882643486297$\\
$0$ & $1$ & $-915512040816756456584520511936$ & $915512036301837960631404835264$\\
$0$ & $2$ & $-915512342182560758868922202866$ & $915512337667642262915806526194$\\
$1$ & $0$ & $-915511408738655678979274036825$ & $915511300382123072556775364185$\\
$1$ & $1$ & $-915511503550723431475555980889$ & $915511598362733057185211864665$\\
$1$ & $2$ & $-915511517095304538975024830041$ & $915511530639827519687867618905$\\
$2$ & $0$ & $-915512771091944923982135940697$ & $915512816240664869220284225113$\\
$2$ & $1$ & $-915510966284690406366591449842$ & $915510943710330433747517307634$\\
$2$ & $2$ & $-915510692008048319081127457216$ & $915510669433688346462053315008$\\
$g$ & $0$ & $-92695058168523872867103415486770$ & $92695058168523989120676667607346$\\
$g$ & $1$ & $-92695062556949563982457323108046$ & $92695062543404924748171228198606$\\
$g$ & $2$ & $-92695058154979362684263029782222$ & $92695058168523885664975872571086$\\
\bottomrule\end{tabular}
\caption{The old-action/reset-action block $U$. Each row gives an early column and the corresponding late column. The row label $g$ means row $3$.}\label{tab:explicitcouplings}
\end{table}

\begin{table}[H]\centering\small
\begin{tabular}{@{}ccr@{}}\toprule
$i$ & $j$ & $C_{ij}$\\\midrule
$0$ & $0$ & $2746533971855354434308184310539$\\
$0$ & $1$ & $2746534374406625756888725029643$\\
$0$ & $2$ & $2746534332385920840451380052747$\\
$1$ & $0$ & $2746534253313543388392959345419$\\
$1$ & $1$ & $2746534659366843049941790982923$\\
$1$ & $2$ & $2746534619276124444945253499659$\\
$2$ & $0$ & $2746534290945076359443142509323$\\
$2$ & $1$ & $2746534701351235547861886010123$\\
$2$ & $2$ & $2746534661561429022635371889419$\\
\bottomrule\end{tabular}
\caption{The cross block $C$ between the early and late reset actions.}
\label{tab:explicitcross}
\end{table}

The unnormalized maximum absolute payoff is $F_{\max}/\Gamma_3$, where $F_{\max}=\max_{i,j}|(F_3)_{ij}|$. Hence normalization to $[-1,1]$ gives $F_3/F_{\max}$. Explicitly, replace $\Gamma_3$ in~\eqref{eq:explicitcubic} by
\[
 F_{\max}=1558679523724207319818320690869604956346777596\approx1.56\times10^{45},
\]
while leaving the integer blocks unchanged. This is the largest absolute entry of $F_3$. Before normalization, the largest absolute payoff is $F_{\max}/\Gamma_3\approx3.006777\times10^6$; division by $F_{\max}$ instead places every payoff in $[-1,1]$.

The same conservative parameter choices give a $20\times20$ order-four game and a $31\times31$ order-five game. In these particular outputs, the largest numerator has $320$ and $1914$ decimal digits. These sizes result from the sufficient parameter bounds used in our construction; we do not claim that such large entries are required. The block structure is already visible in~\eqref{eq:explicitcubic}, so we do not expand those larger matrices in the paper.
\FloatBarrier

\clearpage
\section{Exact event-driven empirical simulations}
\label{sec:exact-computations}
The reproduction code is available in the accompanying repository at \url{https://github.com/awibisono/fictitious-play-experiments}.

We check empirically how the constructed trajectories approach the predicted
$t^{-1/k}$ rates. The conservative parameter choices in the inductive
construction can make this regime occur after an enormous number of rounds.
For example, completing just the first outer phase takes about
$10^{67}$, $10^{351}$, and $10^{1742}$ rounds for $k=4,5,6$.
Updating one round at a time is therefore impractical. Instead, we group rounds
whose best responses can all be verified exactly, and add their action counts
in one operation. This still gives an exact simulation of fictitious play, without approximating any action choice or score update.

We construct the rational matrices by the inductive steps above, for
$k=2,3,4,5,6$, with $4,11,20,31,44$ actions including setup.
Figure~\ref{fig:empirical-low-orders} shows the order-two and order-three
trajectories; Figure~\ref{fig:empirical-high-orders} gives corresponding
phase-boundary observations for orders four through six.
For $k=2,3$ (Figure~\ref{fig:empirical-low-orders}), the computation groups constant-action runs and consecutive RPS cycles. For $k=4,5,6$ (Figure~\ref{fig:empirical-high-orders}), it additionally proposes an entire nested sequence of runs and verifies every decision in that finite sequence by exact polynomial inequalities. This latter procedure is called \emph{certified hierarchical replay} below. The empirical results agree with the theoretical rates proved in Theorem~\ref{thm:main}.

\subsection{Simulation state, exact arithmetic, and payoff scaling}
For each fixed order $k$, let $D_k$ be a positive integer clearing all
denominators of its constructed matrix $\widehat A_k$, including the setup row
and column, and put $B_k=D_k\widehat A_k\in\mathbb Z^{n_k\times n_k}$.
At zero counts all actions tie; both players take the prescribed setup action
$\omega$. Thus, after that first round, the common count column is
$z_1=e_\omega$, the integer score column is $\mathbf u_1=B_kz_1$, and $t=1$.
Thereafter skew-symmetry and equal counts give the same unique best response
$i_*$ to both players. One simultaneous round replaces
\[
 z\longleftarrow z+e_{i_*},\qquad
 \mathbf u\longleftarrow\mathbf u+B_ke_{i_*},\qquad
 t\longleftarrow t+1.
\]
Every state satisfies $\mathbf u=B_kz$ and $t=\mathbf1^\top z$.
In particular, the initial phase label $Q_k$ contributes no fictitious past
rounds: only the single actual setup round is included in $t$.

The code stores matrix entries and polynomial coefficients as exact integer
ratios, using Python arbitrary-precision integers, \texttt{Fraction}, and
SymPy rational arithmetic. Clearing denominators makes the trajectory scores
and counts integers; their bit lengths grow as needed rather than being
restricted to a fixed machine word. Run lengths, score comparisons,
polynomial minima, and state updates use this exact arithmetic. Floating-point
logarithms and finite-precision decimal roots are used to display results and
to specify optional reporting or stopping horizons, never to choose or certify
a best response. We emphasize this distinction because the constructed payoffs and elapsed times greatly exceed the precision or range of ordinary floating-point representations. Computing scores and comparisons with ordinary fixed-precision floating-point arithmetic would quickly lose the small score differences needed to follow the true fictitious-play trajectory.

The empirical gap for the original payoff matrix is already time-averaged:
\[
 g_k(t)=\frac{2\max_i u_i}{D_kt}.
\]
A separate, optional normalization puts the payoff entries in $[-1,1]$:
\[
 M_k=\max_{i,j}|(\widehat A_k)_{ij}|,\qquad
 \widetilde g_k(t)=\frac{g_k(t)}{M_k},\qquad
 \widetilde c_k=\frac{c_k}{M_k}.
\]
Dividing by $t$ forms empirical strategies; dividing by $M_k$ changes payoff
units. The latter leaves every best response and elapsed time unchanged and
lets us display games with a common maximum-entry payoff scale.
We use this normalization when plotting the gaps; the trajectory itself is computed with the exact integer matrix $B_k$.

\subsection{Skipping rounds while preserving the exact trajectory}
\paragraph{Constant-action runs.}
Suppose $i_*$ is the unique maximizer of the current integer scores. The first
integer time at which it ceases to be a strict maximizer is
\[
 \ell=\min_{j:(B_k)_{ji_*}>0}
 \left\lceil\frac{u_{i_*}-u_j}{(B_k)_{ji_*}}\right\rceil.
\]
For the constructed trajectory, at least one rival has $(B_k)_{ji_*}>0$, so the set in this minimum is nonempty. Without such a rival, the current leader would remain best forever. Because
$(B_k)_{i_*i_*}=0$, every leader-minus-rival score difference is affine in the
within-run offset. Strict positivity at offsets $0$ and $\ell-1$ therefore
checks all $\ell$ decisions, against every constructed action, including setup
and the outer reset actions. The algorithm adds $\ell e_{i_*}$ to the counts
and $\ell B_ke_{i_*}$ to the scores, then determines the next maximizer afresh.
At the new state, the program checks that the next maximizer is unique before continuing.

\paragraph{Consecutive inherited RPS cycles.}
A further grouping is available whenever the trajectory enters its inherited
three-action RPS subcycle. A \emph{cycle} here is one run of each of those
three actions, in their cyclic order. Although every order contains these
three inherited actions, an RPS cycle is only a small subword of a higher-order
phase. This grouping is used for the reported $k=2,3$ runs, and it always checks
all other actions too.

Write $\iota_0,\iota_1,\iota_2$ for their coordinates in the constructed game.
The program first finds their three current run lengths $\ell_0,\ell_1,\ell_2$
by the preceding exact rule. It proposes lengths $\ell_i+3r$ in cycle $r$,
and the corresponding cycle-entry score column
\[
 \mathbf u(r)=\mathbf u(0)+r d_0+\binom r2 d_1,\qquad
 d_0=\sum_{i=0}^2\ell_i B_ke_{\iota_i},\qquad
 d_1=3\sum_{i=0}^2B_ke_{\iota_i}.
\]
Every first- and last-decision margin is then a quadratic
$p(r)=\alpha r^2+\beta r+\gamma$. A proposed $N$-cycle batch is accepted
only when all these margins are positive at every integer $0\le r<N$.
For $\alpha>0$, the minimum on an integer interval is attained at an endpoint
or one of the integers adjacent to $-\beta/(2\alpha)$; for $\alpha\le0$,
endpoint tests suffice. We locate the first integer $r$ where any margin is nonpositive using integer square roots and exact comparisons. This is the
\emph{first failure} of the proposed RPS pattern, not a numerical error in the
simulation. Only the preceding certified cycles are grouped. Consequently,
the proposed lengths are justified by checks of the actual matrix, rather
than by assuming a phase-return theorem.

\paragraph{Batch caps and stopping budgets.}
A cap limits how many cycles one computational event may group; it does not skip the remaining rounds or alter their decisions. We choose these caps and the total event budgets to control computation time and sample spacing; they are numerical settings, not parameters required by the proof. For $k=2$, the cap at time $t$ is
\[
 \min\!\left\{2^{40},\,
 \max\!\left\{2,\,\left\lfloor\frac{\lfloor\sqrt{t-1}\rfloor}{100}\right\rfloor+1\right\}\right\}.
\]
The growing cap allows larger batches as time increases, while limiting the distance between saved observations. We stop after $100000$ events. Each accepted batch reaches this cap before the RPS pattern fails, and the next event resumes from the exact resulting state.

For $k=3$, we use the fixed cap $2^{20000}$ to allow batching up to the first failure of the proposed inherited RPS pattern. This deliberately large cap does not restrict any accepted batch: the largest has $251733$ cycles. We stop after $2000000$ events. An event may represent many rounds, so these event budgets are not elapsed fictitious-play times. The reproduction documentation specifies these settings and the sampling cadence.

\subsection{How far the computations ran}
Table~\ref{tab:empirical-horizons} reports the last actual completed round and
the number of complete outer phases, rather than a prescribed equal-time
comparison across games. Here $Q_k$ is the starting phase label chosen by the
construction. In particular,
\[
 \begin{gathered}
 Q_2=1,\qquad Q_3=1733,\qquad Q_4=130490910,\\
 Q_5\approx8.7054945636\times10^{44},\qquad
 Q_6\approx8.1896038987\times10^{172}.
 \end{gathered}
\]
The last two are rounded descriptions of the exact integers in the supplied
construction files. The runs with $k\in\{4,5,6\}$ were certified through phase
$4096Q_k$, comprising $4095Q_k$ complete phases after setup.

\begin{table}[htbp]\centering\small
\begin{tabular}{@{}rrrrr@{}}\toprule
$k$ & Actions & Final rounds $t$ & Complete outer phases & $g_k(t)/(c_kt^{-1/k})$\\\midrule
2 & 4 & $5.30226409\times10^{34}$ & $108548648493631900$ & $1+5.71\times10^{-17}$\\
3 & 11 & $1.67272550\times10^{22}$ & $250000$ & $1+5.16385\times10^{-5}$\\
4 & 20 & $6.39859488\times10^{88}$ & $534360276450$ & $1+1.73610\times10^{-11}$\\
5 & 31 & $7.74720259\times10^{412}$ & $4095Q_5$ & $1+3.55341\times10^{-15}$\\
6 & 44 & $6.08763077\times10^{1935}$ & $4095Q_6$ & $1+8.67538\times10^{-19}$\\\bottomrule
\end{tabular}
\caption{Finite certified horizons under different computational budgets and
grouping methods.}
\label{tab:empirical-horizons}
\end{table}

The $k=3$ run uses two million run-or-cycle events, whereas the runs with $k\in\{4,5,6\}$ verify whole nested phase words. Their much larger horizons reflect this different grouping and the different phase lengths and starting labels. The last column compares each gap to its own theoretical
asymptote. The constant $c_k$ is the positive leading constant from
Lemma~\ref{lem:alltime}, computed from the exact count and height coefficients.
It fixes the vertical position of the reference law as well as its exponent;
it is not fitted to the observations. Since normalization divides both the
gap and its constant by $M_k$, the displayed ratio is unchanged.

\subsection{Observed rates for orders two and three}
Figure~\ref{fig:empirical-low-orders} compares each game with its own predicted law. Even after scaling by the full setup-inclusive maximum entry, the constants $c_k/M_k$ differ greatly (Appendix~\ref{app:constant-size}); a smaller exponent need not give a larger gap at a fixed finite time.
\begin{figure}[htbp]\centering
\includegraphics[width=\textwidth]{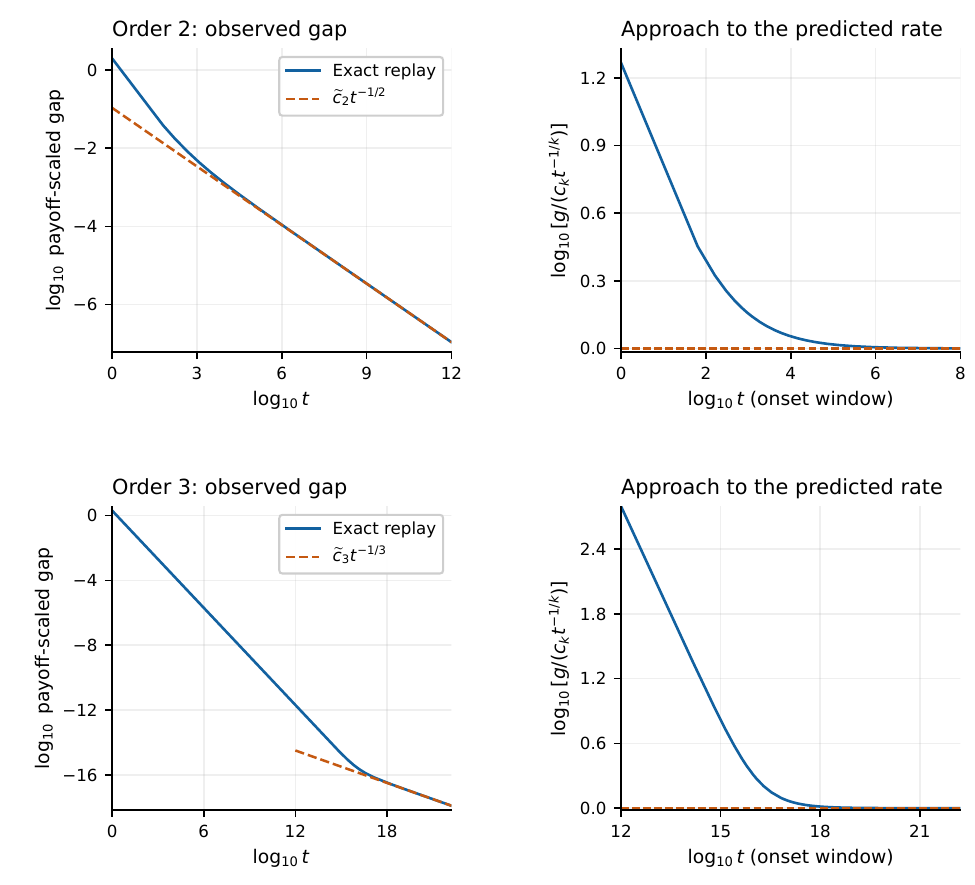}
\caption{Exact trajectory samples for orders two and three. \textbf{Left:} payoff-scaled
empirical gaps and fixed laws $\widetilde c_kt^{-1/k}$. The order-two view stops
at $t=10^{12}$ because its rate is already visible there; the longer certified
endpoint is given in Table~\ref{tab:empirical-horizons}. \textbf{Right:} onset windows
for the logarithm of the gap-to-law ratio, in place of a nearly flat far-tail
view. The zero value means exact agreement with the law. Lines joining saved
samples are visual guides, not additional sampled decisions.}
\label{fig:empirical-low-orders}
\end{figure}
\subsection{Checking complete nested phase words for orders four through six}
The computation for $k>3$ uses the same principle on a larger scale:
propose a finite list of action runs, prove that each proposed decision is a
best response, and only then add its counts to the actual state.
This avoids expanding all the inherited subcycles individually.
To describe the recursion, consider any order-$k$ phase with $k\ge3$ and label $q$, including the inherited order-$3$ phases inside the reported $k=4,5,6$ runs. Its proposed word consists
of its $2k$ reset-action runs, its parent phases with labels
\[
 \sigma=\tau(q+1),\ldots,2\tau(q+1)-1,
\]
and its one-round endpoint action. Here $\tau$ is the fixed integer scale for
this particular lift; each lower lift has its own scale. At the bottom are
the three RPS runs. The reset-count table $M$ and these scales specify the
proposed word. Its count polynomials are recomputed by exact finite summation
of that word, rather than using the supplied full-phase count table $C$ to
advance the trajectory.

From the actual integer score column at the start of a batch, the recomputed
counts give the scores at the first and last decision of every nested run.
The checks require integer lengths at least one, positive bundle counts, and
strict leader-minus-rival margins against every constructed action. The integer
property follows from the integer coefficients in the binomial basis and integer phase
labels; polynomial inequalities verify positivity. Each leader-minus-rival score difference is affine in the within-run offset. Positivity at the first and last decisions therefore implies positivity at every decision of that run, as in the constant-action calculation above.

\paragraph{Covering every nested index.}
If a batch contains $P$ outer phases, then its outer offset satisfies $0\le n\le P-1$. A child-phase offset satisfies $0\le h\le L-1$, where $L=\tau(q+1)$ depends on its ancestors. We introduce one real coordinate in $[0,1]$ for each nesting level. The first is $x_0$: substituting $n=(P-1)x_0$ covers every outer offset. At each child level $j$, substitute $h=(L-1)x_j$, after expressing $L$ in the ancestor coordinates. Thus, $x_j$ specifies the child's relative position in its allowed range. This retains the dependence on ancestor indices and covers a continuous domain containing every required integer index. For notation in the sign test below, relabel these coordinates as $x_1,\ldots,x_\nu$; their domain $[0,1]^\nu$ is the unit cube.

\paragraph{A sufficient exact polynomial sign test.}
After these substitutions the margins and length constraints are rational
polynomials on a unit cube. For coordinate degree bounds $D_1,\ldots,D_\nu$,
write
\[
 p(x)=\sum_{0\le\alpha_j\le D_j} b_\alpha
 \prod_{j=1}^{\nu}\binom{D_j}{\alpha_j}
 x_j^{\alpha_j}(1-x_j)^{D_j-\alpha_j}.
\]
The products are the tensor Bernstein basis: each is nonnegative on the cube
and their sum is one. Therefore $p(x)$ is a weighted average of its
\emph{Bernstein coefficients} $b_\alpha$. All coefficients positive implies
strict positivity; all nonnegative coefficients imply nonnegativity. This computational certificate is used only for the simulations; it is not used in the proof of Theorem~\ref{thm:main}. A negative coefficient only
means that this sufficient test did not certify the proposed batch; it does
not show that the trajectory is invalid. The program advances the state only
after every required test passes. The maximal coordinate degrees are
$(4,3,2)$, $(5,4,3,2)$, and $(6,5,4,3,2)$ for $k=4,5,6$. Figure~\ref{fig:empirical-high-orders} shows the resulting gap samples. Their decay agrees with the predicted $t^{-1/k}$ laws: the gap-to-law relative errors decrease toward zero, reaching the values in Table~\ref{tab:empirical-horizons}.

\begin{figure}[htbp]\centering
\includegraphics[width=\textwidth]{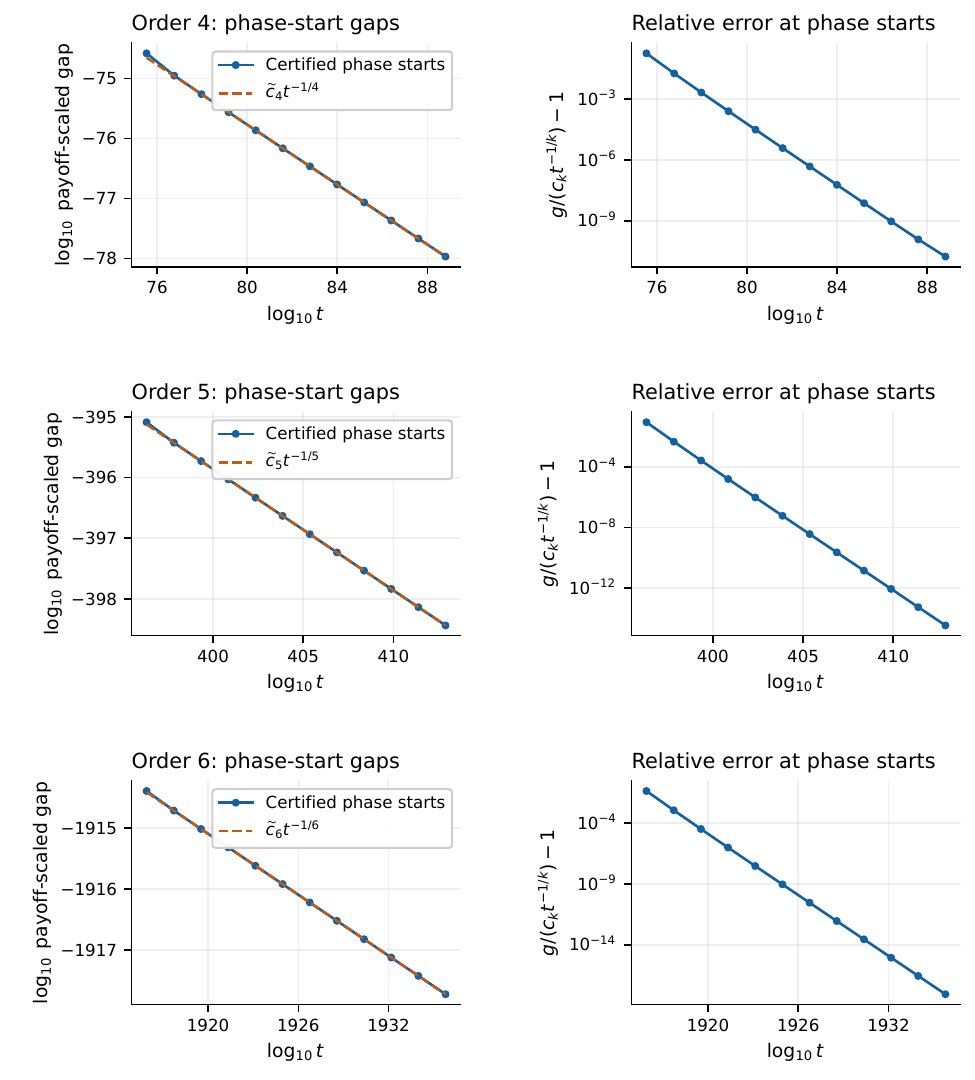}
\caption{Gap comparisons for $k=4,5,6$, using samples at the starts of phases $q=2^jQ_k$, $j=1,\ldots,12$. \textbf{Left:} payoff-scaled gaps and the fixed laws $\widetilde c_kt^{-1/k}$. \textbf{Right:} relative errors on logarithmic vertical axes, showing their approach to zero. Both panels in each row use the same horizontal coordinate $\log_{10}(t)$ and the same time range. These runs save samples at the starts of the indicated outer phases, rather than at the finer run/event intervals of Figure~\ref{fig:empirical-low-orders}. Thus, they do not display the shape of the gap inside a phase. Exact certificates cover all intervening decisions. Lines between samples are visual guides.}
\label{fig:empirical-high-orders}
\end{figure}
\FloatBarrier

\section{AI methodology statement}\label{sec:ai-methodology}

Our results were developed through an iterative process of mathematical analysis and interaction with ChatGPT 5.6 Sol, ChatGPT 6 Astra, and an early internal version of Gemini 4 Pro. 
Beginning in summer 2026, this process first yielded a zero-sum game with duality-gap decay $\Theta(t^{-1/4})$, slower than the $\Theta(t^{-1/3})$ construction of~\cite{wang2025tie}. 
Further interactions with ChatGPT helped develop the inductive construction of a general family of hard instances. 
Subsequent interactions with Gemini 4 Pro helped simplify the construction and proof. 
The authors verified and simplified the construction and proofs. 
ChatGPT also assisted with the writing and presentation. 
The authors take full responsibility for the correctness of the results and their exposition.

While preparing this manuscript, 
we learned that the Cogentic team at Google had also obtained a similar result. Their result was obtained using the Cogentic harness~\citep*{cai2026cogentic}
with an internal Gemini model as the base model.

\paragraph{Acknowledgements.}
The authors thank Yuanhao Wang, Georgios Piliouras, Ryann Sim, Stratis Skoulakis, and Yang Cai for helpful discussions.

\bibliographystyle{plainnat}
\bibliography{references}
\end{document}